\documentclass[a4paper]{amsart}
\usepackage{amsmath,amsthm}
\usepackage{amssymb,mathrsfs,amscd}
\usepackage{braket}
\usepackage{enumerate}
\usepackage{nccbbb}
\usepackage{mathtools}
\usepackage{graphicx}
\usepackage{amsmath}
\usepackage{amssymb}
\usepackage{mathrsfs}
\usepackage{amscd}
\usepackage{braket}
\usepackage{enumerate}
\usepackage{nccbbb}
\usepackage{mathtools}
\usepackage{multicol}
\usepackage{multirow}
\usepackage{amsthm}
\theoremstyle{plain}
\newtheorem{theorem}{Theorem}[section]
\newtheorem{corollary}[theorem]{Corollary}
\newtheorem{lemma}[theorem]{Lemma}
\newtheorem{proposition}[theorem]{Proposition}

\theoremstyle{definition}
\newtheorem{remark}[theorem]{Remark}
\newtheorem{definition}[theorem]{Definition}

\newcommand{\BK}[1]{ \left( #1 \right) }
\newcommand{\BKK}[1]{ \left( \textstyle #1 \right) }
\newcommand{\sqBK}[1]{ \left[ #1 \right] }
\newcommand{\curBK}[1]{ \left\{ #1 \right\} }
\newcommand{\absBK}[1]{ \left| #1 \right| }
\newcommand{\VertBK}[1]{ \left \Vert#1 \right \Vert}

\newcommand{\intLim}{\int\limits}

\newcommand{\CosS}[2]{ \frac{ 2\cos \phi{ #1 } }{ \sigma_{#2} } }

\newcommand{\Maxll}[2]{ e^{ -\BK{ \frac{ 2\cos\phi{#1} }{ \sigma_{#2} } }^2 } }
\newcommand{\Tm}{ T_* }
\newcommand{\TM}{ T^* }
\newcommand{\Tr}{ \sqrt{\frac{\Tm}{\TM}} }

\newcommand{\bfx}{ \textbf{x} }
\newcommand{\bfy}{ \textbf{y} }
\newcommand{\bfxi}{ {\boldsymbol \xi} }
\newcommand{\bfeta}{ {\boldsymbol \eta} }
\newcommand{\bfzeta}{ {\boldsymbol \zeta} }

\newcommand{\bfn}{ \textbf{n} }

\newcommand{\absbfxi}{ \absBK{\boldsymbol\xi} }
\newcommand{\bbR}{ \mathbb{R} }
\newcommand{\calJ}{ {\mathcal J} }
\newcommand{\calN}{ {\mathcal N} }
\newcommand{\calF}{ {\mathcal F} }

\newcommand{\SFr}[1]{ {\bf S}^{\bf Fr}_{#1} }
\newcommand{\SDF}[1]{ {\bf S}^{\bf DFr}_{#1} }
\newcommand{\SLB}[1]{ {\bf S}^{\bf LB}_{#1} }

\DeclareMathOperator{\E}{E}
\DeclareMathOperator{\Prob}{P}
\DeclareMathOperator{\diam}{diam}
\DeclareMathOperator{\esssup}{ess\,sup}

\begin{document}

\title{Equilibrating effect of Maxwell-type boundary condition in highly rarefied gas%\thanks{Grants or other notes
}


\author{Hung-Wen Kuo %etc.
}


\address{Department of Mathematics, National Cheng Kung University, Tainan 70101, Taiwan}
\email{hwkuo@mail.ncku.edu.tw} 
%\subjclass{Primary: 76P05; Secondary: 45K05, 82C40.}            

\maketitle

\begin{abstract}
We study the equilibrating effects of the boundary and intermolecular collision in the kinetic theory for rarefied gases. We consider the Maxwell-type boundary condition, which has weaker equilibrating effect than the commonly studied diffuse reflection boundary condition. The gas region is the spherical domain in $\bbR^d$, $d=1,2.$ First, without the equilibrating effect of the collision, we obtain the algebraic convergence rates to the steady state of free molecular flow with variable boundary temperature. The convergence behavior has intricate dependence on the accommodation coefficient of the Maxwell-type boundary condition. Then we couple the boundary effect with the
intermolecular collision and study their interaction.
We are able to construct the steady state solutions of the full Boltzmann equation for large Knudsen numbers and small boundary temperature variation. We also establish the nonlinear stability  with exponential rate of the stationary Boltzmann solutions. Our analysis is based on the explicit formulations of the boundary condition for symmetric domains.

\end{abstract}

\begin{section}{Introduction}

In kinetic theory, a fundamental and central issue is the equilibrating effects of the boundary
and intermolecular collision. 
In the present study we consider the Maxwell-type boundary condition:
\begin{equation}\label{eq::Maxwell-type::BC}
\begin{dcases}
     F(\bfy,\bfzeta,t)
&=
    \alpha(\bfy)\BK{ \frac{2\pi}{RT(\bfy)} }^{\frac{1}{2}} j_F(\bfy,t) M_{T(\bfy)}(\bfzeta) \\&+ (1-\alpha(\bfy))F(\bfy,\bfzeta-2(\bfzeta\cdot\bfn)\bfn,t),
      \quad \hfill \bfy\in \partial D,  \bfzeta\cdot\bfn>0 ,\\   
       j_F(\bfy,t)
&=
      \intLim_{ \bfzeta_*\cdot\bfn<0 } -\bfzeta_*\cdot\bfn F(\bfy,\bfzeta_*,t) d\bfzeta_*
      : \text{ boundary flux of } F,
\end{dcases}
\end{equation}
where $F$ is the velocity distribution function of the gas particles, $\bfzeta$ is the microscopic velocity, $T(\bfy)$ is the boundary temperature at the boundary point $\bfy$, $\bfn$ is the unit normal vector at the boundary, pointing to the gas region $D$, $\alpha(\bfy)$ $(0\leq\alpha(\bfy)\leq 1)$ is the accommodation coefficient, and $M_T$ is the Maxwell distribution:
\begin{equation*}
     M_T (\bfzeta)
=
    \frac{ e^{ -\frac{\absBK{\bfzeta}^2}{2RT} } }{ \BK{2\pi RT}^{\frac{3}{2}} },
    \ R: \ \text{Boltzmann constant}.
\end{equation*}
The case $\alpha=0$ is called the specular reflection boundary condition, which has no equilibrating effect. The case $\alpha=1$ is called the diffuse refection boundary condition, which has strong, direct equilibrating effect of the boundary thermal information on the gas flows. In this paper we assume that the accommodation coefficient $\alpha,\ 0<\alpha<1,$ is constant. Our purpose is to study the equilibrating effect of the Maxwell-type boundary condition, the dependence of the process of convergence to steady states on the accommodation coefficient $\alpha,\ 0<\alpha<1.$ 
The equilibrating process also depends on the geometry of the boundary. 
Our analysis demands the quantitative method for the study of particle propagation. For this, we will focus on
spherical symmetric domains:   
\begin{equation*}\label{eq::Defn:Domain}
    D = \curBK{ \bfx\in\bbR^d : |\bfx|<1 },
\end{equation*}
in space dimension $d=1,2$. This allows us to use the stochastic formulation of our previous works \cite{Kuo-Liu-Tsai} and \cite{Kuo-Liu-Tsai-2}, which provides an explicit description of the evolution of the free molecular flow.

We decompose the microscopic velocity $\bfzeta=(\zeta_1,\zeta_2,\zeta_3)\in\bbR^3$ into 
\begin{equation}\label{eq::Def:xi:eta}
	\bfxi = (\zeta_1,\ldots, \zeta_d) \in \bbR^d,
	\quad
	\bfeta = (\zeta_{d+1},\ldots,\zeta_3) \in \bbR^{3-d},
\end{equation}
and rewrite the Maxwell-type boundary condition \eqref{eq::Maxwell-type::BC} as
\begin{equation}\label{eq::Diff:Rel:BC}
\begin{dcases}
     F(\bfy,\bfzeta,t)
&=
    \alpha\BK{ \frac{2\pi}{RT(\bfy)} }^{\frac{1}{2}} j_F(\bfy,t) M_{T(\bfy)}(\bfzeta) \\&+ (1-\alpha)F(\bfy,\bfxi-2(\bfxi\cdot\bfn)\bfn,\bfeta,t),
      \quad \hfill \bfy\in \partial D,  \bfxi\cdot\bfn>0 ,\\   
       j_F(\bfy,t)
&=
      \intLim_{ \bfxi_*\cdot\bfn<0 } -\bfxi_*\cdot\bfn F(\bfy,\bfzeta_*,t) d\bfzeta_*
      : \text{ boundary flux of } F,
\end{dcases}
\end{equation}

To focus on the equilibrating effect of boundary, we first consider the free molecular flow:
\begin{align}\label{eq:Fr:Eq:Eq}
\begin{dcases}
	\frac{ \partial  g }{ \partial t }
	+
	\sum_{i=1}^d \zeta_i \frac{ \partial  g }{ \partial x_i } = 0,
    \quad
     g =  g(\bfx,\bfzeta,t), \ \bfx\in D, \ \bfzeta\in\bbR^3, \ t>0,
\\
     g(\bfx,\bfzeta,0) =  g_{in}(\bfx,\bfzeta), \quad \bfx\in D, \ \bfzeta\in\bbR^3,
\\
 g(\bfy,\bfzeta,t)
=
    \alpha\BK{ \frac{2\pi}{RT(\bfy)} }^{\frac{1}{2}} j_g(\bfy,t) M_{T(\bfy)}(\bfzeta) \\
    \quad\quad\quad\quad+ (1-\alpha)g(\bfy,\bfxi-2(\bfxi\cdot\bfn)\bfn,\bfeta,t),
      \quad \hfill \bfy\in \partial D,  \bfxi\cdot\bfn>0 ,\\   
       j_g(\bfy,t)
=
      \intLim_{ \bfxi_*\cdot\bfn<0 } -\bfxi_*\cdot\bfn g(\bfy,\bfzeta_*,t) d\bfzeta_*
      : \text{ boundary flux of } g.
\end{dcases}
\end{align}
The equation for the steady state of the free molecular flow is:
\begin{equation}\label{eq:Fr:Def:S}
\begin{dcases}
    \sum_{i=1}^d \zeta_i \frac{\partial  S}{\partial x_i} = 0,
   \quad
     S =  S(\bfx,\bfzeta), \ \bfx\in D, \ \bfzeta\in\bbR^3,
\\
    \frac{1}{|D|} \int_{D\times\bbR^3}  S(\bfx,\bfzeta) d\bfx d\bfzeta = 1:
	\text{ unit density,}
\\
       S(\bfy,\bfzeta)
    =
     \alpha\BK{ \frac{2\pi}{RT(\bfy)} }^\frac12
      j_S(\bfy,t)
      M_{T(\bfy)}(\bfzeta) \\
      \quad\quad\quad\quad+ (1-\alpha) S(\bfy,\bfxi-2(\bfxi\cdot\bfn)\bfn,\bfeta), 
       \quad \bfy\in\partial D, \ \bfxi\cdot\bfn > 0,
\\
	 \quad\quad
      j_S(\bfy)
    =
     \int_{\bfxi_*\cdot\bfn<0} -\bfxi_*\cdot\bfn S(\bfy,\bfzeta_*) d\bfzeta_*:
	 \text{ boundary flux of } S.
\end{dcases}
\end{equation}
Here we take general initial data $g_{in}$ with finite weighted $L^\infty$ norm:
\begin{equation}\label{eq::Choice:of:nu}
\begin{split}
&
	g_{in} (\bfx,\bfzeta) \in L^{\infty,\mu}_{\bfx,\bfzeta}, \ \mu > 4,
	\\
&
	\VertBK{ g_{in} }_{ L^{\infty,\mu}_{\bfx,\bfzeta} } = \VertBK{ g_{in} }_{\infty,\mu}
	\equiv \mathop{\esssup}_{\bfx\in D, \bfzeta \in \bbR^3} (1+|\bfzeta|)^\mu |g_{in} (\bfx,\bfzeta)|,
\end{split}
\end{equation}
where the choice of $\mu>4$ implies that
\begin{equation*}
    \int_{\bfxi\cdot\bfn<0} \absBK{ \frac{-\bfxi\cdot\bfn }{ (1+|\bfzeta|)^\mu } }  d\bfzeta < \infty.
\end{equation*}
The boundary temperature variation is assumed to be bounded:
\begin{equation*}
     0 < \Tm \equiv \inf_{\partial D} T(\bfy) \leq \TM \equiv \sup_{\partial D} T(\bfy) < \infty.
\end{equation*}

The diffuse reflection boundary condition has strong and direct equilibrating effect, and as a consequence, the convergence to steady state is of the rate of $t^{-d}$, $d$ the space dimension,  \cite{Kuo-Liu-Tsai,Tsuji}. The Maxwell-types boundary condition yields eventually the same rate, with intricate dependence on the accommodation coefficient $\alpha$.

The following Theorem shows
the convergence to the steady state of free molecular flow $S$.
\begin{theorem}[Main Theorem for Free Molecular Flow]
\label{thm:Fr:Soln:Op:Ptws:Esti:new}
For $g_{in}\in L^{\infty,\mu}_{\bfx,\bfzeta}$ and $\mu>4$,  the solution of \eqref{eq:Fr:Eq:Eq} satisfies 
\begin{multline*}
     g(\bfx,\bfzeta,t) -\rho_*S(\bfx,\bfzeta)
=
     O(1)
    \VertBK{ g_{in} }_{\infty,\mu}
    \left\{
            \BK{\frac{ M(\bfzeta) }{ (1+\alpha t)^d }+\frac{(1-\alpha)^{\frac{t^{\epsilon}}{2}}}{(1+|\bfzeta|)^\mu}}
            \bbbone_{ \curBK{ |\bfxi|>\frac{2}{t^{1-\epsilon}} } }\right.\\
            \left.+
            \frac{1}{(1+|\bfzeta|)^\mu}\bbbone_{ \curBK{ |\bfxi|<\frac{2}{t^{1-\epsilon}} } }
    \right\},
\end{multline*}
\begin{equation*}
    \rho_*
\equiv
    \frac{1}{|D|} \int_{D\times\bbR^3} g_{in}(\bfx,\bfzeta) d\bfx d\bfzeta,
\end{equation*}
for any small
$\epsilon,\ 0<\epsilon\leq\frac{1}{400}$. 
\end{theorem}
Theorem \ref{thm:Fr:Soln:Op:Ptws:Esti:new}. immediately implies the following $L^p$ convergence of $g$:
\begin{corollary}\label{cor:Fr:Main:Lp}
For any small
$\epsilon,\ 0<\epsilon\leq\frac{1}{400},$ $g$ converges to $\rho_*S$ in $L^p_{\bfx,\bfzeta}$ for $1\leq p <\infty$:
\begin{equation*}
    \VertBK{ g(\bfx,\bfzeta,t) - \rho_* S(\bfx,\bfzeta) }_{ L^p_{\bfx,\bfzeta} }
=
     O(1) \VertBK{g_{in}}_{\infty,\mu}\Big(\frac{1}{(\alpha t+1)^{d}} +(1-\alpha)^{t^\epsilon/2}+ \frac{1}{(t+1)^{(1-\epsilon)\frac{d}{p}}}\Big).
\end{equation*}
Consequently, there exists $C_{\alpha,\epsilon}>0$ such that
\begin{equation*}
    \VertBK{ g(\bfx,\bfzeta,t) - \rho_* S(\bfx,\bfzeta) }_{ L^p_{\bfx,\bfzeta} }
\leq
     C_{\alpha,\epsilon} \VertBK{g_{in}}_{\infty,\mu} \BK{\frac{1}{(t+1)^{(1-\epsilon)\frac{d}{p}}}}.
\end{equation*}
In particular, the coefficient $C_{\alpha,\epsilon}=O(1)$ when $\alpha=1$.
Hence, we may let $\epsilon\rightarrow 0$ to obtain the optimal rate for diffuse reflection boundary condition:
\begin{equation*}
    \VertBK{ g(\bfx,\bfzeta,t) - \rho_* S(\bfx,\bfzeta) }_{ L^p_{\bfx,\bfzeta} }
=
     O(1) \VertBK{g_{in}}_{\infty,\mu} \BK{\frac{1}{(t+1)^{\frac{d}{p}}}}.
\end{equation*}
\end{corollary}

After studying the boundary effect of Maxwell-type condition for free molecular flow, 
we continue study the additional equilibrating effect of the collision in rarefied gas flow. We use the Boltzmann equation to model gas with intermolecular collision. Consider the initial-boundary value problem of the Boltzmann equation:
\begin{equation}\label{eq:FullBz:Eq}
\begin{dcases}
                \frac{\partial F}{\partial t} + \sum_{i=1}^d \zeta_i \frac{ \partial F}{ \partial x_i }
                 =
                \frac{1}{\kappa} Q(F,F),
                \quad \bfx \in D \subset \bbR^d,
                \ \bfzeta \in \bbR^3, \ t>0,\\
                 F(\bfx,\bfzeta,0)= F_{in}(\bfx,\bfzeta),\quad \bfx \in D \subset \bbR^d,
                \ \bfzeta \in \bbR^3,\\
               \text{ Maxwell-type boundary condition } \eqref{eq::Diff:Rel:BC}.
\end{dcases}
\end{equation}
where $\kappa$ is the Knudsen number which measures how rarefied the gas is, and $Q(\cdot,\cdot)$ is the collision operator, a symmetric bilinear operator. 
\begin{multline*}
     Q(g,h)(\bfzeta)
	=
    \frac12 \intLim_{S^2\times\bbR^3}
    \Big(
                g(\bfzeta') h(\bfzeta'_*)   + h(\bfzeta') g(\bfzeta'_*)
            -   g(\bfzeta) h(\bfzeta_*)     - h(\bfzeta) g(\bfzeta_*)
    \Big)
\\
    \times
     B(\theta,|\bfzeta_*-\bfzeta|) d\Omega d\bfzeta_*,
\end{multline*}
where
\begin{equation*}
    \begin{dcases}
    \bfzeta'
    =
    \bfzeta - \Big( (\bfzeta-\bfzeta_*)\cdot\Omega \Big) \Omega
    \\
    \bfzeta'_*
    =
    \bfzeta_* + \Big( (\bfzeta-\bfzeta_*)\cdot\Omega \Big) \Omega
    \end{dcases}
\quad
    \cos\theta
    =
    \frac{\bfzeta-\bfzeta_*}{|\bfzeta-\bfzeta_*|} \cdot \Omega,
\end{equation*}
and $B$ is the collision kernel which is determined by the interaction potential between two colliding particles. 
Throughout this paper we assume an inverse power hard potential with Grad's angular cut-off or hard sphere. Under this model, $ B(\theta,|\bfzeta_*-\bfzeta|)\sim |\bfzeta-\bfzeta_*|^{\frac{u-4}{u}}|\cos\theta|$, for some $u\geq 4$. \\
\begin{center}
\begin{tabular}{|l|l|l|}
\hline
Maxwell molecule& hard potential & hard sphere\\
\hline
$u=4$ & $4<u<\infty$ & $u=\infty$\\
\hline
\end{tabular}\\
\end{center}
%More explicitly,
%\begin{equation*}
%    B_* |\bfzeta-\bfzeta_*|^{\frac{u-4}{u}} |\cos\theta|
%\leq
%    B(\theta,|\bfzeta-\bfzeta_*|)
%\leq
%    B^*|\bfzeta-\bfzeta_*|^{\frac{u-4}{u}} |\cos\theta|,
%\end{equation*}
%for some constants $0<B_*\leq B^*<\infty$.
By nondimensionalization, \cite{Sone}, we may assume, without loss of generality, that  $0 < \Tm < \TM  = 1$ and the total density unity:
\begin{equation}\label{eq:FullBz:density:is:one}
    \frac{1}{|D|} \int_{D\times\bbR^3} F_{in}(\bfx,\bfzeta) d\bfx d\bfzeta = 1.
\end{equation}
For convenience, denote the Maxwellian $M_{\TM}(\bfzeta)=M_{1}(\bfzeta)=(\pi)^{-\frac32}\exp(-|\bfzeta|^2)$ simply by $M(\bfzeta)$.

Conventionally, to linearize the Boltzmann equation, we expand $F$ around $M$, $F=M+\sqrt M f$. The resulting equation for the perturbation $f$ is
\begin{equation*}
     \frac{ \partial f }{ \partial t }
    +\sum_{i=1}^d \zeta_i \frac{ \partial f }{ \partial x_i }
    -\frac{1}{\kappa} Lf
=
     \frac{1}{\kappa\sqrt M} Q\BK{ \sqrt M f, \sqrt M f },
\end{equation*}
where the linearized collision operator $L$ is defined as
\begin{equation*}
     Lf = \frac{2}{\sqrt M} Q\BK{ \sqrt M f, M },
\end{equation*}
and the linearized Boltzmann equation is
\begin{equation*}
         \frac{ \partial f }{ \partial t }
    +\sum_{i=1}^d \zeta_i \frac{ \partial f }{ \partial x_i }
    -\frac{1}{\kappa} Lf
=
      0.
\end{equation*}
For the intermolecular force model we consider, inverse power hard potential with Grad's angular cut-off or hard spheres, $L$ can be decomposed as the difference of an integral operator $K$ and a multiplicative operator $\nu$:
\begin{equation*}
     L = K - \nu, \quad
    \Big(Kf\Big)(\bfzeta) = \int_{\bbR^3} K(\bfzeta,\bfzeta_*) f(\bfzeta_*) d\bfzeta_*, \quad
    \Big(\nu f\Big)(\bfzeta) = \nu(\bfzeta)f(\bfzeta).
\end{equation*}
$\nu(\bfzeta)$ is a positive function with non-zero infimum:
\begin{equation*}%\label{eq:LB:Defn:nu:zero}
    \nu_0 \equiv \inf_{\bfzeta\in\bbR^3} \nu(\bfzeta) > 0.
\end{equation*}
However, as M does not satisfy the boundary condition \eqref{eq::Diff:Rel:BC}, this linearization is not natural for the present situation, where the boundary effect is significant. Instead, we expand around the stationary free molecular flow S under the Maxwell-type boundary condition and write $F=S+\sqrt M f$.
Because $S$ does not satisfy the Boltzmann equation \eqref{eq:FullBz:Eq}, some more source terms are introduced to the equation for $f$ :
\begin{equation}\label{eq:FullBz:Eq:expand}
\begin{split}
&
     \frac{ \partial f }{ \partial t }
    +\sum_{i=1}^d \zeta_i \frac{ \partial f }{ \partial x_i }
    -\frac{1}{\kappa} Lf
     \\
=&
     \frac{1}{\kappa} L\BKK{ \frac{S-M}{\sqrt M} }
    +\frac{1}{\kappa\sqrt M} Q\BK{ S-M+\sqrt Mf, S-M+\sqrt M f }.
\end{split}
\end{equation}
On the other hand, we have the important property that $S$ satisfies the Maxwell-type boundary condition. As the boundary condition is linear and homogeneous, the boundary condition for $\sqrt M f$ remains the same:
\begin{multline}\label{eq::Diff:Ref:BC:expand}
     f(\bfy,\bfzeta,t) \sqrt{ M(\bfzeta) }
     \\
=
     \alpha\BK{ \frac{2\pi}{RT(\bfy)} }^{\frac{1}{2}}
     \BK{ \int_{\bfxi_*\cdot\bfn<0} -\bfxi_*\cdot\bfn f(\bfy,\bfzeta_*,t) \sqrt{ M(\bfzeta_*) } d\bfzeta_* }
      M_{T(\bfy)}(\bfzeta)
      \\
      +(1-\alpha)f(\bfy,\bfxi-2(\bfxi\cdot\bfn)\bfn,\bfeta,t) \sqrt{ M(\bfzeta) }
     \\
     \bfy\in\partial D, \ \bfxi\cdot\bfn > 0.
\end{multline}
Therefore, we will consider the initial-boundary value problem of equation \eqref{eq:FullBz:Eq:expand}
with boundary condition \eqref{eq::Diff:Ref:BC:expand} and initial data: 
\begin{equation*}
f_{in}(\bfx,\bfzeta)=\frac{F_{in}(\bfx,\bfzeta)-S(\bfx,\bfzeta)}{\sqrt{M(\bfzeta)}}.
\end{equation*}
There is a trade-off between more complicated boundary condition and extract interior source terms. We choose the latter since the effect of Maxwell-type boundary condition has been well analyzed for the free molecular flow. Moreover, this linearization is physically natural for large Knudsen number. The extra interior source terms can be handled by the standard iteration scheme. Therefore, we will first consider \eqref{eq:LB:Eq} as our linearized equation:
\begin{equation}\label{eq:LB:Eq}
\begin{dcases}
         \frac{ \partial f }{ \partial t }
        +\sum_{i=1}^d \zeta_i \frac{ \partial f }{ \partial x_i }
        -\frac{1}{\kappa} Lf
        = 0,
\\
         \text{ Maxwell-type boundary condition } \eqref{eq::Diff:Ref:BC:expand},
\\
         \int_{D\times\bbR^3} \sqrt{M(\bfzeta)} f_{in}(\bfx,\bfzeta) d\bfx d\bfzeta = 0.
\end{dcases}
\end{equation}
From $F=S+\sqrt M f$, the zero total mass condition $\int \sqrt Mf_{in} d\bfx d\bfzeta=0$ is a consequence of nondimensionalization, \eqref{eq:FullBz:density:is:one}, rather than an additional constraint. 
To study the problem \eqref{eq:LB:Eq}, we make two reductions: first we reduce $(\partial_t+\sum\zeta_i\partial_{x_i}-\frac{1}{\kappa}L)$ to $(\partial_t+\sum\zeta_i\partial_{x_i}+\frac{\nu}{\kappa})$, and then reduce $(\partial_t+\sum\zeta_i\partial_{x_i}+\frac{\nu}{\kappa})$ to $(\partial_t+\sum\zeta_i\partial_{x_i})$.

For the linear problem \eqref{eq:LB:Eq}, we prove an exponential decay to zero of $f$, Theorem \ref{thm:LB:Main}. We then use this result to: (i) construct the steady state solution of the Boltzmann equation \eqref{eq:FullBz:Eq:expand}, Theorem \ref{thm:FullBz:Main:1}; (ii) obtain an exponential convergence to the steady state, Theorem \ref{thm:FullBz:Main:2}.

\begin{theorem}[Stability for Linear Boltzmann Equation]\label{thm:LB:Main}
Suppose that $f_{in}\in L^{\infty,-\gamma}_{\bfx,\bfzeta}$, for some constant $\gamma$, $0\leq \gamma \leq 1$, $\int f_{in} \sqrt M d\bfx d\bfzeta=0$, and that, for each $0<\nu'_1<\nu_0$, there exists a positive constant $C_1$ such that for all Knudsen number $\kappa$ with
\begin{equation}\label{eq:LB:requirement:of:k:derived}
\begin{dcases}
      \kappa   \geq C_1 \BK{\frac{1}{\alpha^3(\nu_0-\nu_1')^2}
     +\frac{(2|\log[\alpha(\nu_0-\nu_1')]|)^{400}}{\alpha(\nu_0-\nu_1')|\log(1-\alpha)|^{400}}  }                              &   \text{ for } d=1,   \\
   \frac{ \kappa }{ \log\kappa }  \geq C_1 \BK{\frac{1}{\alpha^{5/2}(\nu_0-\nu_1')^{3/2}}
     +\frac{(2|\log[\alpha(\nu_0-\nu_1')]|)^{400}}{\alpha(\nu_0-\nu_1')|\log(1-\alpha)|^{400}}  }                                          &   \text{ for } d=2,   
\end{dcases}
\end{equation}
the solution $f$ of \eqref{eq:LB:Eq} decays to zero exponentially in time:
\begin{equation}\label{eq:LB:Exp:Decay}
    \VertBK{f(\cdot,\cdot,t)}_{\infty,-\gamma}
\leq
      C\VertBK{f_{in}}_{\infty,-\gamma} e^{ -\frac{\nu'_1}{\kappa} t },
\end{equation}
for some constant $C$ independent of $\kappa, \alpha, \Tm$ and $\nu'_1$.
\end{theorem}
\begin{remark}
Our main interest is in the highly rarefied gas, i.e. the case of large Knudsen number $\kappa$. Thus  we first study the free molecular flow for collisionless gas. Then we consider a perturbation around the steady solution of free molecular flow for the Boltzmann equation. To obtain the exponential stability for linearized Boltzmann equation, we start with the estimate of free molecular flow which is a limiting case of $\kappa=\infty$. Recall that the pointwise estimates of free molecular flow depend on the accommodation coefficient $\alpha$, Theorem \ref{thm:Fr:Soln:Op:Ptws:Esti:new}. Consequently the magnitude of $\kappa$ is related to that of $\alpha$, \eqref{eq:LB:requirement:of:k:derived}.
\end{remark}
In this paper, we consider the situation of variable boundary temperature. It is a highly non-trivial problem to study the existence of steady solution for the Boltzmann equation when the boundary temperature varies. In the case of diffuse reflection with variable temperature, the existence of the steady solution in a convex domain with dimension less or equal than three was proved by Guiraud \cite{Guiraud,Guiraud-2} for arbitrary fixed Knudsen number. The same result was proved for arbitrary domain together with the exponential stability \cite{Guo-Kim}. Moreover, a large data existence result was proved by Arkeryd and Nouri for prescribed initial data \cite{Arkeryd-2}.
In our work, with the exponential convergence of linearized Boltzmann equation, Theorem \ref{thm:LB:Main}, we are able to construct the steady solution as a consequence of time asymptotic analysis. To handle the nonlinear term, we require the small variation of boundary temperature $1-\Tm \ll 1$.
Then we prove the existence of the steady solution for the Boltzmann equation \eqref{eq:FullBz:Eq}.
However, the aim of this paper is to establish the equilibrating effect of boundary and collision. The method we used in this paper can be seen as an alternative approach to the existence problem of steady solution. The quantitative structure of the steady solution on the temperature variation is remained for our future research work.
\begin{theorem}[Existence of Steady Solution for Full Boltzmann Equation]\label{thm:FullBz:Main:1}
Assume that $1-\Tm \ll 1$ and $\kappa\gg 1$ satisfies \eqref{eq:LB:requirement:of:k:derived}. Then the steady state solution $\Phi$ of \eqref{eq:FullBz:Eq} exists and satisfies
\begin{align}
    \label{eq:FullBz:Phi:Esti}
&
    \VertBK{ \Phi }_\infty = O(1-\Tm),
    \\
    \label{eq:FullBz:Phi:Zero:Total:Mass}
&
    \int_{D\times\bbR^3} \Phi(\bfx,\bfzeta) \sqrt{ M(\bfzeta) } d\bfx d\bfzeta = 0.
\end{align}
\end{theorem}
We have already obtained the steady state solution $F_\infty \equiv S+\sqrt M\Phi$ for full Boltzmann equation \eqref{eq:FullBz:Eq}. Moreover, from \eqref{eq:FullBz:Phi:Zero:Total:Mass}, $\int F_\infty d\bfx d\bfzeta = 1$. For general initial boundary value problem, we expand $F$ around $F_\infty$: $F=F_\infty+\sqrt M \psi$. The equation for $\psi$ is
\begin{equation}\label{eq:FullBz:Eq:psi}
\begin{dcases}
    \frac{\partial\psi}{\partial t} 
    + 
    \sum_{i=1}^d \zeta_i \frac{\partial\psi}{\partial x_i} - \frac{1}{\kappa} L \psi
    \\
	\quad
= 
    \frac{2}{\kappa\sqrt M} Q\BK{ \sqrt M \psi, F_\infty - M } + \frac{1}{\kappa\sqrt M} Q(\psi\sqrt M,\psi\sqrt M),
    \\
    \psi(\bfx,\bfzeta,0) = \psi_{in}(\bfx,\bfzeta)= \frac{F_{in}-F_\infty}{\sqrt{M}} \in L^\infty_{\bfx,\bfzeta},
    \\
    \text{ Maxwell-type boundary condition } \eqref{eq::Diff:Ref:BC:expand}
    \\
    \int_{D\times\bbR^3} \psi_{in}(\bfx,\bfzeta) \sqrt{ M(\bfzeta) } d\bfx d\bfzeta = 0.
\end{dcases}
\end{equation}
To reiterate, zero initial total molecular number $\int \psi_{in}\sqrt M d\bfx d\bfzeta=0$ is a consequence of nondimensionalization \eqref{eq:FullBz:density:is:one}, rather than an additional constraint.

To show the stability of steady solution $F_\infty$, it is sufficient to show that $\psi$ decays to zero. We establish the exponential decay rate by using Picard iteration with the estimate \eqref{eq:LB:Exp:Decay} of linearized Boltzmann equation. The following theorem shows that the steady solution of the Boltzmann equation is exponential stable when the Knudsen number $\kappa$ is sufficiently large, \eqref{eq:LB:requirement:of:k:derived}.
\begin{theorem}[Stability of Steady Solution for Full Boltzmann Equation]\label{thm:FullBz:Main:2}
Suppose that \eqref{eq:LB:requirement:of:k:derived} and $\frac{ (1-\Tm) + \VertBK{\psi_{in}}_\infty }{ (\nu_0-\nu_1)^2 } \ll 1$. Then, for any fixed $0<\nu_1<\nu_0$, the solution $\psi$ of \eqref{eq:FullBz:Eq:psi} exists and satisfies
\begin{equation*}
    \VertBK{ \psi(t) }_\infty \leq C \VertBK{ \psi_{in} }_\infty e^{ -\frac{\nu_1}{\kappa}t },
\end{equation*}
for a positive constant independent $C$ of $\alpha,\kappa, \Tm$ and $\nu_1$.
\end{theorem}

The present study focuses first on the equilibrating effect of the boundary condition by considering the free molecular flows. This approach has been considered for the diffuse reflection boundary condition, initiated by \cite{Yu} for one space dimension and then generalized to higher space dimensions by \cite{Kuo-Liu-Tsai} for constant boundary temperature, and by  \cite{Kuo-Liu-Tsai-2} for variable boundary temperature. It also has been studied for half space under gravitational force \cite{Liu-Yu-Gravitation}. This approach makes it possible to study of the full Boltzmann equation when the Knudsen number is large. Other studies consider the case when the Knudsen number is of order one. In other words, they consider the case when the collision plays a role at least as important as the boundary condition in the equilibration of the gases. Also, other studies consider the diffuse reflection boundary condition. 
For this, see \cite{Arkeryd,Guo,Villani} when the boundary temperature is constant, and \cite{Arkeryd,Arkeryd-2,Guiraud,Guiraud-2,Guo-Kim} for variable boundary temperature.
For the specular reflection condition, 
the equilibrating effect has to come from the collision. There have been substantial progresses in this regard on the level of the Boltzmann equation, see \cite{Desvillettes-Villani,Guo,Villani}, and references therein.

\end{section}

\begin{section}{Free Molecular Flow }
\begin{subsection}{Preliminaries and main results}\label{sec:Fr:Prel}

The steady state solution S of free molecular flow under the Maxwell-
type boundary condition \eqref{eq:Fr:Def:S} has been constructed explicitly, \cite{Sone}:
\begin{equation}\label{eq:Fr:Express:S}
\begin{split}
     S(\bfx,\bfzeta)
=&
    \frac{
            1
         }{
            C_S
         }
    \alpha\sum\limits_{i=1}^{\infty}(1-\alpha)^{i-1}\BK{ \frac{2\pi}{RT(\bfx_{(i)})} }^{\frac{1}{2}} M_{T(\bfx_{(i)})}(\bfzeta),
    \\
    C_S
=&
    \frac{1}{|D|} \alpha\sum\limits_{i=1}^{\infty}(1-\alpha)^{i-1}\int \BK{ \frac{2\pi}{RT({\bfx_{(i)}}_*)} }^{\frac{1}{2}} M_{T({\bfx_{(i)}}_*)}(\bfzeta_*) d\bfx_* d\bfzeta_*,
\end{split}
\end{equation}
where the boundary point obtained by tracing back from the given interior point $\bfx$ along the direction $-\frac{\bfxi}{|\bfxi|}$:
\begin{equation*}
	\bfy_B \BKK{ \bfx, \frac{\bfxi}{|\bfxi|} }
=
	\bfx - \bfxi \times
	\sup \curBK{ s\geq 0: \bfx - \bfxi s' \in D, \text{ for all } s' \in (0,s) }, 
\end{equation*}
\begin{equation*}
\begin{split}
\bfx_{(1)} &=	\bfy_B \BKK{ \bfx, \frac{\bfxi}{|\bfxi|} },\\
\bfxi^1 &= \bfxi -2(\bfxi\cdot\bfn(\bfx_{(1)}))\bfn(\bfx_{(1)}),\\
\bfx_{(k+1)} &=	\bfy_B \BKK{ \bfx_{(k)}, \frac{\bfxi^k}{|\bfxi^k|} }, \\
\bfxi^{k+1} &= \bfxi^k -2(\bfxi^k\cdot\bfn(\bfx_{(k+1)}))\bfn(\bfx_{(k+1)}).\\
\end{split}
\end{equation*}
Since the domain $D$ is symmetric, we have the following lemma: 
\begin{lemma}\label{lemma:specular}
For $(\bfx,\bfxi)\in D\times\bbR^d$ and each $k\geq1$,
\begin{equation*}
\begin{split}
&|\bfxi^{k}|=|\bfxi|,\\
&|\bfx_{(k+1)}-\bfx_{(k)}|=|\bfx_{(k+2)}-\bfx_{(k+1)}|\geq|\bfx-\bfx_{(1)}|.
\end{split}
\end{equation*}
\end{lemma}

%Figure \ref{fig-boundary-point}. 
From the explicit expression \eqref{eq:Fr:Express:S},
\begin{equation}\label{eq:FullBz:S:minus:M:Estimate}
     S(\bfx,\bfzeta) - M(\bfzeta)  = O(1-\Tm) M(\bfzeta).
\end{equation}
We note that $S$ has constant boundary flux $1/C_S$:
\begin{equation}\label{eq:Fr:Cons:Flux:S}
\begin{split}
&
    \int_{\bfxi\cdot\bfn<0} -\bfxi\cdot\bfn  S(\bfy,\bfzeta) d\bfzeta
    \\
=&
    \frac{1}{C_S}  \int_{\bfxi\cdot\bfn<0} -\bfxi\cdot\bfn
    \alpha\sum\limits_{i=1}^{\infty}(1-\alpha)^{i-1}\BK{ \frac{2\pi}{RT(\bfy_{(i)})} }^{\frac{1}{2}} M_{T(\bfy_{(i)})}(\bfzeta) d\bfzeta
    \\
=&
    \frac{1}{C_S} \alpha\sum\limits_{i=1}^{\infty}(1-\alpha)^{i-1}\BK{4\pi}^\frac12
    \int_{\hat\bfxi\cdot\bfn<0} -\hat\bfxi\cdot\bfn  M(\hat\bfzeta) d\hat\bfzeta
=
    \frac{1}{C_S}.
\end{split}
\end{equation}

Note that both the evolutionary equation \eqref{eq:Fr:Eq:Eq} and the boundary condition \eqref{eq::Diff:Rel:BC} conserve molecular number, therefore the total molecular number $\int g(\bfx,\bfzeta,t) d\bfx d\bfzeta$ is a constant of time. We define the average total
density as:
\begin{equation*}
    \rho_*
\equiv
    \frac{1}{|D|} \int_{D\times\bbR^3} g_{in}(\bfx,\bfzeta) d\bfx d\bfzeta
=
    \frac{1}{|D|} \int_{D\times\bbR^3} g(\bfx,\bfzeta,t) d\bfx d\bfzeta,
\end{equation*}
a constant associated with $g_{in}$. Due to the equilibrating effect of the Maxwell-type boundary condition, one can expect the solution $g$ to approach the steady state $\rho_* S$. Namely, we expect the function $g - \rho_* S(\bfx,\bfzeta)$ to decay to zero. Moreover, $g-\rho_*S$ satisfies the same evolutionary equation \eqref{eq:Fr:Eq:Eq} and the boundary condition \eqref{eq::Diff:Rel:BC}. Since the space dimension is $d$ and $d<3$, it is natural to integrate out the extra microscopic velocity degrees of freedom:
\begin{align}\label{eq:Fr:Reduced:Defn:g}
         \bar g (\bfx,\bfxi,t)
\equiv&
        \int_{\bbR^{3-d}} \BK{ g(\bfx,\bfzeta,t) - \rho_* S(\bfx,\bfzeta) } d\bfeta,
        \\
         \bar g_{in}(\bfx,\bfxi)
\equiv&
        \int_{\bbR^{3-d}} \BK{ g_{in}(\bfx,\bfzeta) - \rho_* S(\bfx,\bfzeta) } d\bfeta,
        \\
        \boldsymbol{s}(\bfx,\bfxi)
\equiv&
        \int_{\bbR^{3-d}} S(\bfx,\bfzeta)  d\bfeta.
\end{align}
Recall that $\bfeta$ is the last $3-d$ components of $\bfzeta$, \eqref{eq::Def:xi:eta}. Since $S$ has constant boundary flux, the corresponding boundary flux becomes $j_g-\rho_*/C_S$:
\begin{equation}\label{eq:Fr:Reduced:Defn:j}
\begin{split}
     j(\bfy,t)
\equiv &
    \intLim_{\bfxi\cdot\bfn<0} -\bfxi\cdot\bfn
    \bar g (\bfx,\bfxi,t)d\bfxi
    \\
=&    
    \intLim_{\bfxi\cdot\bfn<0} -\bfxi\cdot\bfn
    \BK{ \intLim_{\bbR^{3-d}} \BK{ g(\bfy,\bfzeta,t) - \rho_* S(\bfy,\bfzeta) } d \bfeta } d\bfxi
    \\
=&
    \intLim_{\bfxi\cdot\bfn<0} -\bfxi\cdot\bfn g(\bfy,\bfzeta,t) d \bfzeta - \rho_*/C_S
=
    j_{g}(\bfy,t) - \frac{\rho_*}{C_S}\equiv  j_{g}(\bfy,t) -j_S,
\end{split}
\end{equation}
and the total molecular number becomes zero:
\begin{equation}\label{eq:Fr:Zero:Total:Mass}
    \intLim_{D\times\bbR^3} \BK{ g(\bfx,\bfzeta,t) - \rho_* S(\bfx,\bfzeta) } d\bfx d\bfzeta
=
    \intLim_{D\times\bbR^d} \bar g(\bfx,\bfxi,t) d\bfx d\bfxi
=
     0.
\end{equation}
Moreover, the new functions $\bar g(\bfx,\bfxi,t), j(\bfy,t)$ satisfy equations similar to that for the original functions:
\begin{subequations}\label{eq:Fr:Reduced:Eq}
\begin{align}
&
    \label{eq:Fr:Reduced:Eq:Eq}
    \begin{dcases}
        \frac{\partial \bar g}{\partial t} + \sum_{i=1}^d \xi_i \frac{ \partial \bar g }{ \partial x_i } = 0,
        \quad \bar g = \bar g(\bfx,\bfxi,t),\ \bfx \in D \subset \bbR^d, \ \bfxi \in \bbR^d, \ t >0,
        \\
         \bar g(\bfx,\bfxi,0)= \bar g_{in}(\bfx,\bfxi),        
    \end{dcases}
\\
&
    \label{eq:Fr:Reduced:Eq:Diff:Ref:BC}
    \begin{dcases}
         \bar g(\bfy,\bfxi,t) &=
        \alpha\BK{ \frac{2\pi}{RT(\bfy)} }^{\frac{1}{2}}j(\bfy,t) M_{T(\bfy)}(\bfxi)\\
        &+ (1-\alpha)\bar g(\bfy,\bfxi-2(\bfxi\cdot\bfn)\bfn,t),
        \quad \bfy\in \partial D, \ \bfxi\cdot\bfn>0,
        \\
         M_T (\bfxi) &=
        \intLim_{\mathbb{R}^{3-d}} M_T(\bfzeta) d\bfeta =
        \frac{ e^{ -\frac{\absbfxi^2}{2RT} } }{ \BK{2\pi RT}^{\frac{d}{2}} },
    \end{dcases}
\end{align}
\end{subequations}
but with the additional zero total molecular number condition:
\begin{equation}\label{eq:Fr:Reduced:zero:total:mass}
 \intLim_{D\times\bbR^d} \bar g(\bfx,\bfxi,t) d\bfx d\bfxi = 0, \ t\ge 0.
\end{equation}
Note that $M_T(\bfzeta)$, the Maxwellian, and $M_T(\bfxi)$, the reduced Maxwellian, are generally different as functions. To avoid confusion, we always refer to $M$ as the abbreviation of $M(\bfzeta)$, {\it not} $M(\bfxi)$.

For $\bfx\in D$ and $\bfxi\in\bbR^d$, we define $\tau_b = \tau_b(\bfx,\bfxi)$  the {\bf backward exit time}:
\begin{equation}\label{eq::Backward:time}
	\tau_b(\bfx,\bfxi) \equiv 
	\sup \curBK{ s \geq 0 : \bfx - s' \bfxi \in D, \text{ for all } s' \in (0,s) } 
,
\end{equation}
and
\begin{equation*}
\begin{split}
&t_1=\tau_b(\bfx,\bfxi)=\frac{|\bfx-\bfx_{(1)}|}{|\bfxi|},\\
&t_{k+1}=\tau_b(\bfx_{(k)},\bfxi^{k})=\frac{|\bfx_{(k)}-\bfx_{(k+1)}|}{|\bfxi^k|}.
\end{split}
\end{equation*}
From Lemma \ref{lemma:specular}, we have $t_k=t_2$ for all $k\geq2$.

Suppose that the boundary flux $j$ is given. Then the solutions of the transport equation \eqref{eq:Fr:Reduced:Eq} has explicit form by the characteristic method:
\begin{multline}\label{eq:Fr:Reduced:Chara:Rep}
     g(\bfx,\bfzeta,t) - \rho_*S(\bfx,\bfzeta)\\
=
    \begin{dcases}
         &\alpha \sum\limits_{k=0}^{m-1}(1-\alpha)^{k} \Big(j_g\BKK{ \bfx_{(k+1)}, t-t_1-kt_2 }-j_S\Big)
           \tilde{M}_{T(\bfx_{(k+1)})}\\
         & +\Big((1-\alpha)^{m}g_{in}(\bfx_{(m)}-\bfxi^m(t-t_1-(m-1)t_2),\bfxi^m,\bfeta)\\
         &-\alpha \sum\limits_{k=m}^{\infty}(1-\alpha)^{k} j_S
           \tilde{M}_{T(\bfx_{(k+1)})}\Big)  \quad  
        \text{ for } \tau_b < t,
        \\
         &g_{in}(\bfx-\bfxi t,\bfzeta)-\rho_*S(\bfx,\bfzeta)
        \quad
        \text{ for } t < \tau_b,
     \end{dcases}
\end{multline}
\begin{equation}\label{eq:Fr:Reduced:Chara:Rep1}
     \bar g(\bfx,\bfxi,t)
=
    \begin{dcases}
         &\alpha \sum\limits_{k=0}^{m-1}(1-\alpha)^{k} j\BKK{ \bfx_{(k+1)}, t-t_1-kt_2 }  \tilde{M}_{T(\bfx_{(k+1)})}(\bfxi^k)\\
         & +(1-\alpha)^{m}\bar g_{in}(\bfx_{(m)}-\bfxi^m(t-t_1-(m-1)t_2),\bfxi^m)  \quad  
        \text{ for } \tau_b < t,
        \\
         &\bar g_{in}(\bfx-\bfxi t,\bfxi)
        \quad
        \text{ for } t < \tau_b,
     \end{dcases}
\end{equation}
where
\begin{equation}\label{eq::specular:m}
m=\lfloor\frac{|\bfxi|t-|\bfx-\bfx_{(1)}|}{|\bfx_{(1)}-\bfx_{(2)}|}\rfloor +1,
\end{equation}
and for simplicity of notation we set
\begin{equation*}
\tilde{M}_{T(\bfy)}\equiv \BK{ \frac{2\pi}{RT(\bfy)} }^{\frac{1}{2}}
M_{T(\bfy)}.
\end{equation*}

The following are our main theorems for free molecular flow, which will be proven in the following four sections.
\begin{theorem}[Global Existence for Boundary Flux]\label{thm:Fr:Existence:Bddness}
The solution of \eqref{eq:Fr:Eq:Eq} and \eqref{eq:Fr:Reduced:Eq}, with initial data $g_{in} \in L^{\infty,\mu}_{\bfx,\bfzeta}$, exists and is unique for $\mu>4$. Moreover, there exists $C>0$ such that
\begin{equation}\label{eq:Fr:Coarse:Estimate}
j(\bfy,t)=O(1)\VertBK{g_{in}}_{\infty,\mu}e^{C\frac{-\alpha\ln\alpha}{1-\alpha} t},
\end{equation}
where $C$ depends on $\TM$ and $\Tm$.
\end{theorem}
\begin{theorem}[Decay Rate for Boundary Flux]\label{thm:Fr:Main}
Suppose that $g_{in}\in L^{\infty,\mu}_{\bfx,\bfzeta}$ for some constant $\mu>4$. Then the boundary flux $j(\bfy,t)$, \eqref{eq:Fr:Reduced:Defn:j}, satisfies
\begin{equation*}
j(\bfy,t)\leq C\VertBK{g_{in}}_{\infty,\mu}\BK{\frac{1}{(1+\alpha t)^d}+(1-\alpha)^{t^{\frac{1}{400}}}}
\end{equation*}
for some constant $C$ depending only on $\mu$, $\Tm$ and $\TM$.
\end{theorem}

From \eqref{eq::specular:m}, for $0<\epsilon<1$ we have $m\geq t^{\epsilon}$ for $|\bfxi|>\frac{2}{t^{1-\epsilon}}$. Therefore,
Theorem \ref{thm:Fr:Main} together with \eqref{eq:Fr:Reduced:Chara:Rep} yield immediately the pointwise convergence of the free molecular flow $g$, Theorem \ref{thm:Fr:Soln:Op:Ptws:Esti:new}.

\begin{remark}
When $\alpha=1$, the case of diffuse reflection boundary condition, we have the convergence rate $(1+t)^{-d}$ of the boundary flux, \cite{Kuo-Liu-Tsai,Kuo-Liu-Tsai-2}. Roughly speaking, the equilibrating effect is mainly from sufficiently many collisions with the boundary of diffuse reflection condition when  $t$ is large. For Maxwell-type boundary condition, we have two possibilities after each collision with the boundary: one is diffuse reflection and another is specular reflection. This yields multiple scales in the convergence to the steady solution, and is one of the main causes of the analytical difficulty of this paper. Eventually, we have the convergence rate $(1+\alpha t)^{-d} +(1-\alpha)^{\frac{t^{\epsilon}}{2}}$ of the boundary flux, Theorem \ref{thm:Fr:Main}. $(1-\alpha)^{\frac{t^{\epsilon}}{2}}$ comes from the coefficient $(1-\alpha)$ of specular reflection although specular reflection condition itself has no equilibrating effect. $(1+\alpha t)^{-d}$ is from the diffuse reflection condition where the rate is essentially the same as before. However, we only have diffuse reflection for a multiple of $\alpha$, and so the convergence to steady solution is slower than the case of the complete diffuse reflection. For instance, it starts to converge only when $t>\frac{1}{\alpha}$; before that the solution is simply bounded.
\end{remark}

\end{subsection}

\begin{subsection}{The global exitence for boundary flux}\label{sec:Fr:Stoch:Formulation}
In this subsection, we prove the global existence for the boundary flux function $j(\bfy,t)$. It should be noticed that we may associate the boundary flux with the backward flow of particles. Once particles collide with the boundary, both diffuse reflection and specular reflection occur for the Maxwell-type boundary condition.
Note that diffuse reflection is stochastic and see \cite{Kuo-Liu-Tsai} for more details. In contrast to diffuse reflection, specular reflection is deterministic and has no equilibrating effect. In the following discussion we first give a solution formula of boundary flux for general domains. Here we assume temporarily that the accommodation coefficient is variable, $0<\alpha(\bfy)<1$, for explaining what difficulties arise from this assumption.

%\begin{equation*}
%\tilde{M}_{T(\bfy)}(\bfxi)=\BK{ \frac{2\pi}{RT(\bfy)} }^{\frac{1}{2}}
%M_{T(\bfy)}(\bfxi).
%\end{equation*}
Fix $\bfy\in\partial D, t>0$, the boundary flux can be written as 
\begin{equation}\label{Sol:Formula:Fr:specular}
\begin{split}
j(\bfy,t)&=\intLim_{t<\frac{|\bfy-\bfy_{(1)}|}{|\bfxi_1|}}
\BK{-\bfxi_1\cdot\bfn(\bfy)}\bar g_{in}(\bfy-\bfxi_1t,\bfxi_1)d\bfxi_1\\
&+\intLim_{t>\frac{|\bfy-\bfy_{(1)}|}{|\bfxi_1|}}
\BK{-\bfxi_1\cdot\bfn(\bfy)}\alpha(\bfy_{(1)})
\tilde{M}_{T(\bfy_{(1)})}(\bfxi_1)j\BK{\bfy_{(1)},t-\frac{|\bfy-\bfy_{(1)}|}{|\bfxi_1|}}d\bfxi_1\\
&+\intLim_{t>\frac{|\bfy-\bfy_{(1)}|}{|\bfxi_1|}}
\BK{-\bfxi_1\cdot\bfn(\bfy)}\BK{1-\alpha(\bfy_{(1)})}\bar g\BK{\bfy_{(1)},t-\frac{|\bfy-\bfy_{(1)}|}{|\bfxi_1|},\bfxi_1^1}d\bfxi_1\\
&\equiv j_{in}^{(0)}(\bfy,t)+D^{(1)}(\bfy,t)+E^{(1)}(\bfy,t),
\end{split}
\end{equation}
where $j_{in}^{(0)}$ is a direct contribution of initial data, both $D^{(1)}$ and $E^{(1)}$ are the events that boundary collisions are more than once. More precisely, the first boundary collision takes place at $\bfy_{(1)}$, $D^{(1)}$ and $E^{(1)}$ represent diffuse reflection and specular reflection of the backward flow respectively. We can continue to write down the formulas for $D^{(1)}(\bfy,t)$ and $E^{(1)}(\bfy,t)$:
\begin{equation*}
\begin{split}
D^{(1)}(\bfy,t)&=D^{(1)}_{in}(\bfy,t)+D^{(1)}_{dif}(\bfy,t)+D^{(1)}_{spe}(\bfy,t);\\
E^{(1)}(\bfy,t)&=E^{(1)}_{in}(\bfy,t)+E^{(1)}_{dif}(\bfy,t)+E^{(1)}_{spe}(\bfy,t),
\end{split}
\end{equation*}
where
\begin{multline*}
D^{(1)}_{in}(\bfy,t)=\intLim_{0<t-\frac{|\bfy-\bfy_{(1)}|}{|\bfxi_1|}<\frac{|\bfy_{(1)}-\bfy_{(1,1)}|}{|\bfxi_2|}}
\BK{-\bfxi_1\cdot\bfn(\bfy)}\alpha(\bfy_{(1)})\tilde{M}_{T(\bfy_{(1)})}(\bfxi_1)\\
\BK{-\bfxi_2\cdot\bfn(\bfy_{(1)})}\bar g_{in}(\bfy_{(1)}-\bfxi_2(t-\frac{|\bfy-\bfy_{(1)}|}{|\bfxi_1|}),\bfxi_2)d\bfxi_2d\bfxi_1;
\end{multline*}
\begin{multline*}
D^{(1)}_{dif}(\bfy,t)=\intLim_{t-\frac{|\bfy-\bfy_{(1)}|}{|\bfxi_1|}>\frac{|\bfy_{(1)}-\bfy_{(1,1)}|}{|\bfxi_2|}}
\BK{-\bfxi_1\cdot\bfn(\bfy)}\alpha(\bfy_{(1)})\tilde{M}_{T(\bfy_{(1)})}(\bfxi_1)
\BK{-\bfxi_2\cdot\bfn(\bfy_{(1)})}\\
\alpha(\bfy_{(1,1)})
\tilde{M}_{T(\bfy_{(1,1)})}(\bfxi_2)j\BK{\bfy_{(1,1)},t-\frac{|\bfy-\bfy_{(1)}|}{|\bfxi_1|}-\frac{|\bfy_{(1)}-\bfy_{(1,1)}|}{|\bfxi_2|}}d\bfxi_2d\bfxi_1;
\end{multline*}
\begin{multline*}
D^{(1)}_{spe}(\bfy,t)=\intLim_{t-\frac{|\bfy-\bfy_{(1)}|}{|\bfxi_1|}>\frac{|\bfy_{(1)}-\bfy_{(1,1)}|}{|\bfxi_2|}}
\BK{-\bfxi_1\cdot\bfn(\bfy)}\alpha(\bfy_{(1)})\tilde{M}_{T(\bfy_{(1)})}(\bfxi_1)
\BK{-\bfxi_2\cdot\bfn(\bfy_{(1)})}\\
\BK{1-\alpha(\bfy_{(1,1)}}\bar g\BK{\bfy_{(1,1)},t-\frac{|\bfy-\bfy_{(1)}|}{|\bfxi_1|}-\frac{|\bfy_{(1)}-\bfy_{(1,1)}|}{|\bfxi_2|},\bfxi_2^1}d\bfxi_2d\bfxi_1;
\end{multline*}
\begin{multline*}
E^{(1)}_{in}(\bfy,t)=\intLim_{0<t-\frac{|\bfy-\bfy_{(1)}|}{|\bfxi_1|}<\frac{|\bfy_{(1)}-\bfy_{(2)}|}{|\bfxi_1^1|}}
\BK{-\bfxi_1\cdot\bfn(\bfy)}(1-\alpha(\bfy_{(1)})\\
\bar g_{in}\BK{\bfy_{(1)}-\bfxi_1^1\BK{t-\frac{|\bfy-\bfy_{(1)}|}{|\bfxi_1|}}}d\bfxi_1;\\
\end{multline*}
\begin{multline*}
E^{(1)}_{dif}(\bfy,t)=\intLim_{t-\frac{|\bfy-\bfy_{(1)}|}{|\bfxi_1|}>\frac{|\bfy_{(1)}-\bfy_{(2)}|}{|\bfxi_1^1|}}
\BK{-\bfxi_1\cdot\bfn(\bfy)}(1-\alpha(\bfy_{(1)})\\
\alpha(\bfy_{(2)})
\tilde{M}_{T(\bfy_{(2)})}(\bfxi_1^1)j\BK{\bfy_{(2)},t-\frac{|\bfy-\bfy_{(1)}|}{|\bfxi_1|}-\frac{|\bfy_{(1)}-\bfy_{(2)}|}{|\bfxi_1^1|}}d\bfxi_1;
\end{multline*}
\begin{multline*}
E^{(1)}_{spe}(\bfy,t)=\intLim_{t-\frac{|\bfy-\bfy_{(1)}|}{|\bfxi_1|}>\frac{|\bfy_{(1)}-\bfy_{(2)}|}{|\bfxi_1^1|}}
\BK{-\bfxi_1\cdot\bfn(\bfy)}(1-\alpha(\bfy_{(1)})\\
\BK{1-\alpha(\bfy_{(2)}}\bar g\BK{\bfy_{(2)},t-\frac{|\bfy-\bfy_{(1)}|}{|\bfxi_1|}-\frac{|\bfy_{(1)}-\bfy_{(2)}|}{|\bfxi_1^1|},\bfxi_1^2}d\bfxi_1.
\end{multline*}
%\begin{multline}
%D^{(1)}(\bfy,t)=\intLim_{0<t-\frac{|\bfy-\bfy_{(1)}|}{|\bfxi_1|}<\frac{|\bfy_{(1)}-\bfy_{(1,1)}|}{|\bfxi_2|}}
%\BK{-\bfxi_1\cdot\bfn(\bfy)}\alpha(\bfy_{(1)})\tilde{M}_{T(\bfy_{(1)})}(\bfxi_1)\\
%\BK{-\bfxi_2\cdot\bfn(\bfy_{(1)})}\bar g_{in}(\bfy_{(1)}-\bfxi_2(t-\frac{|\bfy-\bfy_{(1)}|}{|\bfxi_1|}),\bfxi_2)d\bfxi_2d\bfxi_1\\
%+\intLim_{t-\frac{|\bfy-\bfy_{(1)}|}{|\bfxi_1|}>\frac{|\bfy_{(1)}-\bfy_{(1,1)}|}{|\bfxi_2|}}
%\BK{-\bfxi_1\cdot\bfn(\bfy)}\alpha(\bfy_{(1)})\tilde{M}_{T(\bfy_{(1)})}(\bfxi_1)
%\BK{-\bfxi_2\cdot\bfn(\bfy_{(1)})}\\
%\Bigg\{\alpha(\bfy_{(1,1)})
%\tilde{M}_{T(\bfy_{(1,1)})}(\bfxi_2)j\BK{\bfy_{(1,1)},t-\frac{|\bfy-\bfy_{(1)}|}{|\bfxi_1|}-\frac{|\bfy_{(1)}-\bfy_{(1,1)}|}{|\bfxi_2|}}\\
%+\BK{1-\alpha(\bfy_{(1,1)}}\bar g\BK{\bfy_{(1,1)},t-\frac{|\bfy-\bfy_{(1)}|}{|\bfxi_1|}-\frac{|\bfy_{(1)}-\bfy_{(1,1)}|}{|\bfxi_2|},\bfxi_2^1}\Bigg\}d\bfxi_2d\bfxi_1\\
%\equiv D^{(1)}_{in}(\bfy,t)+D^{(1)}_{dif}(\bfy,t)+D^{(1)}_{spe}(\bfy,t),
%\end{multline}
$D^{(1)}_{in}$ is the event that the backward flow reaches initial state after diffuse reflection occurs once. $D^{(1)}_{dif}$ is the event that diffuse reflection occurs again at $\bfy_{(1,1)}$ after the first diffuse reflection at $\bfy_{(1)}$, and $D^{(1)}_{spe}$ is the event that specular reflection occurs at $\bfy_{(1,1)}$ after the first diffuse reflection at $\bfy_{(1)}$. 
%\begin{multline}
%E^{(1)}(\bfy,t)=\intLim_{0<t-\frac{|\bfy-\bfy_{(1)}|}{|\bfxi_1|}<\frac{|\bfy_{(1)}-\bfy_{(2)}|}{|\bfxi_1^1|}}
%\BK{-\bfxi_1\cdot\bfn(\bfy)}(1-\alpha(\bfy_{(1)})\\
%\bar g_{in}\BK{\bfy_{(1)}-\bfxi_1^1\BK{t-\frac{|\bfy-\bfy_{(1)}|}{|\bfxi_1|}}}d\bfxi_1\\
%+\intLim_{t-\frac{|\bfy-\bfy_{(1)}|}{|\bfxi_1|}>\frac{|\bfy_{(1)}-\bfy_{(2)}|}{|\bfxi_1^1|}}
%\BK{-\bfxi_1\cdot\bfn(\bfy)}(1-\alpha(\bfy_{(1)})\\
%\Bigg\{\alpha(\bfy_{(2)})
%\tilde{M}_{T(\bfy_{(2)})}(\bfxi_1^1)j\BK{\bfy_{(2)},t-\frac{|\bfy-\bfy_{(1)}|}{|\bfxi_1|}-\frac{|\bfy_{(1)}-\bfy_{(2)}|}{|\bfxi_1^1|}}\\
%+\BK{1-\alpha(\bfy_{(2)}}\bar g\BK{\bfy_{(2)},t-\frac{|\bfy-\bfy_{(1)}|}{|\bfxi_1|}-\frac{|\bfy_{(1)}-\bfy_{(2)}|}{|\bfxi_1^1|},\bfxi_1^2}\Bigg\}d\bfxi_1,\\
%\equiv E^{(1)}_{in}(\bfy,t)+E^{(1)}_{dif}(\bfy,t)+E^{(1)}_{spe}(\bfy,t),
%\end{multline}
$E^{(1)}_{in}$ is the event that the backward flow reaches initial state after specular reflection occurs once. $E^{(1)}_{dif}$ is the event that diffuse reflection occurs at $\bfy_{(2)}$ after the first specular reflection at $\bfy_{(1)}$, and $E^{(1)}_{spe}$ is the event that specular reflection occurs again at $\bfy_{(2)}$ after the first specular reflection at $\bfy_{(1)}$. 
Note that both $D^{(1)}_{in}$ and $E^{(1)}_{in}$ represent the contribution that the boundary collision takes place exactly once. More precisely, $D^{(1)}_{in}$ and $E^{(1)}_{in}$ represent exactly one diffuse collision and exactly one specular collision, respectively. If we want to compute the contribution that the boundary collision takes place exactly twice, we need to take $D^{(1)}_{dif}, D^{(1)}_{spe},E^{(1)}_{dif}$ and $E^{(1)}_{spe}$ into account. In other words, we must proceed to write down their formulas. Then there are four events arisen for exact two boundary collisions: (\textit{diffuse, diffuse}), (\textit{diffuse, specular}), (\textit{specular, diffuse}) and (\textit{specular, specular}).
One can repeat this process inductively to compute the event that the boundary collision takes place exactly $n$ times. In that case, we need to handle $2^n$ possibilities. That would make the solution formula lengthy and complicated.
That is one of the main causes of the analytical difficulty of this paper.
Another tricky problem is the variable accommodation coefficient $\alpha(\bfy)$. In this paper we assume $\alpha(\bfy)=\alpha$ is a constant, this assumption  not only makes the solution formula easier but also allows us to estimate all combinations of the events. We explain this by considering binomial expansion formally:
\begin{equation}\label{binomial}
\BK{\alpha\textit{diffuse}+(1-\alpha)\textit{specular}}^n
=\sum_{k=0}^n\BK{\alpha\textit{diffuse}}^k\BK{(1-\alpha)\textit{specular}}^{n-k}.
\end{equation}
Then all combinations of events, R.H.S. of \eqref{binomial}, can be dominated by
\begin{equation*}
\BK{\alpha\textit{diffuse}+(1-\alpha)\textit{specular}}^n
=O(1)(\alpha+(1-\alpha))^n=O(1),
\end{equation*}
for example, if we can show each term of diffuse and specular is bounded.
In other words, we can treat the effects caused by diffuse reflection and specular reflection independently when $\alpha$ is constant. For the variable accommodation coefficient $\alpha(\bfy)$, the problem is more delicate and might involve different techniques. This will be our another research work in the future.

From now on we assume $0<\alpha<1$ is constant. We define the following notations inductively:
\begin{equation*}
\begin{split}
&\bfy_{(0)}\equiv \bfy, \quad \bfy_{(k_1,\ldots,k_l,0)}\equiv\bfy_{(k_1,\ldots,k_l)}, 
\quad \bfxi^{0}_l \equiv \bfxi_l, \\ 
&\bfy_{(k_1,\ldots,k_{l-1},i)} =\bfy_B \BKK{ \bfy_{(k_1,\ldots,k_{l-1},i-1)}, \frac{\bfxi^{i-1}_{l}}{|\bfxi^{i-1}_{l}|} },\\
&\bfxi^{i}_l = \bfxi^{i-1}_l -2(\bfxi^{i-1}_l\cdot\bfn(\bfy_{(k_1,\ldots,k_{l-1},i)}))\bfn(\bfy_{(k_1,\ldots,k_{l-1},i)}),\\
\end{split}
\end{equation*}
where $\bfy_{(k_1,\ldots,k_l)}$ indicates the location of particles via the backward flow process that:
\begin{multline*}
(k_1-1)\text{ specular }\rightarrow\text{ diffuse }\rightarrow
(k_2-1)\text{ specular }\rightarrow\text{ diffuse }\rightarrow\\
\cdots
(k_{l-1}-1)\text{ specular }\rightarrow\text{ diffuse }\rightarrow
(k_l-1)\text{ specular }.
\end{multline*}
According to the above discussion, we can find the solution formula of the  boundary flux for general domains.
%We will combine the decay of $(1-\alpha)^m$ after $m$ times specular reflections with the probability estimate
%of diffuse reflections to analyze the coupling of
%these two kinds of reflection in the Maxwell-type boundary condition.

 \begin{multline}\label{Sol:Formula:Fr:new}
j(\bfy,t)=\sum\limits_{k=0}^n\Bigg\{\sum\limits_{l=1}^k\sum\limits_{k_1+\ldots+k_l=l}^k
\intLim_{A^{(l,k)}_{(k_1,\ldots,k_l)}}
\prod\limits_{i=1}^l \BK{-\bfxi_i\cdot\bfn(\bfy_{(k_1,\ldots,k_{i-1})}}
(1-\alpha)^{k_i-1}\\ \alpha
\tilde{M}_{T(\bfy_{(k_1,\ldots,k_i)})}(\bfxi_i^{k_i-1})
\BK{-\bfxi_{l+1}\cdot\bfn(\bfy_{(k_1,\ldots,k_l)})}
(1-\alpha)^{k-k_1-\ldots-k_l}\\
\bar g_{in}\left(\bfy_{(k_1,\ldots,k_l,k-k_1-\ldots-k_l)}-\bfxi_{l+1}^{k-k_1-\ldots-k_l}
\Big(t-\sum\limits_{i=1}^{l}\sum\limits_{j=1}^{k_i}\frac{|\bfy_{(k_1,\ldots,k_{i-1},j-1)}-\bfy_{(k_1,\ldots,k_{i-1},j)}|}{|\bfxi_i^{j-1}|}\right.\\
\left.-\sum\limits_{i=1}^{k-k_1-\ldots-k_l}\frac{|\bfy_{(k_1,\ldots,k_l,i-1)}-\bfy_{(k_1,\ldots,k_l,i)}|}{|\bfxi_{l+1}^{i-1}|}\Big),\bfxi_{l+1}^{k-k_1-\ldots-k_l}\right)
d\bfxi_{l+1}\ldots d\bfxi_1\\
+\intLim_{0<t-\sum\limits_{i=0}^{k-1}\frac{|\bfy_{(i)}-\bfy_{(i+1)}|}{|\bfxi_1^i|}<\frac{|\bfy_{(k)}-\bfy_{(k+1)}|}{|\bfxi_1^k|}}
\BK{-\bfxi_1\cdot\bfn(\bfy)}(1-\alpha)^k \\
\bar g_{in}\BK{\bfy_{(k)}-\bfxi_1^k(t-\sum\limits_{i=0}^{k-1}\frac{|\bfy_{(i)}-\bfy_{(i+1)}|}{|\bfxi_1^i|}),\bfxi_1^k}d\bfxi_1\Bigg\}\\
+\sum\limits_{l=1}^{n+1}\sum\limits_{k_1+\ldots+k_l=l}^{n+1}
\intLim_{B^{(l,n+1)}_{(k_1,\ldots,k_l)}}
\prod\limits_{i=1}^l \BK{-\bfxi_i\cdot\bfn(\bfy_{(k_1,\ldots,k_{i-1})}}(1-\alpha)^{k_i-1}\\ \alpha
\tilde{M}_{T(\bfy_{(k_1,\ldots,k_i)})}(\bfxi_i^{k_i-1})
\BK{-\bfxi_{l+1}\cdot\bfn(\bfy_{(k_1,\ldots,k_l)})}
(1-\alpha)^{n+1-k_1-\ldots-k_l}\\
\bar g\left(\bfy_{(k_1,\ldots,k_l,n+1-k_1-\ldots-k_l)},
t-\sum\limits_{i=1}^{l}\sum\limits_{j=1}^{k_i}\frac{|\bfy_{(k_1,\ldots,k_{i-1},j-1)}-\bfy_{(k_1,\ldots,k_{i-1},j)}|}{|\bfxi_i^{j-1}|}\right.\\
\left.-\sum\limits_{i=1}^{n+1-k_1-\ldots-k_l}\frac{|\bfy_{(k_1,\ldots,k_l,i-1)}-\bfy_{(k_1,\ldots,k_l,i)}|}{|\bfxi_{l+1}^{i-1}|},\bfxi_{l+1}^{n+1-k_1-\ldots-k_l}\right)
d\bfxi_{l+1}\ldots d\bfxi_1\\
+\intLim_{t>\sum\limits_{i=0}^{n}\frac{|\bfy_{(i)}-\bfy_{(i+1)}|}{|\bfxi_1^i|}}
\BK{-\bfxi_1\cdot\bfn(\bfy)}(1-\alpha)^{n+1} 
\bar g\BK{\bfy_{(n+1)},t-\sum\limits_{i=0}^{n}\frac{|\bfy_{(i)}-\bfy_{(i+1)}|}{|\bfxi_1^i|},\bfxi_1^{n+1}}d\bfxi_1,
\end{multline}

where

\begin{multline*}
A^{(l,k)}_{(k_1,\ldots,k_l)}=\\
\left\{0<t-\sum\limits_{i=1}^{l}\sum\limits_{j=1}^{k_i}\frac{|\bfy_{(k_1,\ldots,k_{i-1},j-1)}-\bfy_{(k_1,\ldots,k_{i-1},j)}|}{|\bfxi_i^{j-1}|}-\sum\limits_{i=1}^{k-k_1-\ldots-k_l}\frac{|\bfy_{(k_1,\ldots,k_l,i-1)}-\bfy_{(k_1,\ldots,k_l,i)}|}{|\bfxi_{l+1}^{i-1}|}\right.\\
\left.<\begin{dcases}
\frac{|\bfy_{(k_1,\ldots,k_l)}-\bfy_{(k_1,\ldots,k_l,1)}|}{|\bfxi_{l+1}|}\quad \text{ if } k-k_1-\ldots-k_l=0,\\
\frac{|\bfy_{(k_1,\ldots,k_l,k-k_1-\ldots-k_l)}-\bfy_{(k_1,\ldots,k_l,k_l,k-k_1-\ldots-k_l+1)}|}{|\bfxi_{l+1}^{k-k_1-\ldots-k_l}|}\quad \text{ if } k-k_1-\ldots-k_l>0
\end{dcases}\right\},
\end{multline*}
\begin{multline*}
B^{(l,n+1)}_{(k_1,\ldots,k_l)}=\\
\left\{t>\sum\limits_{i=1}^{l}\sum\limits_{j=1}^{k_i}\frac{|\bfy_{(k_1,\ldots,k_{i-1},j-1)}-\bfy_{(k_1,\ldots,k_{i-1},j)}|}{|\bfxi_i^{j-1}|}+\sum\limits_{i=1}^{n+1-k_1-\ldots-k_l}\frac{|\bfy_{(k_1,\ldots,k_l,i-1)}-\bfy_{(k_1,\ldots,k_l,i)}|}{|\bfxi_{l+1}^{i-1}|}\right\}.
\end{multline*}

Note that
\begin{multline*}
\intLim_{B^{(l,n+1)}_{k_1,\ldots,k_l};k_1+\ldots+k_l=n+1}\Big(\ldots\Big)\\
=\intLim_{B^{(l,n+1)}_{k_1,\ldots,k_l};k_1+\ldots+k_l=n+1}\\
\prod\limits_{i=1}^l \BK{-\bfxi_i\cdot\bfn(\bfy_{(k_1,\ldots,k_{i-1})}}(1-\alpha)^{k_i-1}\alpha
\tilde{M}_{T(\bfy_{(k_1,\ldots,k_i)})}(\bfxi_i^{k_i-1})\\
j\left(\bfy_{(k_1,\ldots,k_l)},
t-\sum\limits_{i=1}^{l}\sum\limits_{j=1}^{k_i}\frac{|\bfy_{(k_1,\ldots,k_{i-1},j-1)}-\bfy_{(k_1,\ldots,k_{i-1},j)}|}{|\bfxi_i^{j-1}|}\right)
d\bfxi_{l}\ldots d\bfxi_1.
\end{multline*}

In the present paper we assume spherical symmetric domains and therefore we can make use of this symmetric property to obtain more precise formulas for the boundary flux by using change of variables:
\begin{equation*}
\begin{dcases}
s_i=\frac{2}{k_i|\bfxi_i|} \quad \text{ if } d=1,\\
\begin{dcases}
s_i=\frac{|\bfy_{(k_1,\ldots,k_{i-1},0)}-\bfy_{(k_1,\ldots,k_{i-1},1)}|}{k_i|\bfxi_i|}\\
\cos\phi_i=-\bfxi_i\cdot\bfn(\bfy_{(k_1,\ldots,k_{i-1})})
\end{dcases} \quad \text{ if } d=2.
\end{dcases}
\end{equation*}
We define
\begin{equation}\label{eq:Fr:Defn:G}
\begin{dcases}
H(\sigma)
\equiv
\BK{\frac{2}{\sigma}}^3 e^{-\BK{\frac{2}{\sigma}}^2} \quad \text{ if } d=1,\\
     G(\phi,\sigma)
\equiv
    \frac{ 1 }{ \pi^{\frac{1}{2}} }  \BK{ \CosS{}{} }^{4} \Maxll{}{}  \quad\text{ if } d=2.
    \end{dcases}
\end{equation}
For $d=2$,
 \begin{multline*}
j(\bfy,t)=\sum\limits_{k=0}^n\Bigg\{\sum\limits_{l=1}^k\sum\limits_{k_1+\ldots+k_l=l}^k
\intLim_{A^{(l,k)}_{(k_1,\ldots,k_l)}}
\prod\limits_{i=1}^l (1-\alpha)^{k_i-1}\\ \alpha
G\Big(\phi_i,\frac{\sqrt{2RT(\bfy_{(k_1,\ldots,k_i)})}s_i}{k_i}\Big)\frac{\sqrt{2RT(\bfy_{(k_1,\ldots,k_i)})}}{k_i}
\BK{-\bfxi_{l+1}\cdot\bfn(\bfy_{(k_1,\ldots,k_l)})}
(1-\alpha)^{k-k_1-\ldots-k_l}\\
\bar g_{in}\Big(\bfy_{(k_1,\ldots,k_l)}-\bfxi_{l+1}^{k-k_1-\ldots-k_l}
\Big(t-\sum\limits_{i=1}^{l}s_i
-\sum\limits_{i=1}^{k-k_1-\ldots-k_l}\frac{|\bfy_{(k_1,\ldots,k_l,i-1)}-\bfy_{(k_1,\ldots,k_l,i)}|}{|\bfxi_{l+1}^{i-1}|}\Big),\bfxi_{l+1}^{k-k_1-\ldots-k_l}\Big)\\
d\bfxi_{l+1}ds_ld\phi_l\ldots ds_1d\phi_1 \\
\end{multline*}

\begin{multline*}
+\intLim_{0<t-\sum\limits_{i=0}^{k-1}\frac{|\bfy_{(i)}-\bfy_{(i+1)}|}{|\bfxi_1^i|}<\frac{|\bfy_{(k)}-\bfy_{(k+1)}|}{|\bfxi_1^k|}}\\
\left.\BK{-\bfxi_1\cdot\bfn(\bfy)}(1-\alpha)^k 
\bar g_{in}\BK{\bfy_{(k)}-\bfxi_1^k(t-\sum\limits_{i=0}^{k-1}\frac{|\bfy_{(i)}-\bfy_{(i+1)}|}{|\bfxi_1^i|}),\bfxi_1^k}d\bfxi_1\right\}\\
\end{multline*}

\begin{multline*}
+\sum\limits_{l=1}^{n+1}\sum\limits_{k_1+\ldots+k_l=l}^{n+1}
\intLim_{B^{(l,n+1)}_{(k_1,\ldots,k_l)}}
\prod\limits_{i=1}^l (1-\alpha)^{k_i-1}\\ \alpha
G\Big(\phi_i,\frac{\sqrt{2RT(\bfy_{(k_1,\ldots,k_i)})}s_i}{k_i}\Big)\frac{\sqrt{2RT(\bfy_{(k_1,\ldots,k_i)})}}{k_i}
\BK{-\bfxi_{l+1}\cdot\bfn(\bfy_{(k_1,\ldots,k_l)})}
(1-\alpha)^{n+1-k_1-\ldots-k_l}\\
\bar g\Big(\bfy_{(k_1,\ldots,k_l)},
t-\sum\limits_{i=1}^{l}s_i
-\sum\limits_{i=1}^{n+1-k_1-\ldots-k_l}\frac{|\bfy_{(k_1,\ldots,k_l,i-1)}-\bfy_{(k_1,\ldots,k_l,i)}|}{|\bfxi_{l+1}^{i-1}|},\bfxi_{l+1}^{n+1-k_1-\ldots-k_l}\Big)\\
d\bfxi_{l+1}ds_ld\phi_l\ldots ds_1d\phi_1\\
\end{multline*}

\begin{multline*}
+\intLim_{t>\sum\limits_{i=0}^{n}\frac{|\bfy_{(i)}-\bfy_{(i+1)}|}{|\bfxi_1^i|}}
\BK{-\bfxi_1\cdot\bfn(\bfy)}(1-\alpha)^{n+1} 
\bar g\BK{\bfy_{(n+1)},t-\sum\limits_{i=0}^{n}\frac{|\bfy_{(i)}-\bfy_{(i+1)}|}{|\bfxi_1^i|},\bfxi_1^{n+1}}d\bfxi_1.
\end{multline*}

For $d=1$,
 \begin{multline*}
j(\bfy,t)=\sum\limits_{k=0}^n\left\{\sum\limits_{l=1}^k\sum\limits_{k_1+\ldots+k_l=l}^k
\intLim_{A^{(l,k)}_{(k_1,\ldots,k_l)}}\right.\\
\prod\limits_{i=1}^l (1-\alpha)^{k_i-1}\alpha
H(\frac{\sqrt{2RT(\bfy_{(k_1,\ldots,k_i)})}s_i}{k_i})\frac{\sqrt{2RT(\bfy_{(k_1,\ldots,k_i)})}}{k_i}\\
\BK{-\bfxi_{l+1}\cdot\bfn(\bfy_{(k_1,\ldots,k_l)})}
(1-\alpha)^{k-k_1-\ldots-k_l}\\
\bar g_{in}\Big(\bfy_{(k_1,\ldots,k_l)}-\bfxi_{l+1}^{k-k_1-\ldots-k_l}
(t-\sum\limits_{i=1}^{l}s_i
-\sum\limits_{i=1}^{k-k_1-\ldots-k_l}\frac{|\bfy_{(k_1,\ldots,k_l,i-1)}-\bfy_{(k_1,\ldots,k_l,i)}|}{|\bfxi_{l+1}^{i-1}|}),\bfxi_{l+1}^{k-k_1-\ldots-k_l}\Big)\\
d\bfxi_{l+1}ds_l\ldots ds_1\\
\end{multline*}

\begin{multline*}
+\intLim_{0<t-\sum\limits_{i=0}^{k-1}\frac{|\bfy_{(i)}-\bfy_{(i+1)}|}{|\bfxi_1^i|}<\frac{|\bfy_{(k)}-\bfy_{(k+1)}|}{|\bfxi_1^k|}}\\
\left.\BK{-\bfxi_1\cdot\bfn(\bfy)}(1-\alpha)^k 
\bar g_{in}\BK{\bfy_{(k)}-\bfxi_1^k(t-\sum\limits_{i=0}^{k-1}\frac{|\bfy_{(i)}-\bfy_{(i+1)}|}{|\bfxi_1^i|}),\bfxi_1^k}d\bfxi_1\right\}
\end{multline*}

\begin{multline*}
+\sum\limits_{l=1}^{n+1}\sum\limits_{k_1+\ldots+k_l=l}^{n+1}
\intLim_{B^{(l,n+1)}_{(k_1,\ldots,k_l)}}\\
\prod\limits_{i=1}^l (1-\alpha)^{k_i-1}\alpha
H(\frac{\sqrt{2RT(\bfy_{(k_1,\ldots,k_i)})}s_i}{k_i})\frac{\sqrt{2RT(\bfy_{(k_1,\ldots,k_i)})}}{k_i}\\
\BK{-\bfxi_{l+1}\cdot\bfn(\bfy_{(k_1,\ldots,k_l)})}
(1-\alpha)^{n+1-k_1-\ldots-k_l}\\
\bar g\Big(\bfy_{(k_1,\ldots,k_l)},
t-\sum\limits_{i=1}^{l}s_i
-\sum\limits_{i=1}^{n+1-k_1-\ldots-k_l}\frac{|\bfy_{(k_1,\ldots,k_l,i-1)}-\bfy_{(k_1,\ldots,k_l,i)}|}{|\bfxi_{l+1}^{i-1}|},\bfxi_{l+1}^{n+1-k_1-\ldots-k_l}\Big)\\
d\bfxi_{l+1}ds_l\ldots ds_1\\
\end{multline*}

\begin{multline*}
+\intLim_{t>\sum\limits_{i=0}^{n}\frac{|\bfy_{(i)}-\bfy_{(i+1)}|}{|\bfxi_1^i|}}
\BK{-\bfxi_1\cdot\bfn(\bfy)}(1-\alpha)^{n+1} 
\bar g\BK{\bfy_{(n+1)},t-\sum\limits_{i=0}^{n}\frac{|\bfy_{(i)}-\bfy_{(i+1)}|}{|\bfxi_1^i|},\bfxi_1^{n+1}}d\bfxi_1.
\end{multline*}

To simplify the equations, we define
\begin{multline}\label{eq:flux:initial}
j_{in}^{(k)}(\bfy,t)=
\intLim_{0<t-\sum\limits_{i=0}^{k-1}\frac{|\bfy_{(i)}-\bfy_{(i+1)}|}{|\bfxi_1^i|}<\frac{|\bfy_{(k)}-\bfy_{(k+1)}|}{|\bfxi_1^k|}}
\BK{-\bfxi_1\cdot\bfn(\bfy)} \\
\bar g_{in}\BK{\bfy_{(k)}-\bfxi_1^k(t-\sum\limits_{i=0}^{k-1}\frac{|\bfy_{(i)}-\bfy_{(i+1)}|}{|\bfxi_1^i|}),\bfxi_1^k}d\bfxi_1,\\
\end{multline}
\begin{equation}\label{eq:flux:E}
\begin{split}
&E^{(n+1)}(\bfy,t)=
\intLim_{t>\sum\limits_{i=0}^{n}\frac{|\bfy_{(i)}-\bfy_{(i+1)}|}{|\bfxi_1^i|}}
\BK{-\bfxi_1\cdot\bfn(\bfy)} 
\bar g\BK{\bfy_{(n+1)},t-\sum\limits_{i=0}^{n}\frac{|\bfy_{(i)}-\bfy_{(i+1)}|}{|\bfxi_1^i|},\bfxi_1^{n+1}}d\bfxi_1,\\
&E^{(0)}(\bfy,t)=j(\bfy,t),
\end{split}
\end{equation}
where $j^{(k)}_{in}$ represents the event that the backward flow reaches initial state after $k$ times specular reflection and $E^{(n+1)}$ represents the event that the boundary collisions are more than $n$ times and the first $n+1$ ones are precisely specular reflections. It should be notice that $E^{(n+1)}$ is not the end. $E^{(n+1)}$ itself involves an infinite series and we will use the coefficient $(1-\alpha)^{n+1}$ of $E^{(n+1)}$ to get the decay for refined estimate later.
Then we have for $d=2$, 
\begin{multline*}
j(\bfy,t)=\sum\limits_{k=0}^n\left\{\sum\limits_{l=1}^{k}\sum\limits_{k_1+k_2+ \cdots +k_l=l}^k
\int (1-\alpha)^{k-k_1-\ldots-k_l} j_{in}^{(k-k_1-\cdots -k_l)}(\bfy_{(k_1,\cdots, k_l)},t-s_1-\cdots -s_l)\right.\\
\left.\prod\limits_{i=1}^{l}(1-\alpha)^{k_i-1}\alpha G\BK{\phi_i,\frac{\sqrt{2RT(\bfy_{(k_1,\cdots, k_i)})}s_i}{k_i}}\frac{\sqrt{2RT(\bfy_{(k_1,\cdots, k_i)})}}{k_i}ds_i d\phi_i
+(1-\alpha)^k j_{in}^{(k)}(\bfy,t)\right\}\\
+\sum\limits_{l=1}^{n+1}\sum\limits_{k_1+k_2+ \cdots +k_l=l}^{n+1}
\int (1-\alpha)^{n+1-k_1-\ldots-k_l}E^{(n+1-k_1-\cdots -k_l)}(\bfy_{(k_1,\cdots, k_l)},t-s_1-\cdots -s_l)\\
\prod\limits_{i=1}^{l}(1-\alpha)^{k_i-1}\alpha G\BK{\phi_i,\frac{\sqrt{2RT(\bfy_{(k_1,\cdots, k_i)})}s_i}{k_i}}\frac{\sqrt{2RT(\bfy_{(k_1,\cdots, k_i)})}}{k_i}ds_i d\phi_i
+(1-\alpha)^{n+1} E^{(n+1)}(\bfy,t),
\end{multline*}
and for $d=1$,
\begin{multline*}
j(\bfy,t)=\sum\limits_{k=0}^n\left\{\sum\limits_{l=1}^{k}\sum\limits_{k_1+k_2+ \cdots +k_l=l}^k
\int (1-\alpha)^{k-k_1-\ldots-k_l} j_{in}^{(k-k_1-\cdots -k_l)}(\bfy_{(k_1,\cdots, k_l)},t-s_1-\cdots -s_l)\right.\\
\left.\prod\limits_{i=1}^{l}(1-\alpha)^{k_i-1}\alpha H\BK{\frac{\sqrt{2RT(\bfy_{(k_1,\cdots, k_i)})}s_i}{k_i}}\frac{\sqrt{2RT(\bfy_{(k_1,\cdots, k_i)})}}{k_i}ds_i
+(1-\alpha)^k j_{in}^{(k)}(\bfy,t)\right\}\\
+\sum\limits_{l=1}^{n+1}\sum\limits_{k_1+k_2+ \cdots +k_l=l}^{n+1}
\int (1-\alpha)^{n+1-k_1-\ldots-k_l}E^{(n+1-k_1-\cdots -k_l)}(\bfy_{(k_1,\cdots, k_l)},t-s_1-\cdots -s_l)\\
\prod\limits_{i=1}^{l}(1-\alpha)^{k_i-1}\alpha H\BK{\frac{\sqrt{2RT(\bfy_{(k_1,\cdots, k_i)})}s_i}{k_i}}\frac{\sqrt{2RT(\bfy_{(k_1,\cdots, k_i)})}}{k_i}ds_i
+(1-\alpha)^{n+1} E^{(n+1)}(\bfy,t).
\end{multline*}

Now we are ready to prove the global existence of the boundary flux function.
\begin{proof}{\textit{of Theorem \ref{thm:Fr:Existence:Bddness}}}\hfil

To compute the boundary flux $j(\bfy,t)$, we need to take all events into account.
In other words, we have to sum up all events for each boundary collision. Hence, we have the following infinite series due to the above discussion.
\begin{multline}\label{eq:fr:globa:1}
j(\bfy,t)=\sum\limits_{k=0}^{\infty}\Bigg\{(1-\alpha)^k j_{in}^{(k)}(\bfy,t)+
\sum\limits_{l=1}^{k}\sum\limits_{k_1+k_2+ \cdots +k_l=l}^k
(1-\alpha)^{k-k_1-\ldots-k_l}
\\ \times
\intLim_{0<s_1+\ldots+s_l<t} j_{in}^{(k-k_1-\cdots -k_l)}(\bfy_{(k_1,\cdots, k_l)},t-s_1-\cdots -s_l)\\
\times
\curBK{ \begin{array}{c@{,~}l}
                        \prod\limits_{i=1}^{l}(1-\alpha)^{k_i-1}\alpha G\BK{\phi_i,\frac{\sqrt{2RT(\bfy_{(k_1,\cdots, k_i)})}s_i}{k_i}}\frac{\sqrt{2RT(\bfy_{(k_1,\cdots, k_i)})}}{k_i}ds_i d\phi_i
                        &
                         d=2,
                        \\
                        \prod\limits_{i=1}^{l}(1-\alpha)^{k_i-1}\alpha H\BK{\frac{\sqrt{2RT(\bfy_{(k_1,\cdots, k_i)})}s_i}{k_i}}\frac{\sqrt{2RT(\bfy_{(k_1,\cdots, k_i)})}}{k_i}ds_i 
                        &
                         d=1
                    \end{array} }
\Bigg\}.
\end{multline}
The index $k$ here means the exact number of boundary collision and we will show the convergence of the series.   
From \eqref{eq:flux:initial} and noting that $|\bfxi^i|=|\bfxi|$ for each $i$, we have
\begin{equation}\label{eq:fr:globa:2}
\begin{split}
&\sum\limits_{k=0}^{\infty}(1-\alpha)^k j_{in}^{(k)}(\bfy,t)\\
&\leq\sum\limits_{k=0}^{\infty}
\int_{\sum\limits_{i=0}^{k-1}\frac{|\bfy_{(i)}-\bfy_{(i+1)}|}{t}<|\bfxi| <\sum\limits_{i=0}^{k}\frac{|\bfy_{(i)}-\bfy_{(i+1)}|}{t}}
\Big|-\bfxi\cdot\bfn(\bfy)\Big|\int\frac{\VertBK{g_{in}}_{\infty,\mu}}{(1+\bfzeta)^\mu}d\bfeta 
d\bfxi\\
&\leq
\VertBK{g_{in}}_{\infty,\mu}\int\frac{|-\bfxi\cdot\bfn(\bfy)|}{(1+\bfzeta)^\mu}d\bfzeta\\
&=O(1)\VertBK{g_{in}}_{\infty,\mu}
\end{split}
\end{equation}
With \eqref{eq:fr:globa:2}, we rewrite \eqref{eq:fr:globa:1} as
\begin{multline*}
j(\bfy,t)=O(1)\VertBK{g_{in}}_{\infty,\mu}+\sum\limits_{k=0}^{\infty}\Bigg\{
\sum\limits_{l=1}^{k}\sum\limits_{k_1+k_2+ \cdots +k_l=l}^k
(1-\alpha)^{k-k_1-\ldots-k_l}
\\ \times
\intLim_{0<s_1+\ldots+s_l<t} j_{in}^{(k-k_1-\cdots -k_l)}(\bfy_{(k_1,\cdots, k_l)},t-s_1-\cdots -s_l)\\
\times
\curBK{ \begin{array}{c@{,~}l}
                        \prod\limits_{i=1}^{l}(1-\alpha)^{k_i-1}\alpha G\BK{\phi_i,\frac{\sqrt{2RT(\bfy_{(k_1,\cdots, k_i)})}s_i}{k_i}}\frac{\sqrt{2RT(\bfy_{(k_1,\cdots, k_i)})}}{k_i}ds_i d\phi_i
                        &
                         d=2,
                        \\
                        \prod\limits_{i=1}^{l}(1-\alpha)^{k_i-1}\alpha H\BK{\frac{\sqrt{2RT(\bfy_{(k_1,\cdots, k_i)})}s_i}{k_i}}\frac{\sqrt{2RT(\bfy_{(k_1,\cdots, k_i)})}}{k_i}ds_i 
                        &
                         d=1
                    \end{array} }
\Bigg\}\\
\end{multline*}
\begin{multline}\label{eq:fr:globa:3}
=O(1)\VertBK{g_{in}}_{\infty,\mu}+\sum\limits_{l=1}^{\infty}(\frac{\alpha}{1-\alpha})^l
\sum\limits_{k_1=1}^{\infty}\ldots\sum\limits_{k_l=1}^{\infty}
\sum\limits_{k=k_1+\ldots+k_l}^{\infty}
(1-\alpha)^{k-k_1-\ldots-k_l}
\\ \times
\intLim_{0<s_1+\ldots+s_l<t} j_{in}^{(k-k_1-\cdots -k_l)}(\bfy_{(k_1,\cdots, k_l)},t-s_1-\cdots -s_l)\\
\times
\curBK{ \begin{array}{c@{,~}l}
                        \prod\limits_{i=1}^{l}(1-\alpha)^{k_i} G\BK{\phi_i,\frac{\sqrt{2RT(\bfy_{(k_1,\cdots, k_i)})}s_i}{k_i}}\frac{\sqrt{2RT(\bfy_{(k_1,\cdots, k_i)})}}{k_i}ds_i d\phi_i
                        &
                         d=2,
                        \\
                        \prod\limits_{i=1}^{l}(1-\alpha)^{k_i} H\BK{\frac{\sqrt{2RT(\bfy_{(k_1,\cdots, k_i)})}s_i}{k_i}}\frac{\sqrt{2RT(\bfy_{(k_1,\cdots, k_i)})}}{k_i}ds_i 
                        &
                         d=1
                    \end{array} }\\
                    \end{multline}
                    By plugging  \eqref{eq:fr:globa:2} into \eqref{eq:fr:globa:3} and direct computations, we have 
                    \begin{multline*}
     j(\bfy,t)=O(1)\VertBK{g_{in}}_{\infty,\mu}\Bigg\{1+  \sum\limits_{l=1}^{\infty}(\frac{\alpha}{1-\alpha})^l
\sum\limits_{k_1=1}^{\infty}\ldots\sum\limits_{k_l=1}^{\infty} 
\intLim_{0<s_1+\ldots+s_l<t}\\
          \times
\curBK{ \begin{array}{c@{,~}l}
                        \prod\limits_{i=1}^{l}(1-\alpha)^{k_i} G\BK{\phi_i,\frac{\sqrt{2RT(\bfy_{(k_1,\cdots, k_i)})}s_i}{k_i}}\frac{\sqrt{2RT(\bfy_{(k_1,\cdots, k_i)})}}{k_i}ds_i d\phi_i
                        &
                         d=2,
                        \\
                        \prod\limits_{i=1}^{l}(1-\alpha)^{k_i} H\BK{\frac{\sqrt{2RT(\bfy_{(k_1,\cdots, k_i)})}s_i}{k_i}}\frac{\sqrt{2RT(\bfy_{(k_1,\cdots, k_i)})}}{k_i}ds_i 
                        &
                         d=1
                    \end{array} }\\  
                    \end{multline*}
\begin{multline*}
                     =O(1)\VertBK{g_{in}}_{\infty,\mu}\Bigg\{1+  \sum\limits_{l=1}^{\infty}(\frac{\alpha}{1-\alpha})^l
\sum\limits_{k_1=1}^{\infty}\frac{(1-\alpha)^{k_1}}{k_1}\ldots\sum\limits_{k_l=1}^{\infty} \frac{(1-\alpha)^{k_l}}{k_l}\\
\intLim_{0<s_1+\ldots+s_l<t}
          \times
\prod\limits_{i=1}^{l}\curBK{ \begin{array}{c@{,~}l}
                        \BK{\pi\VertBK{G}_{L^\infty}\sqrt{2R\TM}}ds_i
                        &
                         d=2,
                        \\
                        \BK{\VertBK{H}_{L^\infty}\sqrt{2R\TM}}ds_i
                        &
                         d=1,
                    \end{array} }\\ 
                    \end{multline*}
\begin{multline*}
   =O(1)\VertBK{g_{in}}_{\infty,\mu}\Bigg\{1+  \sum\limits_{l=1}^{\infty}(\frac{\alpha}{1-\alpha})^l
\Bigg(\sum\limits_{k_1=1}^{\infty}\frac{(1-\alpha)^{k_1}}{k_1}\Bigg)^l
\curBK{ \begin{array}{c@{,~}l}
                       \frac{\BK{\pi\VertBK{G}_{L^\infty}\sqrt{2R\TM}t}^l}{l!}
                        &
                         d=2,
                        \\
                        \frac{\BK{\VertBK{H}_{L^\infty}\sqrt{2R\TM}t}^l}{l!}
                        &
                         d=1,
                    \end{array} }\\                   
\end{multline*}
It is easy to check that
\begin{equation*}
\sum\limits_{k_1=1}^{\infty}\frac{(1-\alpha)^{k_1}}{k_1}=-\ln\alpha,
\end{equation*}
and therefore we have

\begin{equation*}
j(\bfy,t)=O(1)\VertBK{g_{in}}_{ \infty, \mu }
\curBK{ \begin{array}{c@{,~}l}
                        e^{\frac{-\alpha\ln\alpha}{1-\alpha}\pi\VertBK{G}_{L^\infty}\sqrt{2R\TM}t}
                        &
                         d=2,
                        \\
                       e^{\frac{-\alpha\ln\alpha}{1-\alpha}\VertBK{H}_{L^\infty}\sqrt{2R\TM}t}
                        &
                         d=1
                    \end{array} }.
\end{equation*}
\end{proof}
\end{subsection}

\begin{subsection}{Preliminary Estimates}
The discussion of this subsection applies to all space dimension. To avoid complication in notations,
we treat only the $2d$ case. And we will use the following Law of Large Numbers to get a refined estimate. It can be proved by the similar argument as in \cite{Kuo-Liu-Tsai}. Therefore we omit it.
\begin{theorem}[Law of Large Numbers]\label{thm:Fr:Law:of:Large}
There exists some constant $C>0$ such that, for any $\gamma$ and $m$ with $\gamma/(mn)^{\frac{1}{d+1}} > C$,
\begin{equation*}
    \intLim_{ \frac{\gamma}{m} < |\sigma-n\E(X_1)| } \hspace{-20pt} H_n(\sigma) d\sigma
=
    \Prob \Big\{ \frac{\gamma}{m} < |X_1+\ldots+X_n-n\E(X_1)| \Big\}
=
     O(1) \frac{ m^{d+1}n^d \log(\gamma+1) }{ \gamma^{d+1} }.
\end{equation*}
where $X_1,X_2,\ldots,$ are i.i.d. with the probability density function
\begin{equation*}
H(\sigma)=
\begin{dcases}
\BK{ \frac{2}{\sigma} }^{3} e^{ -\BK{\frac{2}{\sigma}}^2 }\quad \text{if } d=1,\\
\int G(\phi,\sigma)d\phi \quad \text{if } d=2,
\end{dcases}
\end{equation*}
and $H_n$ is the probability density function of the sum of i.i.d. random variables, $X_1 +X_2 + \cdots + X_n$. Notice that $H_n$ is a convolution of $H$ itself, $H_n=\underbrace{ \BK{ H*\cdots*H } }_{n\text{ times }}$.
\end{theorem}

In order to get a refined estimate, we start with the specular reflection:
\begin{multline}\label{eq:flux:specular}
j(\bfy,t)=\sum\limits_{k=0}^{m-1}(1-\alpha)^k j_{in}^{(k)}(\bfy,t)\\
+\alpha\sum\limits_{k=1}^{m}(1-\alpha)^{k-1}
\intLim_{-\pi/2}^{\pi/2}\intLim_{0}^{t}
j(\bfy_{(k)},t-s)G(\phi,\frac{s\sqrt{2RT(\bfy_{(k)})}}{k})\frac{\sqrt{2RT(\bfy_{(k)})}}{k} ds d\phi\\
+(1-\alpha)^{m}E^{(m)}(\bfy,t) \quad \text{when } d=2,
\end{multline}
\begin{multline*}
j(\pm1,t)=\sum\limits_{k=0}^{m-1}(1-\alpha)^k j_{in}^{(k)}(\pm1,t)\\
+\alpha\sum\limits_{k=1}^{m}(1-\alpha)^{k-1}
\intLim_{0}^{t}
j(\pm(-1)^k,t-s)H(\frac{s\sqrt{2RT(\pm(-1)^k)}}{k})\frac{\sqrt{2RT(\pm(-1)^k)}}{k} ds \\
+(1-\alpha)^{m}E^{(m)}(\pm1,t) \quad \text{when } d=1.
\end{multline*}
Define
\begin{equation*}
J_{in}^{(m)}(\bfy,t)=\sum\limits_{k=0}^{m-1}(1-\alpha)^k j_{in}^{(k)}(\bfy,t),
\end{equation*}
and rewrite \eqref{eq:flux:specular} as:
\begin{multline}\label{eq:flux:specular:2}
j(\bfy,t)=J_{in}^{(m)}(\bfy,t)+(1-\alpha)^{m}E^{(m)}(\bfy,t)\\
+\alpha\sum\limits_{k=1}^{m}(1-\alpha)^{k-1}
\intLim_{-\pi/2}^{\pi/2}\intLim_{0}^{t}
j(\bfy_{(k)},t-s)G(\phi,\frac{s\sqrt{2RT(\bfy_{(k)})}}{k})\frac{\sqrt{2RT(\bfy_{(k)})}}{k} ds d\phi\\
\end{multline}
By iterating \eqref{eq:flux:specular:2} $n$ times, we have
\begin{multline}\label{eq:flux:specular:ntimes}
j(\bfy,t)=J_{in}^{(m)}(\bfy,t)+(1-\alpha)^{m}E^{(m)}(\bfy,t)\\
+\sum\limits_{i=1}^{n-1}\alpha^i\sum\limits_{k_1=1}^{m}(1-\alpha)^{k_1-1}\cdots
\sum\limits_{k_i=1}^{m}(1-\alpha)^{k_i-1}
\intLim_{0<s_1+s_2+\cdots +s_i<t}\\
\prod\limits_{j=1}^{i} G(\phi_j,\frac{s_j\sqrt{2RT(\bfy_{(k_1,\cdots,k_j)})}}{k_j})\frac{\sqrt{2RT(\bfy_{(k_1,\cdots,k_j)})}}{k_j}\\
\curBK{J_{in}^{(m)}(\bfy_{(k_1,\cdots, k_i)},t-s_1-\cdots -s_i)+(1-\alpha)^{m}E^{(m)}(\bfy_{(k_1,\cdots, k_i)},t-s_1-\cdots -s_i)}ds_i d\phi_i \cdots ds_1 d\phi_1\\
+\alpha^n\sum\limits_{k_1=1}^{m}(1-\alpha)^{k_1-1}\cdots\sum\limits_{k_n=1}^{m}(1-\alpha)^{k_n-1}
\intLim_{0<s_1+s_2+\cdots +s_n<t}j(\bfy_{(k_1,k_2,\cdots, k_n)},t-s_1-s_2-\cdots -s_n)\\
\prod\limits_{i=1}^{n} G(\phi_i,\frac{s_i\sqrt{2RT(\bfy_{(k_1,\cdots,k_i)})}}{k_i})
\frac{\sqrt{2RT(\bfy_{(k_1,\cdots,k_i)})}}{k_i} ds_n d\phi_n \cdots ds_1 d\phi_1.
\end{multline}

To estimate \eqref{eq:flux:specular:ntimes}, we define
\begin{multline*}
j_{(k_1,\cdots,k_i)}^{(i,m)}(\bfy,t)=
\intLim_{0<s_1+s_2+\cdots +s_i<t}
\prod\limits_{j=1}^{i} G(\phi_j,\frac{s_j\sqrt{2RT(\bfy_{(k_1,\cdots,k_j)})}}{k_j})\frac{\sqrt{2RT(\bfy_{(k_1,\cdots,k_j)})}}{k_j}\\
\curBK{J_{in}^{(m)}(\bfy_{(k_1,\cdots, k_i)},t-s_1-\cdots -s_i)+(1-\alpha)^{m}E^{(m)}(\bfy_{(k_1,\cdots, k_i)},t-s_1-\cdots -s_i)}ds_i d\phi_i \cdots ds_1 d\phi_1\\
\end{multline*}
and
\begin{multline*}
J_{(k_1,\cdots,k_n)}^{(n,m)}(\bfy,t)=
\intLim_{0<s_1+s_2+\cdots +s_n<t}j(\bfy_{(k_1,k_2,\cdots, k_n)},t-s_1-s_2-\cdots -s_n)\\
\prod\limits_{i=1}^{n} G(\phi_i,\frac{s_i\sqrt{2RT(\bfy_{(k_1,\cdots,k_i)})}}{k_i})
\frac{\sqrt{2RT(\bfy_{(k_1,\cdots,k_i)})}}{k_i} ds_n d\phi_n \cdots ds_1 d\phi_1.
\end{multline*}
Recall
\begin{equation*}
\begin{split}
J_{in}^{(m)}(\bfy,t)&=\sum\limits_{k=0}^{m-1}(1-\alpha)^k j_{in}^{(k)}(\bfy,t)\\
&=O(1)\sum\limits_{k=0}^{m-1}(1-\alpha)^k\VertBK{g_{in}}_{\infty,\mu}\\
&=O(1)m \VertBK{g_{in}}_{\infty,\mu}.
\end{split}
\end{equation*}
Moreover, we have for $t>1$,
\begin{equation*}
\begin{split}
j_{in}^{(k)}(\bfy,t)&=
\intLim_{\frac{k|\bfy-\bfy_{(1)}|}{t}<|\bfxi|<\frac{(k+1)|\bfy-\bfy_{(1)}|}{t}} 
-\bfxi\cdot\bfn(\bfy)\bar g_{in}(\bfy_{(k)}-\bfxi^k(t-kt_1),\bfxi^k)d\bfxi\\
&=O(1)\VertBK{g_{in}}_{\infty,\mu} \intLim_{\frac{k|\bfy-\bfy_{(1)}|}{t}<|\bfxi|<\frac{(k+1)|\bfy-\bfy_{(1)}|}{t}} 
|\bfxi|^d d|\bfxi|\\
&=O(1)\VertBK{g_{in}}_{\infty,\mu} \frac{\BK{(k+1)|\bfy-\bfy_{(1)}|}^{d+1}-\BK{k|\bfy-\bfy_{(1)}|}^{d+1}}{t^{d+1}}\\
&=O(1)\VertBK{g_{in}}_{\infty,\mu} 
\begin{dcases}
\frac{2k+1}{t^2} \quad \text{ for } d=1,\\
\frac{3k^2+3k+1}{t^3} \quad \text{ for } d=2.
\end{dcases}
\end{split}
\end{equation*}
Thus we have
\begin{equation}\label{eq:spec:in}
\begin{split}
J^{(m)}_{in}(\bfy,t)&=O(1)\VertBK{g_{in}}_{\infty,\mu} \sum\limits_{k=0}^{m-1} (1-\alpha)^k
\begin{dcases}
\frac{2k+1}{t^2} \quad \text{ for } d=1,\\
\frac{3k^2+3k+1}{t^3} \quad \text{ for } d=2,
\end{dcases}\\
&=O(1)\VertBK{g_{in}}_{\infty,\mu}\frac{1+(1-\alpha)m^{d+1}}{t^{d+1}}.
\end{split}
\end{equation}
In order to estimate the remainder terms, we define a priori bound of boundary flux $j(\bfy,t)$.
\begin{definition}
Define the a priori bound $\calJ$ by
\begin{equation*}
    \calJ(t) \equiv \sup_{0\leq s\leq t}
    \begin{dcases}
        \BK{\VertBK{j}_{L^\infty_\bfy} }(s) &   \text{ for }    d=2, \\
        \Big(\absBK{j_+(s)} + \absBK{j_-(s)}\Big)     &   \text{ for }    d=1,
    \end{dcases}
\end{equation*}
where $j_\pm(t) \equiv j(\pm 1,t)$.
\end{definition}
From \eqref{eq:Fr:Reduced:Chara:Rep1} and \eqref{eq:flux:E}, it is easy to show
\begin{equation*}
E^{(m)}(\bfy,t)=O(1)\BK{\calJ(t)+\VertBK{g_{in}}_{\infty,\mu}}.
\end{equation*}

For $i<n$, we divide $j^{(i,m)}_{(k_1,\cdots,k_i)}(\bfy,t)$ into the following two parts:
\begin{align*}
&
     j^{(i,m)}_{(k_1,\cdots,k_i)}(\bfy,t)
    \\
=&
    \BK{ \int_{\mathscr{A}_1} + \int_{ \mathscr{A}_2} }
    \left\{J_{in}^{(m)}
    \BKK{
            \bfy_{(k_1,\cdots,k_i)},
                    t - \frac{k_1\sigma_1}{\sqrt{2RT_{(1)}}} - \ldots
                    - \frac{k_i\sigma_i}{\sqrt{2RT_{(k_1,\ldots,k_i)}}}
        }\right.
        \\
&
        \left.+(1-\alpha)^{m}E^{(m)}
        \BKK{
            \bfy_{(k_1,\cdots,k_i)},
                    t - \frac{k_1\sigma_1}{\sqrt{2RT_{(1)}}} - \ldots
                    - \frac{k_i\sigma_i}{\sqrt{2RT_{(k_1,\ldots,k_i)}}}
        }\right\}
        \times
        \prod_{l=1}^i
    \BK{
             G(\phi_l,\sigma_l) d\sigma_l d\phi_l
       }
    \\
\equiv&
    j^{(i,m)slow}_{(k_1,\cdots,k_i)}(\bfy,t) + j^{(i,m)rare}_{(k_1,\cdots,k_i)}(\bfy,t),
\end{align*}
where
\begin{equation*}
\begin{split}
        \mathscr{A}_1
\equiv &
        \curBK{
                    0 < \frac{ k_1\sigma_1 }{ \sqrt{2RT_{(k_1)} }} + \ldots
                    + \frac{ k_i\sigma_i }{ \sqrt{2RT_{(k_1,\ldots,k_i)} }} < \Tr\frac{t}{2}
        		},
\\
        \mathscr{A}_2
\equiv &
        \curBK{
                    \Tr\frac{t}{2} < \frac{ k_1\sigma_1 }{ \sqrt{2RT_{(k_1)}} } + \ldots
                    + \frac{ k_i\sigma_i }{ \sqrt{2RT_{(k_1,\ldots,k_i)}} } < t
        		}.
\end{split}
\end{equation*}

For $j^{(i,m)slow}_{(k_1,\cdots,k_i)}(\bfy,t)$, the time needed to trace back to an interior point is at least $t/2$:
\begin{equation*}
\begin{split}
&
	\frac{ k_1\sigma_1 }{ \sqrt{2RT_{(k_1)}} } 
	+ 
	\ldots
	+ 
	\frac{ k_i\sigma_i }{ \sqrt{2RT_{(k_1,\ldots,k_i)} }} 
	< \Tr\frac{t}{2} 
	\\
&
	\Longrightarrow\quad
	t - 
	\BK{	
		\frac{ k_1\sigma_1 }{ \sqrt{2RT_{(k_1)} }} 
		+ 
		\ldots
		+ 
		\frac{ k_i\sigma_i }{ \sqrt{2RT_{(k1,\ldots,k_i)}} } 
		}
	> t - \Tr\frac{t}{2}  \geq \frac{t}{2}.
\end{split}
\end{equation*}

Thus
\begin{equation}\label{eq:Fr:j:ith:Slow:Esti:23D}
\begin{split}
&
     j^{(i,m)slow}_{(k_1,\cdots,k_i)}(\bfy,t)
    \\
\leq&    
    \intLim_{\mathscr{A}_1}
   \sup_{\frac{t}{2}<s<t} 
    \BK{\VertBK{ J_{in}^{(m)} }_{L^\infty_\bfy}+(1-\alpha)^m\VertBK{ E^{(m)} }_{L^\infty_\bfy}}(s)
    \times
    \int
           \prod_{l=1}^i  G(\phi_l,\sigma_l)d\phi_l d\sigma_l
    \\
=&
    O(1) \BK{\VertBK{g_{in}}_{\infty,\mu}\frac{1+(1-\alpha)m^{d+1}}{t^{d+1}}
    +(1-\alpha)^m\BK{\VertBK{g_{in}}_{\infty,\mu}+\calJ(t)}}.
\end{split}
\end{equation}
Note that the estimate \eqref{eq:Fr:j:ith:Slow:Esti:23D} relies merely on the smallness of the speed.

For $j^{(i,m)rare}_{(k_1,\cdots,k_i)}(\bfy,t)$, the time consumed to trace back is at least $\Tr t/2$. Therefore,
\begin{equation*}
        \Tr\frac{t}{2m}\sqrt{2R\Tm} < \sigma_1 + \ldots + \sigma_i,
\end{equation*}
\begin{align}
&
		\notag
	 	 j^{(i,m)rare}_{(k_1,\cdots,k_i)}(\bfy,t)
	 	\\
	 	\notag
\leq&
		O(1) \BK{m \VertBK{g_{in}}_{\infty,\mu} 
    +(1-\alpha)^m\BK{\VertBK{g_{in}}_{\infty,\mu}+\calJ(t)}} 
         \\
         \notag    
    \quad&\times
		\intLim_{  \Tr\frac{t}{2m} \sqrt{2R\Tm} < \sigma_1 + \ldots + \sigma_i  }
				\prod_{l=1}^i  G(\phi_l,\sigma_l)d\phi_l d\sigma_l
        \\
        \notag 
=&
		O(1) \BK{m \VertBK{g_{in}}_{\infty,\mu} 
    +\calJ(t)}
		\times 
		\int_{ \Tr\frac{t}{2m}\sqrt{2R\Tm} }^{\infty} 
		 H_i(\sigma) d\sigma
		\\
        \label{eq:z23} 
=&
		O(1) \BK{m \VertBK{g_{in}}_{\infty,\mu} 
    +\calJ(t)}
		\times 
		\Pr \curBK{ X_1+\ldots+X_i > \Tr\frac{t}{2m}\sqrt{2R\Tm} }.
\end{align}

Note that $n,m$, the index of $J^{(n,m)}_{(k_1,\cdots,k_i)}(\bfy,t)$, are  variables at our disposal. {\it Throughout this paper we assume $t/mn \gg 1$}. Recall that our final choice of $n$ and $m$ is $mn=\lfloor t^r \rfloor$, $r\in (0,(d+1)^{-1})$. Thus $r<1$ and so, for large $t$, $mn\ll t.$ Since $\E(X_1+\ldots+X_i)=i\E(X_1) \sim i \leq n  \ll \Tr\frac{t}{2m}\sqrt{2R\Tm}$, \eqref{eq:z23} represents the probability of a rare event. We now apply the law of large numbers, Theorem \ref{thm:Fr:Law:of:Large}, to estimate \eqref{eq:z23}: choose the truncation variable $\gamma$ to be $\sqrt{2R\Tm} \frac{t}{3}$. Under the assumption $t/mi > t/mn \gg 1$, we have
\begin{equation*}
	\curBK{ \Tr\sqrt{2R\Tm}\frac{t}{2m} < \sigma } \subset \{\gamma<|\sigma-i\E(X_1)| \},
	\quad
	\frac{ \gamma }{ (mn)^{\frac{1}{d+1}} } \sim \frac{ t }{ (mn)^{\frac{1}{d+1}} } \gg 1.
\end{equation*}
Therefore, we can apply Theorem \ref{thm:Fr:Law:of:Large} to obtain
\begin{equation}\label{eq:Fr:j:ith:Rare:Esti:23D}
    j^{(i,m)rare}_{(k_1,\cdots,k_i)}(\bfy,t) = O(1) \BK{m \VertBK{g_{in}}_{\infty,\mu} 
    +\calJ(t)} \frac{ m^{d+1}i^d \log(t+1) }{ t^{d+1} }.
\end{equation}
Note that the estimate \eqref{eq:Fr:j:ith:Rare:Esti:23D} relies merely on the law of large numbers.

For $J^{(n,m)}_{(k_1,\ldots,k_n)}$, we conduct a similar decomposition:
\begin{align*}
     J^{(n,m)}_{(k_1,\ldots,k_n)}(\bfy,t)
=&
    \BK{ \int_{\mathscr{B}_1} + \int_{ \mathscr{B}_2} }
  	 j
    \BKK{
            \bfy^{(k_1,\ldots,k_n)},
             t - \frac{k_1\sigma_1}{\sqrt{2RT_{(k_1)}}} - \ldots
               - \frac{k_n\sigma_n}{\sqrt{2RT_{(k_1,\ldots,k_n)}}}
        }
        \\
&
        \times
        \prod_{l=1}^n
             G(\phi_l,\sigma_l)d\phi_l d\sigma_l
    \\
\equiv&
   \Lambda_{(k_1,\ldots,k_n)}^{(n,m)}(\bfy,t) + J^{(n,m)rare}_{(k_1,\cdots,k_i)}(\bfy,t),
\end{align*}
where
\begin{equation*}
\begin{split}
        \mathscr{B}_1
\equiv&
        \curBK{
                    0 < \frac{ k_1\sigma_1 }{ \sqrt{2RT_{(k_1)} }} + \ldots
                    + \frac{ k_n\sigma_n }{ \sqrt{2RT_{(k_1,\ldots,k_n)}} } < \Tr\frac{t}{2} 
        		}
        \\
        \mathscr{B}_2
\equiv&
        \curBK{
                    \Tr\frac{t}{2} < \frac{ k_1\sigma_1 }{ \sqrt{2RT_{(k_1)}} } + \ldots
                    + \frac{ k_1\sigma_n }{ \sqrt{2RT_{(k_1,\ldots,k_n)}} } < t
        		}.
\end{split}
\end{equation*}
We can apply the same argument of $j^{(n,m)rare}_{(k_1,\cdots,k_i)}$ to $J^{(n,m)rare}_{(k_1,\cdots,k_i)}$ to obtain:
\begin{equation}\label{eq:Fr:Jn:Rare:Esti:23D}
    J^{(n,m)rare}_{(k_1,\cdots,k_i)}(\bfy,t) = O(1) \calJ(t) \frac{ m^{d+1} n^d \log(t+1) }{ t^{d+1} },
    \text{ whenever } t/mn \gg 1.
\end{equation}
We omit the details.
\begin{lemma}\label{sum:event}
\begin{equation*}
\begin{split}
&\sum\limits_{k_1=0}^{m-1}\cdots\sum\limits_{k_i=0}^{m-1}(1-\alpha)^{k_1+\ldots+k_i}
=\BK{\frac{1-(1-\alpha)^{m}}{\alpha}}^i\\
&\sum\limits_{i=1}^{n-1}\alpha^i\sum\limits_{k_1=0}^{m-1}\cdots\sum\limits_{k_i=0}^{m-1}(1-\alpha)^{k_1+\ldots+k_i}
=\sum\limits_{i=1}^{n-1}\BK{1-(1-\alpha)^{m}}^i=O(1)\alpha n
\end{split}
\end{equation*}
\end{lemma}
Since
\begin{equation*}
\begin{split}
	 j(\bfy,t) 
=& 
    J^{(m)}_{in}(\bfy,t)+(1-\alpha)^m E^{(m)}(\bfy,t)
    \\
+&    	
	\sum_{i=1}^{n-1}\alpha^i \sum_{k_1=1}^m (1-\alpha)^{k_1-1}\cdots \sum_{k_i=1}^m (1-\alpha)^{k_i-1}
	j^{(i,m)}_{(k_1,\ldots,k_i)}(\bfy,t) 
	\\
+& \alpha^n \sum_{k_1=1}^m (1-\alpha)^{k_1-1}\cdots \sum_{k_n=1}^m (1-\alpha)^{k_n-1}
    J^{(n,m)}_{(k_1,\ldots,k_n)}(\bfy,t)
	\\
\end{split}
\end{equation*}
putting \eqref{eq:Fr:j:ith:Slow:Esti:23D}, \eqref{eq:Fr:j:ith:Rare:Esti:23D}, and \eqref{eq:Fr:Jn:Rare:Esti:23D} together we have:
\begin{theorem}\label{thm:Fr:Integral:Eqn}
For $t/mn \gg 1$,
\begin{align*}
      j(\bfy,t)
=&
      O(1) \BK{\frac{ m(mn)^{d+1}\log(t+1) }{ t^{d+1} }+(1-\alpha)^m}
      \VertBK{ g_{in} }_{\infty,\mu}  \\
     +&
     O(1) \BK{\frac{ (mn)^{d+1}\log(t+1) }{ t^{d+1} }+(1-\alpha)^m}
      \calJ(t)\\      
    +&\alpha^n \sum_{k_1=1}^m (1-\alpha)^{k_1-1}\cdots \sum_{k_n=1}^m (1-\alpha)^{k_n-1}
    \Lambda_{(k_1,\ldots,k_n)}^{(n,m)}(\bfy,t), 
\end{align*}
where
\begin{subequations}\label{eq:Fr:Defn:Lambda:23D}
\begin{multline}\label{eq:Fr:Defn:Lambda:s:23D}
    \Lambda_{(k_1,\ldots,k_n)}^{(n,m)}(\bfy,t)
\equiv
    \intLim_{ 0 < s_1 + \ldots + s_n < \Tr \frac{t}{2} }
 j\BK{  \bfy_{(k_1,\ldots,k_n)}, t-s_1-\ldots-s_n }
    \\
    \times
     \BK{
            \prod_{l=1}^n G \BKK{ \phi_l, \frac{s_l \sqrt{2RT(\bfy_{(k_1,\ldots,k_l)})}}{k_l} }
            \frac{\sqrt{2RT(\bfy_{(k_1,\ldots,k_l)})}}{k_l}d\phi_l d s_l
        }
        \quad \text{ when } d=2,
\end{multline}
\begin{multline}\label{eq:Fr:Defn:Lambda:s:1D}
    \Lambda_{(k_1,\ldots,k_n)}^{(n,m)}(\pm1,t)
\equiv
    \intLim_{ 0 < s_1 + \ldots + s_n < \Tr \frac{t}{2} }
 j\BK{  \pm1_{(k_1,\ldots,k_n)}, t-s_1-\ldots-s_n }
    \\
    \times
     \BK{
            \prod_{l=1}^n H \BKK{  \frac{s_l \sqrt{2RT(\pm1_{(k_1,\ldots,k_l)})}}{k_l} }
            \frac{\sqrt{2RT(\pm1_{(k_1,\ldots,k_l)})}}{k_l} d s_l
        }
        \quad \text{ when } d=1,
\end{multline}
\end{subequations}
\end{theorem}

With the aid of the law of large numbers, we have estimated $j_{(k_1,\ldots,k_n)}^{(n,m)}$ and $J_{(k_1,\ldots,k_n)}^{(n,m)rare}$. The remaining term $\Lambda_{(k_1,\ldots,k_n)}^{(n,m)}$ consists of the main event, which requires more effort to estimate. 

To show the convergence of boundary flux $j$, we need to use the crucial conservation of molecular number, \eqref{eq:Fr:Zero:Total:Mass} and \eqref{eq:Fr:Reduced:Chara:Rep}, as
\begin{equation}\label{eq:Fr:Consv:of:Mole:Numb}
\begin{split}
      j(\bfy,t)
&=
      j(\bfy,t)  - \frac{1}{C_S |D|} \int_{D\times\bbR^d  } \bar g(\bfx,\bfxi,t) ~ d\bfx d\bfxi
     \\
&=
     \frac{1}{|D|} \int_{D\times\bbR^d}
     \Big(j(\bfy,t)\boldsymbol{s}(\bfx,\bfxi) - \frac{1}{C_S}\bar g(\bfx,\bfxi,t) \Big) 
      d\bfx d\bfxi\\
&=
     \frac{1}{|D|} \int_{|\bfxi|<\frac{|\bfx-\bfx_{(1)}|}{t}}
     \Big(j(\bfy,t)\boldsymbol{s}(\bfx,\bfxi) - \frac{1}{C_S}\bar g(\bfx,\bfxi,t) \Big) 
      d\bfx d\bfxi  \\
&+
     \frac{1}{|D|} \int_{\frac{|\bfx-\bfx_{(1)}|}{t}<|\bfxi|<\frac{|\bfx_{(1)}-\bfx_{(2)}|}{\log (t+1)}}
     \Big(j(\bfy,t)\boldsymbol{s}(\bfx,\bfxi) - \frac{1}{C_S}\bar g(\bfx,\bfxi,t) \Big) 
      d\bfx d\bfxi \\
&+
     \frac{1}{|D|} \int_{|\bfxi|>\frac{|\bfx_{(1)}-\bfx_{(2)}|}{\log (t+1)}}
     \Big(j(\bfy,t)\boldsymbol{s}(\bfx,\bfxi) - \frac{1}{C_S}\bar g(\bfx,\bfxi,t) \Big) 
      d\bfx d\bfxi   \\
&\equiv j_{in}(\bfy,t) + j_{mid}(\bfy,t) + j_{fl}(\bfy,t) .      
\end{split}
\end{equation}
It is easy to see that if we choose $K=K(\bfx,\bfxi,t)$ such that
\begin{equation*}
K-1<\frac{t}{\log(t+1)}-\frac{|\bfx-\bfx_{(1)}|}{|\bfx_{(1)}-\bfx_{(2)}|}<K,
\end{equation*}
then we have
\begin{equation}\label{eq:K:bdd}
\frac{|\bfx-\bfx_{(1)}|+(K-1)|\bfx_{(1)}-\bfx_{(2)}|}{t}<\frac{|\bfx_{(1)}-\bfx_{(2)}|}{\log (t+1)}<\frac{|\bfx-\bfx_{(1)}|+K|\bfx_{(1)}-\bfx_{(2)}|}{t}.
\end{equation}
Since the domain is spherically symmetric, it is easy to show that
\begin{equation*}
\frac{|\bfx-\bfx_{(1)}|}{|\bfx_{(1)}-\bfx_{(2)}|}\leq1,
\end{equation*}
and therefore $K\approx\frac{t}{\log(t+1)}$.
With \eqref{eq:K:bdd}, \eqref{eq:Fr:Reduced:Chara:Rep1}, and \eqref{eq::specular:m}, we have
\begin{multline}\label{eq:jin}
j_{in}(\bfy,t)=
\frac{1}{C_S|D|} \int_{|\bfxi|<\frac{|\bfx-\bfx_{(1)}|}{t}}
	 \left\{
	 		\alpha\sum\limits_{i=1}^{\infty}(1-\alpha)^{i-1}j(\bfy,t)
	 		\BK{ \frac{2\pi}{RT(\bfx_{(i)})} }^{\frac{1}{2}} M_{T(\bfx_{(i)})}(\bfxi) \right.\\
	 		- \bar g_{in}(\bfx-\bfxi t,\bfxi)\Bigg\}
      d\bfx d\bfxi,\\
\end{multline}
\begin{multline}\label{eq:jmid}
j_{mid}(\bfy,t)\leq
 \frac{1}{C_S|D|} 
      \sum\limits_{k=1}^{K}
      \int_{A_k} 
             \Bigg\{\alpha\sum\limits_{i=1}^{k}(1-\alpha)^{i-1}\sqBK{j(\bfy,t)-j(\bfx_{(i)},t-t_1-...-t_i)}
             \\
             \times
             \BK{\frac{2\pi}{RT(\bfx_{(i)})}}^{\frac{1}{2}}M_{T(\bfx_{(i)})}(\bfxi)
             +(1-\alpha)^k\Bigg(\alpha\sum\limits_{i=1}^{\infty}(1-\alpha)^{i-1}j(\bfy,t)
	 		\BK{ \frac{2\pi}{RT(\bfx_{(k+i)})} }^{\frac{1}{2}}\\
	 		\times M_{T(\bfx_{(k+i)})}(\bfxi)
             - \bar g_{in}(\bfx_{(k)}-\bfxi^k(t-t_1-...-t_k),\bfxi^k)\Bigg)\Bigg\}d\bfx d\bfxi,\\
\end{multline}
\begin{multline}\label{eq:jfl}
j_{fl}(\bfy,t)=
      \frac{1}{C_S|D|} 
      \int_{|\bfxi|>\frac{|\bfx_{(1)}-\bfx_{(2)}|}{\log (t+1)}} 
             \Bigg\{\alpha\sum\limits_{i=1}^{K}(1-\alpha)^{i-1}
             \sqBK{j(\bfy,t)-j(\bfx_{(i)},t-t_1-...-t_i)}\\
             \BK{\frac{2\pi}{RT(\bfx_{(i)})}}^{\frac{1}{2}}M_{T(\bfx_{(i)})}(\bfxi)
             +(1-\alpha)^K\Bigg(\alpha\sum\limits_{i=1}^{\infty}(1-\alpha)^{i-1}j(\bfy,t)
	 		\BK{ \frac{2\pi}{RT(\bfx_{(K+i)})} }^{\frac{1}{2}}
	 		\\
	 		 M_{T(\bfx_{(K+i)})}(\bfxi)
             - \bar g(\bfx_{(K)}-\bfxi^K(t-t_1-...-t_K),\bfxi^K)\Bigg)\Bigg\}d\bfx d\bfxi,\\
             \end{multline}
where
\begin{equation*}
A_k=\curBK{\frac{|\bfx-\bfx_{(1)}|+(k-1)|\bfx_{(1)}-\bfx_{(2)}|}{t}<|\bfxi|<\frac{|\bfx-\bfx_{(1)}|+k|\bfx_{(1)}-\bfx_{(2)}|}{t}}.
\end{equation*}            
Each component of $j$ can be estimated in terms of  $\calJ(t)$ and the fluctuation of $j$. 
Therefore, it suffices to consider the fluctuation of $j$. For $t'<t$, $mn/t'\ll 1$, Theorem \ref{thm:Fr:Integral:Eqn} yields
\begin{multline}\label{eq:Fr:Flu:Lambda}
      j(\bfy,t) - j(\bfy',t')
=
      O(1) \BK{\frac{ m(mn)^{d+1}\log(t'+1) }{  t'^{d+1} }+(1-\alpha)^m}
      \VertBK{ g_{in} }_{\infty,\mu} 
\\
+O(1) \BK{\frac{ (mn)^{d+1}\log(t'+1) }{ t'^{d+1} }+(1-\alpha)^m}
       \calJ(t) \\
    +
     \alpha^n \sum_{k_1=1}^m (1-\alpha)^{k_1-1}\cdots \sum_{k_n=1}^m (1-\alpha)^{k_n-1}
     \Big( \Lambda_{(k_1,\ldots,k_n)}^{(n,m)}(\bfy,t) - \Lambda_{(k_1,\ldots,k_n)}^{(n,m)}(\bfy',t') \Big). 
\end{multline}
And note that $\alpha^n \sum_{k_1=1}^m (1-\alpha)^{k_1-1}\cdots \sum_{k_n=1}^m (1-\alpha)^{k_n-1}\leq1$.
Therefore, we need only to estimate the fluctuation of $\Lambda_{(k_1,\ldots,k_n)}^{(n,m)}$ and show that they are uniform for each $(k_1,\ldots,k_n)$. Since
\begin{multline}\label{eq:Fr:Flu:Lambda-sep}
	\Lambda_{(k_1,\ldots,k_n)}^{(n,m)}(\bfy,t) - \Lambda_{(k_1,\ldots,k_n)}^{(n,m)}(\bfy',t')
\\=
	\Big( \Lambda_{(k_1,\ldots,k_n)}^{(n,m)}(\bfy,t) - \Lambda_{(k_1,\ldots,k_n)}^{(n,m)}(\bfy,t') \Big)
	+
	\Big( \Lambda_{(k_1,\ldots,k_n)}^{(n,m)}(\bfy,t') - \Lambda_{(k_1,\ldots,k_n)}^{(n,m)}(\bfy',t') \Big),
\end{multline}
we may consider temporal and spacial fluctuation separately. We study the temporal fluctuation in Section \ref{sec:Fr:Temp:Fluc:Esti} and the spacial fluctuation in Section \ref{sec:Fr:Spac:Fluc}.
\begin{remark}
For the Maxwell-type boundary condition, which is a convex combination of the specular reflection condition and the diffuse reflection condition, the intricate dependence on the accommodation coefficient $\alpha$ yields serious analytical difficulties beyond those in \cite{Kuo-Liu-Tsai} and \cite{Kuo-Liu-Tsai-2}. One needs to consider all the events of particle colliding with the boundary many times, \eqref{eq:fr:globa:1}. In that case we need to deal with the problem that the diffuse reflections are coupled with specular reflections. Roughly speaking, for the events that specular reflections are more than diffuse reflections,
we need a new idea to obtain the decay rate even if the specular reflection itself doesn't have the equilibrating effect. We achieve the aim through \eqref{eq:spec:in} and \eqref{eq:Fr:j:ith:Slow:Esti:23D}. For the events that diffusion reflections are more than specular reflections, we may modify the analysis from the previous works \cite{Kuo-Liu-Tsai} and \cite{Kuo-Liu-Tsai-2} to get the decay rate, \eqref{eq:Fr:Jn:Rare:Esti:23D}. However, the appearance of specular reflection will slow down the decay rate. Finally, we succeed in combining all events via Theorem \ref{thm:Fr:Integral:Eqn}.
\end{remark}
\end{subsection}

\begin{subsection}{Temporal Fluctuation estimate}\label{sec:Fr:Temp:Fluc:Esti}
In this subsection we consider the temporal fluctuation. Recall
\begin{multline*}
\Lambda_{(k_1,\cdots,k_n)}^{(n,m)}(\bfy,t)=
\intLim_{0<s_1+s_2+\cdots +s_n<\Tr t/2}j(\bfy_{(k_1,k_2,\cdots, k_n)},t-s_1-s_2-\cdots -s_n)\\
\prod\limits_{i=1}^{n} G(\phi_i,\frac{s_i\sqrt{2RT(\bfy_{(k_1,\cdots,k_i)})}}{k_i})
\frac{\sqrt{2RT(\bfy_{(k_1,\cdots,k_i)})}}{k_i} ds_n d\phi_n \cdots ds_1 d\phi_1
\end{multline*}
\begin{multline*}
=
\intLim_{0<\frac{k_1\sigma_1}{2RT_{(k_1)}}+\frac{k_2\sigma_2}{2RT_{(k_1,k_2)}}+\cdots +\frac{k_n\sigma_n}{2RT_{(k_1,\ldots,k_n)}}<t/2}\\
j\BK{\bfy_{(k_1,k_2,\cdots, k_n)},t-\frac{k_1\sigma_1}{2RT_{(k_1)}}-\ldots-\frac{k_n\sigma_n}{2RT_{(k_1,\ldots,k_n)}}}\\
\prod\limits_{i=1}^{n} G(\phi_i,\sigma_i) d\sigma_n d\phi_n \cdots d\sigma_1 d\phi_1
\quad \text{when } d=2,
\end{multline*}

\begin{multline*}
\Lambda_{(k_1,\cdots,k_n)}^{(n,m)}(\pm1,t)=
\intLim_{0<s_1+s_2+\cdots +s_n<\Tr t/2}j(\pm1_{(k_1,k_2,\cdots, k_n)},t-s_1-s_2-\cdots -s_n)\\
\prod\limits_{i=1}^{n} H(\frac{s_i\sqrt{2RT(\pm1_{(k_1,\cdots,k_i)})}}{k_i})
\frac{\sqrt{2RT(\pm1_{(k_1,\cdots,k_i)})}}{k_i} ds_n  \cdots ds_1 
\end{multline*}
\begin{multline*}
=
\intLim_{0<\frac{k_1\sigma_1}{2RT_{(k_1)}}+\frac{k_2\sigma_2}{2RT_{(k_1,k_2)}}+\cdots +\frac{k_n\sigma_n}{2RT_{(k_1,\ldots,k_n)}}<t/2}\\
j\BK{\pm1_{(k_1,k_2,\cdots, k_n)},t-\frac{k_1\sigma_1}{2RT_{(k_1)}}-\ldots-\frac{k_n\sigma_n}{2RT_{(k_1,\ldots,k_n)}}}\\
\prod\limits_{i=1}^{n} H(\sigma_i) d\sigma_n  \cdots d\sigma_1 
\quad \text{when } d=1.
\end{multline*}

We note that the kernel $H(\sigma)$ and $G(\phi,\sigma)$ are smooth in $\sigma$, and hence we may differentiate $\Lambda_{(k_1,\cdots,k_n)}^{(n,m)}$
with respect to $t$ directly to obtain an explicit expression.
\begin{lemma}\label{lem:Fr:Temp:Fluc:Formula}
Let $n$ be any positive integer. $\Lambda_{(k_1,\ldots,k_n)}^{(n,m)}(\bfy,t)$ is $C^1$ with respect to $t$. Their derivatives has two parts:
\begin{equation*}
\begin{split}
	\frac{\partial \Lambda_{(k_1,\ldots,k_n)}^{(n,m)}}{ \partial t} (\bfy,t)
=&
	\mathcal B_{(k_1,\ldots,k_n)}^{(n,m)}(\bfy,t) + \mathcal V_{(k_1,\ldots,k_n)}^{(n,m)}(\bfy,t), 
\end{split}
\end{equation*}
The first term $\mathcal B_{(k_1,\ldots,k_n)}^{(n,m)}$ is the boundary term:
\begin{multline}\label{eq:Fr:Temp:Fluc:Formula:B:23D}
	 \mathcal B_{(k_1,\ldots,k_n)}^{(n,m)}(\bfy,t)
=
    -\BK{1-\Tr\frac{1}{2}}
     \intLim_{ s_1+\ldots+s_n = \Tr\frac{t}{2}  }
     \prod_{l=1}^n G \BKK{ \phi_l, \frac{s_l \sqrt{2RT_{(k_1,\ldots,k_l)} }}{k_l}} \frac{\sqrt{2RT_{(k_1,\ldots,k_l)}}}{k_l}
\\
     \times
      j \BK{ \bfy_{(k_1,\ldots,k_n)}, t-\Tr\frac{t}{2} }
      ds_1 \cdots ds_{n-1}d^n\phi
      \quad \text{when } d=2,
\end{multline}
\begin{multline}\label{eq:Fr:Temp:Fluc:Formula:B:1D}
	 \mathcal B_{(k_1,\ldots,k_n)}^{(n,m)}(\pm1,t)
=
    -\BK{1-\Tr\frac{1}{2}}
     \intLim_{ s_1+\ldots+s_n = \Tr\frac{t}{2}  }
     \prod_{l=1}^n H \BKK{  \frac{s_l \sqrt{2RT_{(k_1,\ldots,k_l)} }}{k_l}} \frac{\sqrt{2RT_{(k_1,\ldots,k_l)}}}{k_l}
\\
     \times
      j \BK{ \pm1_{(k_1,\ldots,k_n)}, t-\Tr\frac{t}{2} }
      ds_1 \cdots ds_{n-1}
      \quad \text{when } d=1.
\end{multline}
The second term $\mathcal V_{(k_1,\ldots,k_n)}^{(n,m)}$ is the volume term:
\begin{multline}\label{eq:Fr:Temp:Fluc:Formula:V:23D}
     \mathcal V_{(k_1,\ldots,k_n)}^{(n,m)}(\bfy,t)
=
	 \frac{1}{n}
	 \int
     \int^t_{ t-\Tr\frac{t}{2} }
	 \intLim_{s_1+\ldots+s_n=t-s}
     \sum_{l=1}^n
         G \BK{ \phi_1, \frac{s_1 \sqrt{2RT_{(k_1)} }}{k_l}}
        \times \cdots 
        \\
        \times
        \frac{\sqrt{2RT_{(k_1,\ldots,k_l)}}}{k_l} \frac{ \partial G }{ \partial \sigma } \BK{ \phi_l, \frac{s_l \sqrt{2RT_{(k_1,\ldots,k_l)}}}{k_l} }
        \times \cdots \times
         G \BK{ \phi_n, \frac{s_n \sqrt{2RT_{(k_1,\ldots,k_n)}}}{k_n} }
     \\
     \times 
      d s_1 \cdots ds_{n-1}
      ~
      j(\bfy_{(k_1,\ldots,k_n)},s) d s 
     \frac{\sqrt{2RT_{(k_1)}}}{k_1} \cdots \frac{\sqrt{2RT_{(k_1,\ldots,k_n)}}}{k_n}
     d^n\phi
     \quad \text{when } d=2,
\end{multline}
\begin{multline}\label{eq:Fr:Temp:Fluc:Formula:V:1D}
     \mathcal V_{(k_1,\ldots,k_n)}^{(n,m)}(\pm1,t)
=
	 \frac{1}{n}
	 \int
     \int^t_{ t-\Tr\frac{t}{2} }
	 \intLim_{s_1+\ldots+s_n=t-s}
     \sum_{l=1}^n
         H \BK{  \frac{s_1 \sqrt{2RT_{(k_1)} }}{k_l}}
        \times \cdots 
        \\
        \times
        \frac{\sqrt{2RT_{(k_1,\ldots,k_l)}}}{k_l} \frac{ \partial H }{ \partial \sigma } \BK{  \frac{s_l \sqrt{2RT_{(k_1,\ldots,k_l)}}}{k_l} }
        \times \cdots \times
         H \BK{  \frac{s_n \sqrt{2RT_{(k_1,\ldots,k_n)}}}{k_n} }
     \\
     \times 
      d s_1 \cdots ds_{n-1}
      ~
      j(\pm1_{(k_1,\ldots,k_n)},s) d s 
     \frac{\sqrt{2RT_{(k_1)}}}{k_1} \cdots \frac{\sqrt{2RT_{(k_1,\ldots,k_n)}}}{k_n}
     \quad \text{when } d=1.
\end{multline}
\end{lemma}

The lemma can be proved by a similar argument as in \cite{Kuo-Liu-Tsai,Kuo-Liu-Tsai-2} and we omit it.
The boundary term $\mathcal B_n$ can be easily bounded as following. First, we have
\begin{subequations}
\begin{multline*}
     |\mathcal B_{(k_1,\ldots,k_n)}^{(n,m)}(\bfy,t)|
=
	 O(1)
     \sup_{ \frac{t}{2}<s<t } \BK{ \VertBK{j}_{L^\infty_\bfy} }(s)
     \\
     \times
     \intLim_{ s_1+\ldots+s_n = \frac{t}{2}  }
     \BK{ \prod_{l=1}^n G \BKK{ \phi_l, \frac{s_l \sqrt{2RT_{(k_1,\ldots,k_l)}}}{k_l} } 
     \frac{\sqrt{2RT_{(k_1,\ldots,k_l)}}}{k_l} }
      ds_1 \cdots ds_{n-1}d^n\phi
\end{multline*}
\begin{multline*}
=
	 O(1)
     \sup_{ \frac{t}{2}<s<t } \BK{ \VertBK{j}_{L^\infty_\bfy} }(s)
     \\
     \times
     \intLim_{\frac{k_1 \sigma_1}{2RT_{(k_1)}} + \ldots+ \frac{k_n \sigma_n}{2RT_{(k_1,\ldots,k_n)}} 
     = \frac{t}{2}}
     \prod_{l=1}^n G \BKK{\phi_l,\sigma_l} 
      ds_1 \cdots ds_{n-1}d^n\phi
      \quad \text{when } d=2,
\end{multline*}
\begin{multline*}
     |\mathcal B_{(k_1,\ldots,k_n)}^{(n,m)}(\pm1,t)|
=
	 O(1)
    \sup_{ \frac{t}{2}<s<t } \Big( \absBK{j_+(s)} + \absBK{j_-(s)} \Big)
     \\
     \times
     \intLim_{ s_1+\ldots+s_n = \frac{t}{2}  }
     \BK{ \prod_{l=1}^n H \BKK{ \frac{s_l \sqrt{2RT_{(k_1,\ldots,k_l)}}}{k_l} } 
     \frac{\sqrt{2RT_{(k_1,\ldots,k_l)}}}{k_l} }
      ds_1 \cdots ds_{n-1}
\end{multline*}
\begin{multline*}
=
	 O(1)
   \sup_{ \frac{t}{2}<s<t } \Big( \absBK{j_+(s)} + \absBK{j_-(s)} \Big)
     \\
     \times
     \intLim_{\frac{k_1 \sigma_1}{2RT_{(k_1)}} + \ldots+ \frac{k_n \sigma_n}{2RT_{(k_1,\ldots,k_n)}} 
     = \frac{t}{2}}
     \prod_{l=1}^n H \BKK{\sigma_l} 
      ds_1 \cdots ds_{n-1}
      \quad \text{when } d=1.
\end{multline*}
\end{subequations}
Hence given $t'<t$ with $t'/mn \gg 1$, by the law of large numbers, Theorem \ref{thm:Fr:Law:of:Large},
\begin{subequations} \label{eq:z2}
\begin{multline}
	\int_{t'}^t |\mathcal B_{(k_1,\ldots,k_n)}^{(n,m)}(\bfy,s)| ds
=
     O(1) 
    \sup_{ \frac{t}{2}<s<t } \BK{ \VertBK{j}_{L^\infty_\bfy} }(s)
    \times
    \frac{m^{3}n^2\log(t'+1)}{t'^{3}},\\
    \quad \text{when } d=2,
\end{multline}
\begin{multline}
	\int_{t'}^t |\mathcal B_{(k_1,\ldots,k_n)}^{(n,m)}(\pm1,s)| ds
=
     O(1) 
   \sup_{ \frac{t}{2}<s<t } \Big( \absBK{j_+(s)} + \absBK{j_-(s)} \Big)
    \times
    \frac{m^{2}n\log(t'+1)}{t'^{2}},\\
    \quad \text{when } d=1.
\end{multline}
\end{subequations}

We next turn to the major term $\mathcal V_n$, First, as above, we have
\begin{subequations}\label{eq:z3}
\begin{multline}
     |\mathcal V_{(k_1,\ldots,k_n)}^{(n,m)}(\bfy,t)|
\leq
	 \frac{1}{n}
     \sup_{ \frac{t}{2}<s<t } \BK{ \VertBK{j}_{L^\infty_\bfy} }(s)
     \\
     \times
     \int \absBK{   
     				\sum_{l=1}^n
                	\frac{\sqrt{2RT_{(k_1,\ldots,k_l)}}}{k_l} G(\phi_1,\sigma_1) \cdots
                    \frac{\partial G}{\partial \sigma}(\phi_l,\sigma_l)
                    \cdots G(\phi_n,\sigma_n)
                 }
      d^n\sigma d^n\phi
      \quad \text{when } d=2,
\end{multline}
\begin{multline}
     |\mathcal V_{(k_1,\ldots,k_n)}^{(n,m)}(\pm1,t)|
\leq
	 \frac{1}{n}
     \sup_{ \frac{t}{2}<s<t } \Big( \absBK{j_+(s)} + \absBK{j_-(s)} \Big)
     \\
     \times
     \int \absBK{   
     				\sum_{l=1}^n
                	\frac{\sqrt{2RT_{(k_1,\ldots,k_l)}}}{k_l} H(\sigma_1) \cdots
                    \frac{\partial H}{\partial \sigma}(\sigma_l)
                    \cdots H(\sigma_n)
                 }
      d^n\sigma 
      \quad \text{when } d=1.
\end{multline}
\end{subequations}

To estimate \eqref{eq:z3}, we use the following lemma.
\begin{lemma}\label{lem:Fr:Temp:Flux:Decay:in:n}
For any integer $n>1$,
\begin{subequations}\label{eq:Fr:Temp:Flux:Decay:in:n}
\begin{multline}\label{eq:Fr:Temp:Flux:Decay:in:n:23D}
     \int \absBK{
                    \sum_{l=1}^n
                    \frac{\sqrt{2RT_{(k_1,\ldots,k_l)}}}{k_l} G(\phi_1,\sigma_1) \cdots
                    \frac{\partial G}{\partial \sigma}(\phi_l,\sigma_l)
                    \cdots G(\phi_n,\sigma_n)
                }
      d^n\sigma d^n\phi
     \\
=
      O\BK{ \BK{n\log n}^\frac12 },
\end{multline}
\begin{multline}\label{eq:Fr:Temp:Flux:Decay:in:n:1D}
     \int \absBK{
                    \sum_{l=1}^n
                    \frac{\sqrt{2RT_{(k_1,\ldots,k_l)}}}{k_l} H(\sigma_1) \cdots
                    \frac{\partial H}{ \partial\sigma}(\sigma_l)
                    \cdots H(\sigma_n)
                }
      d^n\sigma 
     \\
=
      O\BK{ \BK{n}^\frac12 },
\end{multline}
\end{subequations}
Consequently,
\begin{align*}
	|\mathcal V_{(k_1,\ldots,k_n)}^{(n,m)}(\bfy,t)|
=& 
	O(1) 
	\sup_{ \frac{t}{2}<s<t } \BK{ \VertBK{j}_{L^\infty_\bfy} }(s)
	\times
	\BK{ \frac{\log n}{n} }^\frac12,\text{ when } d=2,\\
		|\mathcal V_{(k_1,\ldots,k_n)}^{(n,m)}(\pm1,t)|
=& 
	O(1) 
	\sup_{ \frac{t}{2}<s<t } \Big( \absBK{j_+(s)} + \absBK{j_-(s)} \Big)
	\times
	\BK{ \frac{1}{n} }^\frac12,
	\text{ when } d=1.
\end{align*}
\end{lemma}
\begin{proof}\hfil

Note that $k_l\geq1$ for each $l$, so
\begin{align*}
&\absBK{
                    \sum_{l=1}^n
                    \frac{\sqrt{2RT_{(k_1,\ldots,k_l)}}}{k_l} G(\phi_1,\sigma_1) \cdots
                    \frac{\partial G}{\partial \sigma}(\phi_l,\sigma_l)
                    \cdots G(\phi_n,\sigma_n)
                }\\
                &\leq
                \absBK{
                    \sum_{l=1}^n
                    \sqrt{2RT_{(k_1,\ldots,k_l)}} G(\phi_1,\sigma_1) \cdots
                    \frac{\partial G}{\partial \sigma}(\phi_l,\sigma_l)
                    \cdots G(\phi_n,\sigma_n)
                }
\end{align*}
Then we follow the Lemma 3 in \cite{Kuo-Liu-Tsai-2} to conclude the proof.
\end{proof}

The following theorem follows from Lemma \ref{lem:Fr:Temp:Flux:Decay:in:n} together with \eqref{eq:z2}.

\begin{theorem}\label{thm:flu-tem-sep}
Let $t'<t$, then, for $1\ll t'/mn$,
\begin{multline*}
     \Lambda^{(n,m)}_{(k_1,\ldots,k_n)}(\bfy,t) - \Lambda^{(n,m)}_{(k_1,\ldots,k_n)}(\bfy,t')
     \\
=
      O(1)
     \sup_{ \frac{t'}{2}<s<t } \BK{ \VertBK{j}_{L^\infty_\bfy} }(s)
     \times
     \BK{
             \frac{ m^{3}n^2\log(t'+1)}{ t'^{3} }
            +
              \BK{ \frac{\log n}{n} }^\frac12 (t-t')
          },\text{ when } d=2,
     \\
\end{multline*}
\begin{multline*}
     \Lambda^{(n,m)}_{(k_1,\ldots,k_n)}(\pm1,t) - \Lambda^{(n,m)}_{(k_1,\ldots,k_n)}(\pm1,t')
     \\
=
      O(1)
     \sup_{ \frac{t'}{2}<s<t } \Big( \absBK{j_+(s)} + \absBK{j_-(s)} \Big)
     \times
     \BK{
             \frac{ m^{2}n\log(t'+1)}{ t'^{2} }
            +
              \BK{ \frac{1}{n} }^\frac12 (t-t')
          },\text{ when } d=1.
     \\
\end{multline*}
\end{theorem}

From  Theorem \ref{thm:flu-tem-sep}, \eqref{eq:Fr:Flu:Lambda} and \eqref{eq:Fr:Flu:Lambda-sep}, we obtain:
\begin{corollary}[Temporal Fluctuation Estimate]\label{cor:Fr:Temp:Fluc}
Let $t'<t$, then, for $1\ll t'/mn$,
\begin{multline*}
      j(\bfy,t) - j(\bfy,t')
=
      O(1) \BK{\frac{m(mn)^{3}\log(t'+1)}{t'^{3}}+(1-\alpha)^m}
     \BK{ \VertBK{g_{in}}_{\infty,\mu} + \calJ(t) }
     \\
    +
      O(1)
     \sup_{ \frac{t'}{2}<s<t } \BK{ \VertBK{j}_{L^\infty_\bfy} }(s)
     \times \BK{ \frac{\log n}{n} }^\frac12 (t-t'),
      \text{ for } d=2,
\end{multline*}
\begin{multline*}
      j(\pm1,t) - j(\pm1,t')
=
      O(1) \BK{\frac{m(mn)^{2}\log(t'+1)}{t'^{2}}+(1-\alpha)^m}
     \BK{ \VertBK{g_{in}}_{\infty,\mu} + \calJ(t) }
     \\
    +
      O(1)
     \sup_{ \frac{t'}{2}<s<t } \Big( \absBK{j_+(s)} + \absBK{j_-(s)} \Big)
     \times \BK{ \frac{1}{n} }^\frac12 (t-t'),
      \text{ for } d=1.
\end{multline*}
\end{corollary}

\end{subsection}

\begin{subsection}{Spacial Fluctuation Estimate}\label{sec:Fr:Spac:Fluc}

In this subsection, we investigate the spacial fluctuation. 
\begin{theorem}[Spacial Fluctuation Estimate]\label{thm:Fr:Spac:Fluc}
Suppose that  $t/mn,\ t/mN,\ N \gg 1, 0<q<1$, $\bfy,\bfy'\in\partial D$. Then
\begin{multline*}
     j(\bfy,t) - j(\bfy',t)
=
     O(1) \BK{\frac{ ((mn)^{3}+m^3N^2) \log(t+1) }{ t^{3} } +(1-\alpha)^m}\BK{ \VertBK{g_{in}}_{\infty,\mu} + \calJ(t) }
    \\
    +
     O(1) \sup_{\frac{t}{2}<s<t} \BK{\VertBK{j}_{L^\infty_\bfy}}(s) \times
    \BK{
            \BK{\frac{\log n}{n}}^\frac12 mN + \BK{\frac{\log N}{N}}^\frac12
        },\\
        \text{ when } d=2,
\end{multline*}
\begin{multline*}
|j(+1,t)-j(-1,t)|=O(1)\BK{\frac{m^3n^{2}\log(t+1)}{t^{2}}+(1-\alpha)^m}
\BK{\calJ(t)+\VertBK{g_{in}}_{\infty,\mu} }\\
+O(1)\BK{\BK{ \frac{1}{n} }^\frac12 t^q+\BK{ \frac{m}{t^q} }^2}
     \sup_{ \frac{t}{2} < s < t } \Big( \absBK{j_+(s)} + \absBK{j_-(s)} \Big)\\
     \text{ when } d=1.
\end{multline*}
\end{theorem}

We consider the one dimensional case first, which is much simpler than the multidimensional cases.

\ \\
{\sc One dimensional case, $d=1$.}\\

This case differs from the multidimensional cases in a fundamental sense: unlike the multidimensional cases, the boundary, comprising of two points, is {\it discrete}. This makes the one dimensional case much easier. Instead of processing $\Lambda_{\pm,n}$, we directly estimate: $j(+1,t)-j(-1,t)$. Recall
\begin{multline*}
j(+1,t)=\sum\limits_{k=0}^{m-1}(1-\alpha)^k j_{in}^{(k)}(+1,t)\\
+\alpha\sum\limits_{k=1}^{m}(1-\alpha)^{k-1}
\intLim_{0}^{t}
j((-1)^{k},t-s)H(\frac{s\sqrt{2RT((-1)^{k})}}{k})\frac{\sqrt{2RT((-1)^{k})}}{k} ds \\
+(1-\alpha)^{m}E^{(m)}(+1,t).
\end{multline*}
From
\begin{equation*}
\intLim_{0}^\infty H(\frac{s\sqrt{2RT((-1)^{k})}}{k})\frac{\sqrt{2RT((-1)^{k})}}{k} ds=1,
\end{equation*}
we have
\begin{equation*}
j(-1,t)=\alpha\sum\limits_{k=1}^{m}(1-\alpha)^{k-1}
\intLim_{0}^{\infty}j(-1,t)H(\frac{s\sqrt{2RT((-1)^{k})}}{k})\frac{\sqrt{2RT((-1)^{k})}}{k} ds
+(1-\alpha)^m j(-1,t).
\end{equation*}
Now consider
\begin{multline*}
j(+1,t)-j(-1,t)
=\sum\limits_{k=0}^{m-1}(1-\alpha)^k j_{in}^{(k)}(+1,t)\\
+\alpha\sum\limits_{k=1}^{m}(1-\alpha)^{k-1}
\intLim_{0}^{t}
\BK{j((-1)^{k},t-s)-j(-1,t)}H(\frac{s\sqrt{2RT((-1)^{k})}}{k})\frac{\sqrt{2RT((-1)^{k})}}{k} ds \\
-\alpha\sum\limits_{k=1}^{m}(1-\alpha)^{k-1}
\intLim_{t}^{\infty}
j(-1,t)H(\frac{s\sqrt{2RT((-1)^{k})}}{k})\frac{\sqrt{2RT((-1)^{k})}}{k} ds \\
+(1-\alpha)^{m}\BK{E^{(m)}(+1,t)-j(-1,t)}.
\end{multline*}
As before,
\begin{align*}
&\sum\limits_{k=0}^{m-1}(1-\alpha)^k j_{in}^{(k)}(+1,t)
=O(1)\VertBK{g_{in}}_{\infty,\mu} \frac{m^2}{( t+1)^2}\\
&(1-\alpha)^{m}\BK{E^{(m)}(+1,t)-j(-1,t)}=O(1)(1-\alpha)^m
\BK{\calJ(t)+\VertBK{g_{in}}_{\infty,\mu} },
\end{align*}
and we decompose
\begin{multline*}
\intLim_{0}^{t}
\BK{j((-1)^{k},t-s)-j(-1,t)}H(\frac{s\sqrt{2RT((-1)^{k})}}{k})\frac{\sqrt{2RT((-1)^{k})}}{k} ds\\
=\intLim_{0}^{t^q}\Big(\ldots\Big)+\intLim_{t^q}^{t/2}\Big(\ldots\Big)
+\intLim_{t/2}^{t}\Big(\ldots\Big),
\end{multline*}
where $0<q<1$ is to be determined.
And then
\begin{align*}
&\alpha\sum\limits_{k=1}^{m}(1-\alpha)^{k-1}
\intLim_{t}^{\infty}
j(-1,t)H(\frac{s\sqrt{2RT((-1)^{k})}}{k})\frac{\sqrt{2RT((-1)^{k})}}{k} ds\\
&=O(1)\calJ(t)\alpha\sum\limits_{k=1}^{m}(1-\alpha)^{k-1}\frac{k^2}{t^2}
=O(1)\calJ(t)\frac{m^2}{ t^2},\\
&\alpha\sum\limits_{k=1}^{m}(1-\alpha)^{k-1}\intLim_{t/2}^{t}\Big(\ldots\Big)=
O(1)\calJ(t)\alpha\sum\limits_{k=1}^{m}(1-\alpha)^{k-1}\frac{k^2}{t^2}
=O(1)\calJ(t)\frac{m^2}{ t^2},\\
&\alpha\sum\limits_{k=1}^{m}(1-\alpha)^{k-1}\intLim_{t^q}^{t/2}\Big(\ldots\Big)=
O(1)\BK{ \frac{1}{t^q }}^2
     \sup_{ \frac{t}{2} < s < t } \Big( \absBK{j_+(s)} + \absBK{j_-(s)} \Big)
\alpha\sum\limits_{k=1}^{m}(1-\alpha)^{k-1}k^2\\
&=O(1)\BK{ \frac{m}{t^q} }^2
     \sup_{ \frac{t}{2} < s < t } \Big( \absBK{j_+(s)} + \absBK{j_-(s)} \Big).
\end{align*}
Finally, we consider
\begin{multline*}
\intLim_{0}^{t^q}
\BK{j((-1)^{k},t-s)-j(-1,t)}H(\frac{s\sqrt{2RT((-1)^{k})}}{k})\frac{\sqrt{2RT((-1)^{k})}}{k} ds\\
=\intLim_{0}^{t^q}
\BK{j((-1)^{k},t-s)-j((-1)^{k},t)+j((-1)^{k},t)-j(-1,t)}H(\frac{s\sqrt{2RT((-1)^{k})}}{k})\frac{\sqrt{2RT((-1)^{k})}}{k} ds,
\end{multline*}
and it is easy to see that
\begin{multline*}
 \alpha\sum\limits_{k=1}^{m}(1-\alpha)^{k-1}
 \intLim_{0}^{t^q}
\BK{j((-1)^{k},t-s)-j((-1)^k,t)}H(\frac{s\sqrt{2RT((-1)^{k})}}{k})\frac{\sqrt{2RT((-1)^{k})}}{k} ds\\
\leq \sup_{ t-t^q < t' < t } \absBK{ j(\pm1,t) - j(\pm1,t') },\\
\alpha\sum\limits_{k=1}^{m}(1-\alpha)^{k-1}
 \intLim_{0}^{t^q}
\BK{j((-1)^{k},t)-j(-1,t)}H(\frac{s\sqrt{2RT((-1)^{k})}}{k})\frac{\sqrt{2RT((-1)^{k})}}{k} ds\\
=\alpha\sum\limits_{k=1,k\text{ even}}^{m}(1-\alpha)^{k-1}
 \intLim_{0}^{t^q}
\BK{j(+1,t)-j(-1,t)}H(\frac{s\sqrt{2RT((-1)^{k})}}{k})\frac{\sqrt{2RT((-1)^{k})}}{k} ds\\
\leq\frac{1-\alpha}{2-\alpha}\BK{j(+1,t)-j(-1,t)}.
\end{multline*}
Hence
\begin{multline*}
\frac{1}{2-\alpha}\BK{j(+1,t)-j(-1,t)}=O(1)\BK{\frac{m^2}{ t^2}+(1-\alpha)^m}
\BK{\calJ(t)+\VertBK{g_{in}}_{\infty,\mu} }\\
+\sup_{ t-t^q < t' < t } \absBK{ j(\pm1,t) - j(\pm1,t') }\\
+O(1)\BK{ \frac{m}{t^q}}^2
     \sup_{ \frac{t}{2} < s < t } \Big( \absBK{j_+(s)} + \absBK{j_-(s)} \Big),
\end{multline*}
and it follows
\begin{multline*}
|j(+1,t)-j(-1,t)|=O(1)\BK{\frac{m^3n^{2}\log(t+1)}{t^{2}}+\frac{m^2}{ t^2}+(1-\alpha)^m}
\BK{\calJ(t)+\VertBK{g_{in}}_{\infty,\mu} }\\
+O(1)\BK{ \frac{m}{t^q} }^2
     \sup_{ \frac{t}{2} < s < t } \Big( \absBK{j_+(s)} + \absBK{j_-(s)} \Big)\\
+O(1)
     \sup_{ \frac{t}{2}<s<t } \Big( \absBK{j_+(s)} + \absBK{j_-(s)} \Big)
     \times \BK{ \frac{1}{n} }^\frac12 t^q.
\end{multline*}

Subsequently, we can apply the temporal fluctuation estimate, Corollary \ref{cor:Fr:Temp:Fluc}, to the term $\sup_{t-t^q< t'<t} |j(\pm1,t)-j(\pm1,t')|$. Hence Theorem \ref{thm:Fr:Spac:Fluc} for $d=1$ follows. Note that $t-t^q>\frac{t}{2}$ so long as $t\gg 1$.

\ \\
{\sc Multidimensional cases, $d=2$.}\\

As noted before, to estimate spacial fluctuation we invoke another variable $N$. From now on $N$ will be the index of $\Lambda^{(N,m)}_{(k_1,\ldots,k_N)}$.

The boundary $\partial D$ is unit circle, so we parametrize it by the polar coordinates. Given two boundary points $\bfy$ and $\bfy'$, let $\bfy'$ be point of degree zero, and denote the polar angle of $\bfy$ by $\theta$. Denote the relative polar angle of $\bfy_{(k_1,\ldots,k_{l-1},1)}$ with respect to $\bfy_{(k_1,\ldots,k_{l-1})}$ by $\theta_l$. Then for $1\leq i \leq k_l$, the relative polar angle of $\bfy_{(k_1,\ldots,k_{l-1},i)}$ with respect to $\bfy_{(k_1,\ldots,k_{l-1},i-1)}$ is also $\theta_l$ because of the specular reflection,  i.e. $\theta+k_1\theta_1+\ldots+k_l\theta_l$ stands for the absolute polar angle of $\bfy_{(k_1,\ldots,k_l)}$. Since $\partial D$ is the unit circle, $\theta_l=\pi-2\phi_l$,
%Figure \ref{fig:circle}. 
To simplify the notation, put $T_{(k_1,\ldots,k_l)}=T(\theta+k_1\theta_1+\ldots+k_l\theta_l)$, and $T'_{(k_1,\ldots,k_l)}=T(k_1\theta_1+\ldots+k_l\theta_l)$. Under this coordinate system, we have %\eqref{eq:Fr:Defn:Lambda:sigma:23D},
\begin{equation}\label{eq:Fr:Spac:Change:Variable:2D}
\begin{split}
&
    \Lambda_{(k_1,\ldots,k_N)}^{(N,m)}(\bfy,t)
    \\
&=
    \Lambda_{(k_1,\ldots,k_N)}^{(N,m)}(\theta,t)  
=
    \intLim_{ \frac{k_1\sigma_1}{\sqrt{2RT_{(k_1)}}} + \ldots + \frac{k_N\sigma_N}{\sqrt{2RT_{(k_1,\ldots,k_N)}}} < \Tr\frac{t}{2} }
    \prod_{l=1}^N G(\phi_l,\sigma_l)
    \\
&
	\quad
    \times
     j\BKK{
            \theta + k_1\theta_1 + \ldots + k_N\theta_N ,
             t - \frac{k_1\sigma_1}{\sqrt{2RT_{(k_1)}}} - \ldots - \frac{k_N\sigma_N}{\sqrt{2RT_{(k_1,\ldots,k_N)}}}
           }
     d^N \sigma d^N \phi
    \\
=&
    \intLim_{ 
    			\frac{k_1\sigma_1}{\sqrt{2R\tilde T_{(k_1)}}} + \ldots 
    			+ \frac{k_N\sigma_N}{\sqrt{2R\tilde T_{(k_1,\ldots,k_N)}}} < \Tr\frac{t}{2} 
    		}
    \prod_{l=1}^N G \BKK{ \phi_l + \frac{\theta}{2k_lN}, \sigma_l }
    \\
&
	\quad
    \times
     j\BKK{
            k_1\theta_1 + \ldots + k_N\theta_N ,
             t - \frac{k_1\sigma_1}{\sqrt{2R\tilde T_{(k_1)}}} - \ldots - \frac{k_N\sigma_N}{\sqrt{2R\tilde T_{(k_1,\ldots,k_N)}}}
           }
     d^N \sigma d^N \phi,
\end{split}
\end{equation}
where $\tilde T_{(k_1,\ldots,k_l)} = T(\theta+k_1\theta_1+\cdots+k_N\theta_N-\frac{l\theta}{N})$. 
\begin{remark}
In \cite{Kuo-Liu-Tsai} and \cite{Kuo-Liu-Tsai-2} we can deal with the case of spherical domain in $\bbR^d$ for $d=1,2,3,$ but in this paper we only consider the the spherical domain in $\bbR^d$ for $d=1,2$. As in our previous works \cite{Kuo-Liu-Tsai} and \cite{Kuo-Liu-Tsai-2} we need the symmetric property of domain to calculate exactly the spacial fluctuation. The symmetry of the boundary allows us to tract the exact location of the particle after multiple reflections. This is an essential ingredient of our analysis on treating spacial fluctuation. For two dimensional case, thanks to polar coordinate we are able to tract the exact location of the particle on a circle after mixed specular reflections and diffuse reflections. This allows us to conduct a change of variables, \eqref{eq:Fr:Spac:Change:Variable:2D}, and to estimate the spacial fluctuation, \eqref{eq:z5}. However, in three dimensional case there is no universal coordinate to tract the exact location of the particle on a sphere after mixed specular reflections and diffuse reflections. The three dimensional case might require mathematical analysis different from ours. We hope to return to this problem in the future.
\end{remark}

From \eqref{eq:Fr:Spac:Change:Variable:2D}, we have
\begin{equation}\label{eq:z5}
\begin{split}
&
    \Lambda_{(k_1,\ldots,k_N)}^{(N,m)}(\bfy,t) - \Lambda_{(k_1,\ldots,k_N)}^{(N,m)}(\bfy',t) = \Lambda_{(k_1,\ldots,k_N)}^{(N,m)}(\theta,t) - \Lambda_{(k_1,\ldots,k_N)}^{(N,m)}(0,t)
    \\
=&
    \intLim_{ 
    			\frac{k_1\sigma_1}{\sqrt{2R\tilde T_{(k_1)}}} + \ldots 
    			+ \frac{k_N\sigma_N}{\sqrt{2R\tilde T_{(k_1,\ldots,k_N)}}} < \Tr\frac{t}{2} 
    		}
    \prod_{l=1}^N G \BKK{ \phi_l + \frac{\theta}{2k_lN}, \sigma_l }
	\\
&
	\hspace{.1\textwidth}
    \times
     j\BKK{
            k_1\theta_1 + \ldots + k_N\theta_N ,
             t - \frac{k_1\sigma_1}{\sqrt{2R\tilde T_{(k_1)}}} - \ldots - \frac{k_N\sigma_N}{\sqrt{2R\tilde T_{(k_1,\ldots,k_N)}}}
           }
     d^N \sigma d^N \phi
    \\
&
    -
    \intLim_{ 
    			\frac{k_1\sigma_1}{\sqrt{2RT'_{(k_1)}}} + \ldots 
    			+ \frac{k_N\sigma_N}{\sqrt{2RT'_{(k_1,\ldots,k_N)}}} < \Tr\frac{t}{2} 
    		}
    \prod_{l=1}^N G(\phi_l,\sigma_l)
	\\
&
	\hspace{.1\textwidth}
    \times
     j\BKK{
            k_1\theta_1 + \ldots + k_N\theta_N ,
             t - \frac{k_1\sigma_1}{\sqrt{2RT'_{(k_1)}}} - \ldots - \frac{k_N\sigma_N}{\sqrt{2RT'_{(k_1,\ldots,k_N)}}}
           }
     d^N \sigma d^N \phi.
\end{split}
\end{equation}
The two terms in \eqref{eq:z5} differ from each other in three places: the domain of integration, the angular variable of the transition PDF $G$, and the time variable of $j$. 

We now break the spatial fluctuation into three parts:
\begin{equation*}
    \Lambda_{(k_1,\ldots,k_N)}^{(N,m)}(\bfy,t) - \Lambda_{(k_1,\ldots,k_N)}^{(N,m)}(\bfy',t) = \Lambda_{(k_1,\ldots,k_N)}^{(N,m)}(\theta,0) - \Lambda_{(k_1,\ldots,k_N)}^{(N,m)}(0,t)
=
	 U_1+U_2+U_3,
\end{equation*}
where
\begin{multline*}
	 U_1
\equiv
    \BK{
             \intLim_{ 
             			\frac{k_1\sigma_1}{\sqrt{2R\tilde T_{(k_1)}}} 
             			+ \ldots + \frac{k_N\sigma_N}{\sqrt{2R\tilde T_{(k_1,\ldots,k_N)}}} < \Tr\frac{t}{2}
             		}
            -\intLim_{ 
            			\frac{k_1\sigma_1}{\sqrt{2RT'_{(k_1)}}} 
            			+ \ldots + \frac{k_N\sigma_N}{\sqrt{2RT'_{(k_1,\ldots,k_N)}}} < \Tr\frac{t}{2}
            		}
        }\times
    \\
    \BK{ \prod_{l=1}^N G \BKK{ \phi_l + \frac{\theta}{2k_lN}, \sigma_l } }
     j\BKK{
            k_1\theta_1 + \ldots + k_N\theta_N ,
             t - \frac{k_1\sigma_1}{\sqrt{2R\tilde T_{(k_1)}}} - \ldots - \frac{k_N\sigma_N}{\sqrt{2R\tilde T_{(k_1,\ldots,k_N)}}}
            }
     d^N\sigma d^N\phi,
\end{multline*}
\begin{multline*}
	 U_2
\equiv
    \intLim_{ \frac{k_1\sigma_1}{\sqrt{2RT'_{(k_1)}}} + \ldots 
    			+ \frac{k_N\sigma_N}{\sqrt{2RT'_{(k_1,\ldots,k_N)}}} < \Tr\frac{t}{2}
    		}
    \prod_{l=1}^N G \BKK{ \phi_l + \frac{\theta}{2k_lN}, \sigma_l }
	\\
    \times
    \Bigg[
             j\BKK{
                    k_1\theta_1 + \ldots + k_N\theta_N ,
                    t - \frac{k_1\sigma_1}{\sqrt{2R\tilde T_{(k_1)}}} - \ldots 
                    - \frac{k_N\sigma_N}{\sqrt{2R\tilde T_{(k_1,\ldots,k_N)}}}
                  }
    \\
          -
                    j\BKK{
                            k_1\theta_1 + \ldots + k_N\theta_N ,
                             t - \frac{k_1\sigma_1}{\sqrt{2RT'_{(k_1)}}} - \ldots - \frac{k_N\sigma_N}{\sqrt{2RT'_{(k_1,\ldots,k_N)}}}
                          }
            \Bigg]
                    d^N\sigma d^N\phi,
\end{multline*}
\begin{multline*}
     U_3
\equiv
    \intLim_{ \frac{k_1\sigma_1}{\sqrt{2RT'_{(k_1)}}} + \ldots + \frac{k_N\sigma_N}{\sqrt{2RT'_{(k_1,\ldots,k_N)}}} < \Tr\frac{t}{2} }
    \BK{
            \prod_{l=1}^N G \BKK{ \phi_l + \frac{\theta}{2k_lN}, \sigma_l } - \prod_{l=1}^N G(\phi_l,\sigma_l)
        }
    \\
    \times
     j\BKK{
                k_1\theta_1 + \ldots + k_N\theta_N ,
                 t - \frac{k_1\sigma_1}{\sqrt{2RT'_{(k_1)}}} - \ldots - \frac{k_N\sigma_N}{\sqrt{2RT'_{(k_1,\ldots,k_N)}}}
            }
     d^N\sigma d^N\phi.
\end{multline*}
$U_1,U_2,U_3$ register the difference in domain of integration, the time variable of $j$, and angular variable of the transition PDF $G$, respectively.

We now proceed to estimate $U_1$, $U_2$, and $U_3$. Consider first $U_1$. As noted before, $U_1$ registers the difference in domain of integration. Denote by $A\ominus B$ the symmetric difference $(A\setminus B)\cup(B\setminus A)$. Since $k_i\leq m$ for each $i$ and $T\geq T_*$ on the boundary, one can observe that both of the events
\begin{align*}
	\mathscr E_1
\equiv&
	\curBK{     
			\frac{k_1\sigma_1}{\sqrt{2R\tilde T_{(k_1)}}} + \ldots 
        	+ \frac{k_N\sigma_N}{\sqrt{2R\tilde T_{(k_1,\ldots,k_N)}}} < \Tr\frac{t}{2} 
        	}
    \\
	\mathscr E_2
\equiv&	
    \curBK{   \frac{k_1\sigma_1}{\sqrt{2RT'_{(k_1)}}} + \ldots 
    			+ \frac{k_N\sigma_N}{\sqrt{2RT'_{(k_1,\ldots,k_N)}}} < \Tr\frac{t}{2}
    		}
\end{align*}
contain
\begin{equation*}
        \curBK{ \sigma_1+\ldots+\sigma_N < \sqrt{2R\Tm} \Tr\frac{t}{2m} }.
\end{equation*}
So we have
\begin{equation}\label{eq:z8}
	\mathscr E_1 \ominus \mathscr E_2 
\subset
	\curBK{ \sigma_1+\ldots+\sigma_N \geq \sqrt{2R\Tm} \Tr\frac{t}{2m} }.
\end{equation}
This implies that $\mathscr E_1\ominus \mathscr E_2$, and thereby $U_1$, is a rare event, and can be estimated by the law of large numbers, Theorem \ref{thm:Fr:Law:of:Large}.

\begin{lemma}\label{lem:Fr:Spac:Fluc:U1:23D}
Assume $t/(Nm) \gg 1$, 
\begin{align}\label{eq:Fr:U1:Esti}
     U_1
=&
     O(1) \frac{m^3N^2\log(t+1)}{t^{3}} \times \calJ(t).
\end{align}
\end{lemma}

\begin{proof}\hfil

From \eqref{eq:z8}, 
\begin{equation*}
     |U_1|
\leq
     \calJ(t) \times
      2 \hspace{-20pt}
     \intLim_{\sigma_1+\ldots+\sigma_N \geq \sqrt{2R\Tm}\Tr\frac{t}{2m}   }
     \hspace{-20pt}
      G(\phi_1,\sigma_1) \cdots G(\phi_N,\sigma_N) d^N\sigma
      d^n\phi.
\end{equation*}
Applying law of large numbers Theorem \ref{thm:Fr:Law:of:Large} with $\gamma = \sqrt{2R\Tm}\Tr\frac{t}{3}$, we conclude \eqref{eq:Fr:U1:Esti}.
\end{proof}

$U_2$ records the difference in time variable of $j$, which is exactly the temporal fluctuation of $j$. Therefore, we can apply our previous estimate in temporal fluctuation to this part.

\begin{lemma}\label{lem:Fr:Spac:Fluc:U2:23D}
Suppose that $t/mn,t/mN,N \gg 1$. Then
\begin{equation}\label{eq:Fr:U2:Esti}
\begin{split}
     U_2
=&
     O(1) \BK{\frac{(mn)^{3}\log(t+1)}{t^{3}} +(1-\alpha)^m}\BK{ \VertBK{g_{in}}_{\infty,\mu} + \calJ(t) }\\
   & +
     O(1) \sup_{\frac{t}{2}<s<t} \BK{\VertBK{j}_{L^\infty_\bfy}}(s)
    \times
    \BK{
            \frac{\log N}{N} + \BK{\frac{\log n}{n}}^\frac12 mN
        }.
\end{split}
\end{equation}
\end{lemma}

\begin{proof}\hfil

\begin{multline*}
     U_2
    =
    \intLim_{ 
    			\frac{k_1\sigma_1}{\sqrt{2RT'_{(k_1)}}} + \ldots 
    			+ \frac{k_N\sigma_N}{\sqrt{2RT'_{(k_1,\ldots,k_N)}}} < \Tr\frac{t}{2} 
    		}
    \prod_{l=1}^N G \BKK{ \phi_l+\frac{\theta}{2k_lN}, \sigma_l }
\\
    \times
    \Bigg[
             j\BKK{
                    k_1\theta_1 + \ldots + k_N\theta_N,
                     t - \frac{k_1\sigma_1}{\sqrt{2R\tilde T_{(k_1)}}} - \ldots 
                     - \frac{k_N\sigma_N}{\sqrt{2R\tilde T_{(k_1,\ldots,k_N)}}}
                  }
\\
     -
             j\BKK{
                    k_1\theta_1 + \ldots + k_N\theta_N,
                     t - \frac{k_1\sigma_1}{\sqrt{2RT'_{(k_1)}}} - \ldots - \frac{k_N\sigma_N}{\sqrt{2RT'_{(k_1,\ldots,k_N)}}}
                  }
    \Bigg]
     d^N\sigma d^N\phi
\end{multline*}
\begin{multline*}
    \hphantom{U_2}
    =
    \BK{\mathop{
                    \int_{ |\sigma_1+\ldots+\sigma_N - N\E[X_1]| > N }
                    +
                    \int_{ |\sigma_1+\ldots+\sigma_N - N\E[X_1]| < N }
                }
                _{ \frac{k_1\sigma_1}{\sqrt{2RT'_{(k_1)}}} + \ldots + \frac{k_N\sigma_N}{\sqrt{2RT'_{(k_1,\ldots,k_N)}}} < \Tr\frac{t}{2} }}
     (\cdots)
    \equiv
     U_{21} + U_{22},
    \\
\end{multline*}

For $U_{21}$, since
\begin{multline*}
        \curBK{     \frac{k_1\sigma_1}{\sqrt{2RT'_{(k_1)}}} + \ldots + \frac{k_N\sigma_N}{\sqrt{2RT'_{(k_1,\ldots,k_N)}}} < \Tr\frac{t}{2} }
\subset
        \\
        \left\{     k_1\sigma_1+\ldots+ k_N\sigma_N < \sqrt{2R\TM} \Tr\frac{t}{2} \middle\}
\subset
        \middle\{   \frac{k_1\sigma_1}{\sqrt{2R\tilde T_{(k_1)}}} + \ldots + \frac{k_N\sigma_N}{\sqrt{2R\tilde T_{(k_1,\ldots,k_N)}}} < \frac{t}{2} \right\},
\end{multline*}
we have
\begin{equation*}
\begin{split}
      |U_{21}|
\leq&
     \sup_{ \frac{t}{2}<s<t } \BK{\VertBK{j}_{L^\infty_\bfy}}(s)
     \\
&
     \times 2
     \intLim_{ |\sigma_1+\ldots+\sigma_N - N\E[X_1]| > N}
     \hspace{-20pt}
      G(\phi_1,\sigma_1) \cdots G(\phi_N,\sigma_N) d^N\sigma d^N\phi.
\end{split}
\end{equation*}
Applying Theorem \ref{thm:Fr:Law:of:Large} with $\gamma = N$, we obtain
\begin{equation}\label{eq:z10}
    U_{21} = O(1)\sup\limits_{ \frac{t}{2}<s<t } \BK{\VertBK{j}_{L^\infty_\bfy}}(s) \times \frac{\log N}{N}.
\end{equation}
Note that the prerequisite $N/N^{\frac{1}{3}} \gg 1$ of Theorem \ref{thm:Fr:Law:of:Large} is satisfied since $N \gg 1$. Next, for $U_{22}$, by Corollary \ref{cor:Fr:Temp:Fluc},
\begin{equation}\label{eq:z11}
\begin{split}
&
         |U_{22}|
        \\
\leq&
        \sup
        \curBK{
                    \absBK{ j(\bfy,s) - j(\bfy,s') }
                    :
                    s,s' \in \BK{ t - mN\frac{1+\E(X_1)}{\sqrt{2R\Tm}} ,t }, \ \bfy\in\partial D
              }
        \\
=&
      O(1) \BK{\frac{ (nm)^{3}\log(t+1) }{ t^{3} } +(1-\alpha)^m }
     \BK{ \VertBK{g_{in}}_{\infty,\mu} + \calJ(t) }
     \\
&
    +
      O(1)
     \sup_{ \frac{t}{2}<s<t } \BK{ \VertBK{j}_{L^\infty_\bfy} }(s)
     \times \BK{ \frac{\log n}{n} }^\frac12 mN.
\end{split}
\end{equation}
From \eqref{eq:z10} and \eqref{eq:z11} we conclude \eqref{eq:Fr:U2:Esti}.
\end{proof}

Finally, we investigate $U_3$. $U_3$ involves only the angular difference of those PDF $G$. No difference in boundary temperature are included. Therefore, the estimate of $U_3$ is reduced to the constant boundary temperature case, as in \cite{Kuo-Liu-Tsai}. Since all the $\theta$ dependences appear only in the $N$ copy of $G$, $U_3$ is a $C^1$ function of $\theta$. Moreover, $U_3|_{\theta=0}=0$. By direct computations,
\begin{multline}\label{eq:z12}
    \absBK{ \frac{dU_3}{d\theta} }
    \leq
    \sup_{\frac{t}{2}<s<t} \BK{\VertBK{j}_{L^\infty_\bfy}}(s)
\\  
    \times
    \frac{1}{N}
	\int
    \frac{1}{2}
    \absBK{
			\sum_{l=1}^N G(\phi_1,\sigma_1) \cdots
			\frac{1}{k_l}\frac{\partial G}{\partial\phi}(\phi_l,\sigma_l)
			\cdots G(\phi_N,\sigma_N)
			}
	 d^N\phi d^N\sigma.
\end{multline}

The RHS of \eqref{eq:z12} can be derived by the similar argument as in \cite{Kuo-Liu-Tsai}. Thanks to $k_l\geq 1$ for each $l$, the following lemma allows to obtain a decay of $U_3$ in $N$. For a proof, see \cite{Kuo-Liu-Tsai}.

\begin{lemma}\label{lem:Spac:Fluc:Decay:in:m:23D} 
\begin{multline*}
	\int
	\absBK{
			\sum_{l=1}^N G(\phi_1,\sigma_1) \cdots
            \frac{\partial G}{\partial\phi}(\phi_l,\sigma_l)
            \cdots G(\phi_N,\sigma_N)
           }
	 d^N\phi
\\
	=
     O\BK{ (N\log N)^\frac12 },
\end{multline*}

Therefore,
\begin{equation}\label{eq:Fr:U3:Esti}
        |U_3|
\leq
        \int_0^\pi \absBK{\frac{dU_3}{d\theta}} d\theta
=
        O(1)
        \sup_{\frac{t}{2}<s<t} \BK{\VertBK{j}_{L^\infty_\bfy}}(s)
        \times \BK{ \frac{\log N}{N} }^\frac12.
\end{equation}
\end{lemma}

Under the assumption $t/mn,t/mN,N \gg 1$, patching \eqref{eq:Fr:U1:Esti}, \eqref{eq:Fr:U2:Esti}, and \eqref{eq:Fr:U3:Esti} together we have
\begin{multline*}
    \Lambda_{(k_1,\ldots,k_N)}^{(N,m)}(\bfy,t) - \Lambda_{(k_1,\ldots,k_N)}^{(N,m)}(\bfy',t)\\
=
     O(1) \BK{\frac{ (mn)^{3}\log(t+1) }{ t^{3} } +(1-\alpha)^m}\BK{ \VertBK{g_{in}}_{\infty,\mu} + \calJ(t) }
    \\
    +
     O(1)
    \sup_{\frac{t}{2}<s<t} \BK{\VertBK{j}_{L^\infty_\bfy}}(s)
    \times
    \BK{ \BK{\frac{\log N}{N}}^\frac12 + \BK{\frac{\log n}{n}}^\frac12 mN }.
\end{multline*}
Plugging this to \eqref{eq:Fr:Flu:Lambda}, we conclude Theorem \ref{thm:Fr:Spac:Fluc}.

\end{subsection}
\begin{subsection}{Convergence of Boundary Flux}\label{sec:Fr:Conv:Bdy:Flux}
In this subsection, we prove our main theorem, Theorem \ref{thm:Fr:Main}, for free molecular flow.

To apply a priori estimate, we need to establish the boundedness of $j$ first: $\VertBK{j}_{L^\infty_\bfy}=O(1)\VertBK{g_{in}}_{\infty,\mu}$, cf. Lemma \ref{lem:Fr:Bddness:of:j}.  We apply a priori estimate {\it twice}: in the first time we obtain a rougher estimate, the boundedness of $j$, and in the second time we use the boundedness of $j$ to obtain the convergence rate $(\alpha t+1)^{-d}+(1-\alpha)^{t^{\frac{1}{400}}}$ of $j$.

Now we recall $j(\bfy,t)=j_{in}(\bfy,t)+j_{mid}(\bfy,t)+j_{fl}(\bfy,t)$, \eqref{eq:Fr:Consv:of:Mole:Numb}, \eqref{eq:jin}, \eqref{eq:jmid} and \eqref{eq:jfl}.

\begin{proposition}\label{prop:Fr:Consv:Esit}
For $t>1$,
\begin{subequations}\label{eq:Fr:Consv:Eq}
\begin{align}
	\label{eq:Fr:Consv:Eq:I}
	 j_{in}(\bfy,t)
=&
    \frac{O(1)}{t^d} \BK{ \VertBK{g_{in}}_{\infty,\mu} + \calJ(t) },
	\\
	\label{eq:Fr:Consv:Eq:II}
	 j_{mid}(\bfy,t)
=&
    O(1)\BK{\Big(\alpha\sum\limits_{i=1}^{\lfloor K/2\rfloor}(1-\alpha)^{i-1}i^{d-1}+
    \sum\limits_{k=1}^K(1-\alpha)^kk^{d-1}\Big)
    \frac{1}{t^d}+(1-\alpha)^{K/2}\frac{K^d}{t^d}} \calJ(t) \\
    \notag
    +&O(1)\sum\limits_{k=1}^K(1-\alpha)^k\frac{k^{d-1}}{t^d}\VertBK{g_{in}}_{\infty,\mu}
    + O(1)\sup_{\frac{t}{2}<s<t} \BK{\VertBK{j}_{L^\infty_\bfy}}(s)\frac{K^d}{t^d},\\
	\label{eq:Fr:Consv:Eq:III}
	 j_{fl}(\bfy,t)
=&
	O(1)\sup_{ \substack{ t'\in \BK{ t-K^p\log(t+1), t } \\ \bfy,\bfy'\in\partial D } }
    \absBK{ j(\bfy,t) - j(\bfy',t') }+O(1)(1-\alpha)^{K^p}\BK{\calJ(t)+\VertBK{g_{in}}_{\infty,\mu} }.
\end{align}
\end{subequations}
\end{proposition}

\begin{proof}\hfil

For $j_{in}$, since $|\bfxi| < \frac{|\bfx-\bfx_{(1)}|}{t} \leq \frac{\diam(D)}{t} = \frac{2}{t}$,
\begin{align*}
&
	|j_{in}(\bfy,t)|
	\\
=&
    \absBK{
            \int_{|\bfxi|<\frac{|\bfx-\bfx_{(1)}|}{t}}
            \sqBK{ \alpha\sum\limits_{i=1}^{\infty}(1-\alpha)^{i-1}j(\bfy,t)
	 		\BK{ \frac{2\pi}{RT(\bfx_{(i)})} }^{\frac{1}{2}} M_{T(\bfx_{(i)})}(\bfxi) 
            - \bar g_{in}(\bfx -\bfxi t,\bfxi) }
             d\bfx d\bfxi
          }
    \\
=&
    O(1)\int_{ |\bfxi| < \frac{2}{t} }
    \sqBK{
            \calJ(t) M(\bfxi)
            +
            \VertBK{ g_{in} }_{\infty,\mu}
            \BK{ \int \frac{1}{(1+|\bfzeta|)^\mu} d\bfeta }
         }
     d\bfxi
    \\
=&
    \frac{O(1)}{t^d}
    \BK{ \calJ(t) + \VertBK{g_{in}}_{\infty,\mu} }.
\end{align*}

By  \eqref{eq:jmid},
\begin{multline*}
j_{mid}(\bfy,t)\leq\frac{1}{C_S|D|} 
      \sum\limits_{k=1}^{K}
      \int_{A_k} 
             \Bigg\{\alpha\sum\limits_{i=1}^{k}(1-\alpha)^{i-1}\sqBK{j(\bfy,t)-j(\bfx_{(i)},t-t_1-...-t_i)}
             \\
             \BK{\frac{2\pi}{RT(\bfx_{(i)})}}^{\frac{1}{2}}M_{T(\bfx_{(i)})}(\bfxi)
             +(1-\alpha)^k\Bigg(\alpha\sum\limits_{i=1}^{\infty}(1-\alpha)^{i-1}j(\bfy,t)
	 		\BK{ \frac{2\pi}{RT(\bfx_{(k+i)})} }^{\frac{1}{2}} M_{T(\bfx_{(k+i)})}(\bfxi)\\
             - \bar g_{in}(\bfx_{(k)}-\bfxi^k(t-t_1-...-t_k),\bfxi^k)\Bigg)\Bigg\}d\bfx d\bfxi
             \equiv I+II.
\end{multline*}
Direct computations yield
\begin{equation*}
\begin{split}
II&=\sum\limits_{k=1}^K\int_{A_k}(1-\alpha)^k\Bigg\{\alpha\sum\limits_{i=1}^{\infty}(1-\alpha)^{i-1}j(\bfy,t)
	 		\BK{ \frac{2\pi}{RT(\bfx_{(k+i)})} }^{\frac{1}{2}} M_{T(\bfx_{(k+i)})}(\bfxi)\\
	 		& \quad\quad\quad\quad\quad\quad\quad\quad\quad\quad
             - \bar g_{in}(\bfx_{(k)}-\bfxi^k(t-t_1-...-t_k),\bfxi^k)\Bigg\}d\bfx d\bfxi\\
             &=O(1)\sum\limits_{k=1}^K(1-\alpha)^k\int_{A_k}\BK{\calJ(t)
             +\VertBK{g_{in}}_{\infty,\mu}}d\bfx d\bfxi\\
             &= O(1)\BK{\calJ(t)+\VertBK{g_{in}}_{\infty,\mu}}
             \begin{dcases}
     \sum\limits_{k=1}^K(1-\alpha)^k\frac{1}{t} \quad \text{ if } d=1,\\
     \sum\limits_{k=1}^K(1-\alpha)^k\frac{k}{t^2}\quad \text{ if } d=2.\\
             \end{dcases}
\end{split}
\end{equation*}
It ia easy to show that
\begin{equation*}
t-t_1-...-t_i\geq \frac{1}{2}t \iff k\geq 2i \quad \text{ on } A_k,
\end{equation*}
and therefore, we rewrite $I$ as
\begin{equation*}
\begin{split}
I&=\sum\limits_{k=1}^K\int_{A_k} \alpha\sum\limits_{i=1}^{k}(1-\alpha)^{i-1}
   \Big(j(\bfy,t)-j(\bfx_{(i)},t-t_1-..-t_i)\Big)\\
    & \quad\quad\quad\quad\quad\quad\quad\quad\quad\quad\times 
   \BK{ \frac{2\pi}{RT(\bfx_{(i)})} }^{\frac{1}{2}} M_{T(\bfx_{(i)})}(\bfxi) d\bfx d\bfxi\\
   &= \alpha\sum\limits_{i=1}^{K}(1-\alpha)^{i-1}\sum\limits_{k=i}^{K}\int_{A_k}\Big(\ldots\Big)\\
&= \alpha \left\{\sum\limits_{i=1}^{\lfloor K/2\rfloor}(1-\alpha)^{i-1}\BK{\sum\limits_{k=i}^{2i-1}+\sum\limits_{k=2i}^{K}}
   +\sum\limits_{i=\lfloor K/2\rfloor +1}^{K}(1-\alpha)^{i-1}\sum\limits_{k=i}^{K}\right\}\Big(\ldots\Big).
\end{split}
\end{equation*}
Direct computations yield
\begin{align*}
&\alpha\sum\limits_{i=\lfloor K/2\rfloor +1}^{K}(1-\alpha)^{i-1}\sum\limits_{k=i}^{K}\Big(\ldots\Big)\\
&=O(1)\calJ(t)\sum\limits_{i=\lfloor K/2 \rfloor +1}^{K}(1-\alpha)^{i-1}\alpha\sum\limits_{k=i}^{K}\int_{A_k}
\BK{ \frac{2\pi}{R\Tm} }^{\frac{1}{2}} M(\bfxi) d\bfxi d\bfx\\
&=O(1)\calJ(t)(1-\alpha)^{K/2}\frac{K^d}{t^d},
\end{align*}

\begin{align*}
&
\sum\limits_{i=1}^{\lfloor K/2\rfloor}(1-\alpha)^{i-1}\alpha\sum\limits_{k=2i}^{K}\int_{A_k}
   \Big(j(\bfy,t)-j(\bfx_{(i)},t-t_1-..-t_i)\Big)
   \BK{ \frac{2\pi}{RT(\bfx_{(i)})} }^{\frac{1}{2}} M_{T(\bfx_{(i)})}(\bfxi) d\bfxi d\bfx   \\
=&
     O(1)\sup_{\frac{t}{2}<s<t} \BK{\VertBK{j}_{L^\infty_\bfy}}(s)
   \sum\limits_{i=1}^{\lfloor K/2\rfloor}(1-\alpha)^{i-1}\alpha\sum\limits_{k=2i}^{K}\int_{A_k}
\BK{ \frac{2\pi}{R\Tm} }^{\frac{1}{2}} M(\bfxi) d\bfxi d\bfx\\
=&
O(1)\sup_{\frac{t}{2}<s<t} \BK{\VertBK{j}_{L^\infty_\bfy}}(s)\frac{K^d}{t^d},
\end{align*}

\begin{align*}
&\sum\limits_{i=1}^{\lfloor K/2\rfloor}(1-\alpha)^{i-1}\alpha\sum\limits_{k=i}^{2i-1}\int_{A_k}
   \Big(j(\bfy,t)-j(\bfx_{(i)},t-t_1-..-t_i)\Big)
   \BK{ \frac{2\pi}{RT(\bfx_{(i)})} }^{\frac{1}{2}} M_{T(\bfx_{(i)})}(\bfxi) d\bfxi d\bfx   \\
&=O(1) \calJ(t)\sum\limits_{i=1}^{\lfloor K/2\rfloor }(1-\alpha)^{i-1}\alpha\sum\limits_{k=i}^{2i-1}\int_{A_k}
\BK{ \frac{2\pi}{R\Tm} }^{\frac{1}{2}} M(\bfxi) d\bfxi d\bfx\\
&=O(1)\calJ(t)\sum\limits_{i=1}^{\lfloor K/2\rfloor}(1-\alpha)^{i-1}\alpha
\begin{dcases}
\frac{i}{t} \quad \text{ if } d=1,\\
\frac{i^2}{t^2} \quad \text{ if } d=2.
\end{dcases}
\end{align*}

For $j_{fl}$,
\begin{multline*}
 \int_{|\bfxi|>\frac{|\bfx_{(1)}-\bfx_{(2)}|}{\log(t+1)}} 
 \Bigg\{\alpha\sum\limits_{i=K+1}^{\infty}(1-\alpha)^{i-1}j(\bfy,t)
             \BK{\frac{2\pi}{RT(\bfx_{(i)})}}^{\frac{1}{2}}M_{T(\bfx_{(i)})}(\bfxi)\\
          -(1-\alpha)^K \bar g(\bfx_{(K)}-\bfxi^K(t-t_1-...-t_K),\bfxi^K)\Bigg\}d\bfx d\bfxi\\
             = O(1)(1-\alpha)^{K}\BK{\calJ(t)+\VertBK{g_{in}}_{\infty,\mu} },
\end{multline*}

\begin{multline*}
\intLim_{|\bfxi|>\frac{|\bfx_{(1)}-\bfx_{(2)}|}{\log(t+1)}} 
             \alpha\sum\limits_{i=1}^{K}(1-\alpha)^{i-1}\Big(j(\bfy,t)-j(\bfx_{(i)},t-t_1-...-t_i)\Big)
             \BK{\frac{2\pi}{RT(\bfx_{(i)})}}^{\frac{1}{2}}M_{T(\bfx_{(i)})}(\bfxi)d\bfx d\bfxi\\
             =\intLim_{|\bfxi|>\frac{|\bfx_{(1)}-\bfx_{(2)}|}{\log(t+1)}} 
             \alpha\Big(\sum\limits_{i=1}^{\lfloor K^p\rfloor}+\sum\limits_{i=\lfloor K^p\rfloor}^{K}\Big)
             \Big(\ldots\Big)\\
             =O(1)\sup_{ \substack{ t'\in \BK{ t-K^p\log(t+1), t } \\ \bfy,\bfy'\in\partial D } }
    \absBK{ j(\bfy,t) - j(\bfy',t') }+O(1)(1-\alpha)^{K^p}\BK{\calJ(t)+\VertBK{g_{in}}_{\infty,\mu} },
\end{multline*}
where $0<p<1$ is any fixed number.
\end{proof}

With the estimates  \eqref{eq:Fr:Consv:Eq:I}, \eqref{eq:Fr:Consv:Eq:II}, and \eqref{eq:Fr:Consv:Eq:III}, the main task is to study the RHS of  \eqref{eq:Fr:Consv:Eq:III},  the fluctuation of $j$. We have treated this in Sections \ref{sec:Fr:Stoch:Formulation}, \ref{sec:Fr:Temp:Fluc:Esti}, and \ref{sec:Fr:Spac:Fluc}. The fluctuation estimate, the following theorem, follows directly from  Corollary \ref{cor:Fr:Temp:Fluc} on the temporal fluctuation estimate, and Theorem \ref{thm:Fr:Spac:Fluc} on the spacial fluctuation estimate.

\begin{theorem}[Fluctuation estimate]\label{thm:Fr:Fluc}
Let $t'<t$. $\bfy,\bfy'\in\partial D$, then for sufficiently large $t'/mn,\ t'/mN$ and $N$,
\begin{subequations}\label{eq:Fr:Fluc:Esti}
\begin{multline}
     j(\bfy,t) - j(\bfy',t')
	\\
=
     O(1) \BK{ \sup_{\frac{t'}{2}<s<t} \BKK{ \VertBK{j}_{L^{\infty}_{\bfy}}}(s) }
    \times
    \BK{ \BK{\frac{\log N}{N}}^\frac12 + \BK{\frac{\log n}{n}}^\frac12 mN }
    (t-t')
    \\
+
     O(1)
    \BK{\frac{ (m^4n^{3}+m^3N^2)\log(t'+1) }{ t'^{3} } +(1-\alpha)^m}
    \BK{ \VertBK{g_{in}}_{\infty,\mu} + \calJ(t) },\\
    \text{ when } d=2,
\end{multline}
\begin{multline}
|j(\bfy,t)-j(\bfy',t')|=O(1)\BK{\frac{m^3n^{2}\log(t'+1)}{t'^{2}}+(1-\alpha)^m}
\BK{\calJ(t)+\VertBK{g_{in}}_{\infty,\mu} }\\
+O(1)\BK{\BK{ \frac{(t-t')+t^q}{n^\frac12} } +\BK{ \frac{m}{t^q} }^2}
     \sup_{ \frac{t'}{2} < s < t } \Big( \absBK{j_+(s)} + \absBK{j_-(s)} \Big)\\
     \text{ when } d=1.
\end{multline}
\end{subequations}
\end{theorem}
Here we will complete the proof Theorem \ref{thm:Fr:Main} using Theorem \ref{thm:Fr:Fluc}.

Recall $K\approx\frac{t}{\log(t+1)}$ as $t\gg1$. And hence for each $0<p<1$, $K^p\log(t+1)=o(t)$ as $t\gg1$.
Plugging Theorem \ref{thm:Fr:Fluc} into \eqref{eq:Fr:Consv:Eq:III}, we obtain the following estimate of $j_{fl}$:
\begin{subequations}\label{eq:Fr:Consv:Eq:III:Flu}
\begin{multline}
%\tag{\ref{eq:Fr:Consv:Eq:III}'}
     j_{fl}(\bfy,t)
=
     O(1)
    \BK{\frac{ (m^4n^{3}+m^3N^2)\log(t+1) }{ t^{3} } +(1-\alpha)^m +(1-\alpha)^{\BK{\frac{t}{\log(t+1)}}^p}}
    \\ \times
    \BK{ \VertBK{g_{in}}_{\infty,\mu} + \calJ(t) }
    \\
    +
     O(1) \BK{ \sup_{ t/2 < s < t } 
    \BKK{ \VertBK{j}_{L^{\infty}_{\bfy}}}(s) }
    \BK{ \BK{\frac{\log N}{N}}^\frac12 + \BK{\frac{\log n}{n}}^\frac12 mN }
    \BK{\frac{t}{\log(t+1)}}^p\log(t+1),\\
    \text{ when } d=2,
\end{multline}
\begin{multline}
%\tag{\ref{eq:Fr:Consv:Eq:III}'}
j_{fl}(\bfy,t)=O(1)\BK{\frac{m^3n^{2}\log(t+1)}{t^{2}}+(1-\alpha)^m+(1-\alpha)^{\BK{\frac{t}{\log(t+1)}}^p}}
\BK{\calJ(t)+\VertBK{g_{in}}_{\infty,\mu} }\\
+O(1)\BK{\frac{\BK{ \BK{\frac{t}{\log(t+1)}}^p\log(t+1)+t^q}}{n^\frac12} } +\BK{ \frac{m}{t^q} }^2
     \sup_{ t/2 < s < t } \Big( \absBK{j_+(s)} + \absBK{j_-(s)} \Big)\\
     \text{ when } d=1.
\end{multline}
\end{subequations}
We first establish the unifrom boundedness of $j$:

\begin{lemma}\label{lem:Fr:Bddness:of:j}
The boundary flux $j$ is uniformly bounded:
 \begin{equation}\label{uniform}
  \calJ(t) = O(1) \VertBK{g_{in}}_{\infty,\mu}, \ t \geq 0. 
  \end{equation}
\end{lemma}

\begin{proof}\hfil

From \eqref{eq:Fr:Consv:Eq:I}, \eqref{eq:Fr:Consv:Eq:II}, and \eqref{eq:Fr:Consv:Eq:III:Flu}, we have for $t>1$,
\begin{equation*}
\begin{split}
	 j_{in}(\bfy,t)
=&
    \frac{O(1)}{t^d} \BK{ \VertBK{g_{in}}_{\infty,\mu} + \calJ(t) },
    \\
	 j_{mid}(\bfy,t)
=&
    O(1)\frac{1}{(\log(1+t))^d}\BK{\VertBK{g_{in}}_{\infty,\mu}+\calJ(t)},
   \\
	 j_{fl}(\bfy,t)
=&
      O(1)
    \Bigg(\frac{ m^{d+2}n^{d+1}\log(t+1) }{ t^{d+1} } +(1-\alpha)^m +(1-\alpha)^{\BK{\frac{t}{\log(t+1)}}^p}\\
    +&\curBK{ \begin{array}{c@{,~}l}
    \frac{ m^3N^2\log(t+1) }{ t^{3} }
                        &
                         d=2,
                        \\
                       0
                        &
                         d=1
                    \end{array} }
    \Bigg)
    \BK{\VertBK{g_{in}}_{\infty,\mu}+\calJ(t)} 
    \\
    +&\curBK{ \begin{array}{c@{,~}l}
    \frac{ m^3N^2\log(t+1) }{ t^{3} }+
                       \BK{ \BK{\frac{\log N}{N}}^\frac12 + \BK{\frac{\log n}{n}}^\frac12 mN }
    \BK{\frac{t}{\log(t+1)}}^p\log(t+1)
                        &
                         d=2,
                        \\
                       \frac{\BK{\frac{t}{\log(t+1)}}^p\log(t+1)+t^q}{n^{1/2}}+\BK{\frac{m}{t^q}}^2
                        &
                         d=1
                    \end{array} }\calJ(t).
\end{split}
\end{equation*}
Therefore,
\begin{equation*}
\begin{split}
 &
      j(\bfy,t)
     \\
=&
	  j_{in}(\bfy,t) + j_{mid}(\bfy,t)  + j_{fl}(\bfy,t)
	 \\
=&
      O(1)
     \Bigg[\frac{1}{(\log(t+1))^d}+(1-\alpha)^{\BK{\frac{t}{\log(t+1)}}^p} +(1-\alpha)^m + \frac{ m^{d+2}n^{d+1} \log(t+1) }{ t^{d+1} } \\
      \hphantom{O(1)aa\Bigg[}
     +&\curBK{ \begin{array}{c@{,~}l}
                       \frac{ m^{3}N^2 \log(t+1) }{ t^{3} }
                        &
                         d=2,
                        \\
                      0
                        &
                         d=1
                    \end{array} }
                    \Bigg]
     \VertBK{g_{in}}_{\infty,\mu}
     \\
    +&
     O(1)
     \Bigg[
            \frac{1}{(\log(1+t))^d}+(1-\alpha)^{\BK{\frac{t}{\log(t+1)}}^p} +(1-\alpha)^m + \frac{ m^{d+2}n^{d+1} \log(t+1) }{ t^{d+1} } 
     \\
	 \hphantom{O(1)aa\Bigg[}
            +&
            \curBK{ \begin{array}{c@{,~}l}
                        \frac{ m^{3}N^2 \log(t+1) }{ t^{d+1} }+\BK{ \BK{\frac{\log n}{n}}^\frac12 mN + \BK{ \frac{\log N}{N} }^\frac12 } 
                        \BK{\frac{t}{\log(t+1)}}^{p}\log(t+1)
                        &
                         d=2,
                        \\
                        \BK{ \frac{\BK{\frac{t}{\log(t+1)}}^p\log(t+1)+t^q}{n^\frac12} } +\BK{ \frac{m}{t^q} }^2
                        &
                         d=1
                    \end{array} }
     \Bigg]
     \calJ(t).
\end{split}
\end{equation*}

So far we only have to assume $t/mn,\ t/Nm,\ N\gg 1$. 
Now we set $n=n(t)=\lfloor t^{r_1}\rfloor$, $m=m(t)=\lfloor t^{r_2}\rfloor$ and $N=N(t)=\lfloor t^{r_3}\rfloor$
, where $0<r_1,r_2,r_3<1$ are to be determined. In order to get
\begin{equation*}
\begin{split}
&
    \lim_{t\rightarrow\infty}
    \BK{
           \frac{ 1 }{ (\log(1+ t))^d } +(1-\alpha)^{\BK{\frac{t}{\log(t+1)}}^p} +(1-\alpha)^m + \frac{ m^{d+2}n^{d+1} \log(t+1) }{ t^{d+1} } 
        }
= 0,
    \\
&
    \lim_{t\rightarrow\infty}
            \curBK{ \begin{array}{c@{,~}l}
                        \frac{ m^{3}N^2 \log(t+1) }{ t^{d+1} }+\BK{ \BK{\frac{\log n}{n}}^\frac12 mN + \BK{ \frac{\log N}{N} }^\frac12 } 
                        \BK{\frac{t}{\log(1+t)}}^p\log(1+t)
                        &
                         d=2,
                        \\
                       \BK{ \frac{\BK{\frac{t}{\log(t+1)}}^p\log(t+1)+t^q}{n^\frac12} } +\BK{ \frac{m}{t^q} }^2
                        &
                         d=1
                    \end{array} }
=0,
\end{split}
\end{equation*}
we need
\begin{equation*}
\begin{split}
&r_2+r_3+p<\frac{1}{2}r_1,\\
&p<\frac{1}{2}r_3,\\
&r_1+r_2<1,\\
&r_2+r_3<1,\\
&r_2<q,\\
&p,q<\frac{1}{2}r_1.
\end{split}
\end{equation*}
This can be done by choosing any $0<r_1<6/7$ and setting $q=r_1/3, r_2=r_1/6, r_3=r_1/12, p=r_1/36$.
Therefore, there exists $t_*>0$ such that, for all $t>t_*,$
\begin{equation}\label{eq:z4}
		\left\{	\begin{array}{c@{,~}l}
       				\BK{ \VertBK{j}_{L^\infty_\bfy} }(t) 		&   d=2 \\
       				\Big(\absBK{j_+(t)} + \absBK{j_-(t)}\Big)	&   d=1
       			\end{array}	\right\}
\leq
        O(1) \VertBK{g_{in}}_{\infty,\mu}
       +
       \frac{1}{2} \calJ(t).
\end{equation}
Moreover, from \eqref{eq:Fr:Coarse:Estimate}, now that $t_*$ is fixed, $j$ is bounded by a constant multiple of $\VertBK{g_{in}}_{\infty,\mu}$ for all $0\leq t\leq t_*$. Hence, \eqref{eq:z4} actually holds for all $t$, and this implies that $\frac{1}{2}\calJ(t) = O(1) \VertBK{g_{in}}_{\infty,\mu}$ and the lemma is proved.
\end{proof}

With the boundedness of $j$, \eqref{uniform}, we can perform the second a priori estimate to obtain the $(\alpha t+1)^{-d}+(1-\alpha)^{t^{\frac{1}{400}}}$ decay of $j$, Theorem \ref{thm:Fr:Main}. First, since $\calJ = O(1)\Vert g_{in}\Vert_{\infty,\mu}$ and the following Lemma \ref{lemma:series}, we can rewrite our previous estimates \eqref{eq:Fr:Consv:Eq:I}, \eqref{eq:Fr:Consv:Eq:II}, and \eqref{eq:Fr:Consv:Eq:III:Flu} as the following in Proposition \ref{pro:Fr:refined}.
We will need the following identities, whose simple proof is omitted. 
\begin{lemma}\label{lemma:series}
For $0<x<1$,
\begin{equation*}
\begin{split}
&\sum\limits_{k=1}^\infty kx^k = \frac{x}{(1-x)^2},\\
&\sum\limits_{k=1}^\infty k^2x^k = \frac{x(x+1)}{(1-x)^3}.\\
\end{split}
\end{equation*}
\end{lemma}

\begin{proposition}\label{pro:Fr:refined}
For $t>1$,
\begin{subequations}\label{eq:Fr:Consv:Eq:Refined}
\begin{align}
     j_{in}(\bfy,t) 
=&
    \frac{O(1)}{ t^d} \VertBK{g_{in}}_{\infty,\mu},
	\\
     j_{mid}(\bfy,t) 
=&
    O(1)\BK{\frac{1}{(\alpha t)^d} +(1-\alpha)^{\BK{\frac{t}{\log(t+1)}}}\frac{1}{(\log(1+ t))^d}}\VertBK{g_{in}}_{\infty,\mu}\\
    \notag +&
    O(1)\frac{1}{(\log(1+t))^d} \sup_{\frac{t}{2}<s<t} \BK{\VertBK{j}_{L^\infty_\bfy}}(s),
\end{align}
and
\begin{multline}\label{eq:Fr:Consv:Eq:Refined:III}
    j_{fl}(\bfy,t)
=
     O(1)
    \BK{\frac{ (m^4n^{3}+m^3N^2)\log(t+1) }{ t^{3} } +(1-\alpha)^m +(1-\alpha)^{\BK{\frac{t}{\log(t+1)}}^p}}
    \VertBK{g_{in}}_{\infty,\mu} 
    \\
    +
     O(1) \BK{ \sup_{ t/2 < s < t } 
    \BKK{ \VertBK{j}_{L^{\infty}_{\bfy}}}(s) }
    \BK{ \BK{\frac{\log N}{N}}^\frac12 + \BK{\frac{\log n}{n}}^\frac12 mN }
    \BK{\frac{t}{\log(1+t)}}^p\log(1+t),\\
    \text{ when } d=2,
\end{multline}
\begin{multline*}
\tag{\ref{eq:Fr:Consv:Eq:Refined:III}}
j_{fl}(\bfy,t)=O(1)\BK{\frac{m^3n^{2}\log(t+1)}{t^{2}}+(1-\alpha)^m+(1-\alpha)^{\BK{\frac{t}{\log(t+1)}}^p}}
\VertBK{g_{in}}_{\infty,\mu} \\
+O(1)\BK{\frac{\BK{ \BK{\frac{t}{\log(t+1)}}^p\log(1+t)+t^q}}{n^\frac12}  +\BK{ \frac{m}{t^q} }^2}
     \sup_{ t/2 < s < t } \Big( \absBK{j_+(s)} + \absBK{j_-(s)} \Big)\\
     \text{ when } d=1.
\end{multline*}
\end{subequations}
\end{proposition}
Now we are ready to prove Theorem \ref{thm:Fr:Main}.
\begin{definition}
The a priori norm of $j$ is a function of $t$ defined as
\begin{equation}\label{priorinorm}
    \calN(t)
\equiv
    \sup_{ 0 \leq s \leq t }
    \begin{dcases}
           \BK{(\alpha s)^{-2}+(1-\alpha)^{s^{\frac{1}{400}}}}^{-1} \BK{ \VertBK{ j }_{L^\infty_\bfy} }(s)    & \text{ for }   d=2, \\
           \BK{(\alpha s)^{-1}+(1-\alpha)^{s^{\frac{1}{400}}}}^{-1} \Big(\absBK{j_+(s)} + \absBK{j_-(s)}\Big)   & \text{ for }   d=1,
    \end{dcases}
\end{equation}
where $j_\pm(s)=j(\pm1,s)$.
\end{definition}

\begin{proof}{\textit{of Theorem \ref{thm:Fr:Main}} }\hfill

From \eqref{eq:Fr:Consv:Eq:Refined} and for $t>1$, 
\begin{align*}
     j_{mid}(\bfy,t) 
=&
    O(1)\BK{\frac{1}{(\alpha t)^d} +(1-\alpha)^{\BK{\frac{t}{\log(t+1)}}}\frac{1}{(\log(1+ t))^d}}\VertBK{g_{in}}_{\infty,\mu}\\
    +&
    O(1)\frac{1}{(\log(1+t))^d} \BK{(\alpha t)^{-d}+(1-\alpha)^{t^{\frac{1}{400}}}}\calN(t),
\end{align*}
and
\begin{multline*}
j_{fl}(\bfy,t)
=
     O(1)
    \Bigg[\frac{ m^{d+2}n^{d+1}\log(t+1) }{ t^{d+1} } +(1-\alpha)^m +(1-\alpha)^{K^p}
    +\curBK{	\begin{array}{c@{,~}l}
    			\frac{ m^{3}N^2\log(t+1) }{ t^{3} }
    			&
    			 d=2
    			\\
				0
				&
				 d=1  			
    		\end{array} }
    		\Bigg]    
    \VertBK{g_{in}}_{\infty,\mu} 
    \\
    +
     O(1) \BK{(\alpha t)^{-d}+(1-\alpha)^{t^{\frac{1}{400}}}}\calN(t)
\curBK{	\begin{array}{c@{,~}l}
    			\BK{ \BK{\frac{\log N}{N}}^\frac12 + \BK{\frac{\log n}{n}}^\frac12 mN }
    \BK{\frac{t}{\log(1+t)}}^p\log(1+t),
    			&
    			 d=2
    			\\
				\BK{ \frac{\BK{\frac{t}{\log(1+t)}}^p\log(1+t)+t^q}{n^\frac12} } +\BK{ \frac{m}{t^q} }^2
				&
				 d=1  			
    		\end{array} }.    
    \end{multline*}
Therefore,
\begin{align*}
&
      j(\bfy,t)
     \\
=&
     O(1)
    \Bigg[\frac{1}{(\alpha t)^d} +(1-\alpha)^{\BK{\frac{t}{\log(t+1)}}}\frac{1}{(\log(1+t))^{d}}+\frac{ m^{d+2}n^{d+1}\log(t+1) }{ t^{d+1} } +(1-\alpha)^m \\
    &+(1-\alpha)^{\BK{\frac{t}{\log(t+1)}}^p}
     +\curBK{	\begin{array}{c@{,~}l}
    			\frac{ m^{3}N^2\log(t+1) }{ t^{3} }
    			&
    			 d=2
    			\\
				0
				&
				 d=1  			
    		\end{array} }\Bigg]
    \VertBK{g_{in}}_{\infty,\mu} 
     \\
&
    +
     O(1) \BK{(\alpha t)^{-d}+(1-\alpha)^{t^{\frac{1}{400}}}} \calN(t) 
     \\
&
	 \hphantom{+}
     \times
     \sqBK{
            \frac{1}{(\log(1+t))^d}
            +
            \curBK{ \begin{array}{c@{,~}l}
                        \BK{ \BK{\frac{\log N}{N}}^\frac12 + \BK{\frac{\log n}{n}}^\frac12 mN }
    \BK{\frac{t}{\log(1+t)}}^p\log(1+t)
                        &
                         d=2
                        \\
                       \BK{ \frac{\BK{\frac{t}{\log(1+t)}}^p\log(1+t)+t^q}{n^\frac12} } 
                       +\BK{ \frac{m}{t^q }}^2
                        &
                         d=1
                    \end{array} }
     	} .
\end{align*}
Now we choose $0<r_1,r_2,r_3\ll 1$ and set $n=n(t)=\lfloor t^{r_1}\rfloor$, $m=m(t)=\lfloor t^{r_2}\rfloor$, $N=N(t)=\lfloor t^{r_3}\rfloor$  so that
\begin{equation*}
\begin{split}
&
\Bigg[(1-\alpha)^{\BK{\frac{t}{\log(t+1)}}}\frac{1}{(\log(1+t))^{d}}+\frac{ m^{d+2}n^{d+1}\log(t+1) }{ t^{d+1} } +(1-\alpha)^m +(1-\alpha)^{\BK{\frac{t}{\log(t+1)}}^p}\\
     &+\curBK{	\begin{array}{c@{,~}l}
    			\frac{ m^{3}N^2\log(t+1) }{ t^{3} }
    			&
    			 d=2
    			\\
				\frac{m^2}{ t^2}
				&
				 d=1  			
    		\end{array} }\Bigg]
= O\Big((\alpha t)^{-d}+(1-\alpha)^{t^{\frac{p}{2}}}\Big),
    \\
&
    \lim_{t\rightarrow\infty}
            \frac{1}{(\log(1+t))^d}
            +
            \curBK{ \begin{array}{c@{,~}l}
                        \BK{ \BK{\frac{\log N}{N}}^\frac12 + \BK{\frac{\log n}{n}}^\frac12 mN }
    \BK{\frac{t}{\log(1+t)}}^p\log(1+t)
                        &
                         d=2
                        \\
                       \BK{ \frac{\BK{\frac{t}{\log(1+t)}}^p\log(1+t)+t^q}{n^\frac12} } 
                       +\BK{ \frac{m}{t^q }}^2
                        &
                         d=1
                    \end{array} }
=0.
\end{split}
\end{equation*}
This can be done if $0<r_1,r_2,r_3,p$, and $q$ satisfy
\begin{equation*}
\begin{split}
&r_2+r_3+p<\frac{1}{2}r_1,\\
&p<\frac{1}{2}r_3,\\
&\frac{p}{2}<r_2,\\
&r_1+r_2<\frac{1}{d+2},\\
&r_2+r_3<\frac{1}{d+1},\\
&r_2<q,\\
&p,q<\frac{1}{2}r_1.
\end{split}
\end{equation*}
In fact, we may choose $p=\frac{1}{200},r_3=\frac{1}{90},r_2=\frac{1}{30},q=\frac{1}{25},r_1=\frac{1}{10}$ so that the above inequalities hold.
Consequently, we can find some sufficiently large $t_*>0$ such that for all $t>t_*$
\begin{equation}\label{eq:z6}
       \curBK{ \begin{array}{c@{,~}l}
                    \BK{ \VertBK{j}_{L^\infty_\bfy} }(t) &   d=2 \\
                    \Big(\absBK{j_+(t)} + \absBK{j_-(t)}\Big)             &   d=1
                \end{array}}
\leq
        \Big((\alpha t)^{-d}+(1-\alpha)^{t^{\frac{1}{400}}}\Big)
        \BK{O(1)  \VertBK{g_{in}}_{\infty,\mu} 
       +
       \frac{1}{2} \calN(t)}.
\end{equation}
Therefore,
\begin{equation}\label{eq:z7}
       \curBK{ \begin{array}{c@{,~}l}
                    \Big((\alpha t)^{-2}+(1-\alpha)^{t^{\frac{1}{400}}}\Big)^{-1}\BK{ \VertBK{j}_{L^\infty_\bfy} }(t) &   d=2 \\
                    \Big((\alpha t)^{-1}+(1-\alpha)^{t^{\frac{1}{400}}}\Big)^{-1}\Big(\absBK{j_+(t)} + \absBK{j_-(t)}\Big)             &   d=1
                \end{array}}
\leq
        O(1)  \VertBK{g_{in}}_{\infty,\mu} 
       +
       \frac{1}{2} \calN(t).
\end{equation}
Moreover, by Lemma \ref{lem:Fr:Bddness:of:j}, for all $t<t_*$
\begin{equation}\label{eq:z71}
       \curBK{ \begin{array}{c@{,~}l}
                    \Big((\alpha t)^{-2}+(1-\alpha)^{t^{\frac{1}{400}}}\Big)^{-1}\BK{ \VertBK{j}_{L^\infty_\bfy} }(t)        &   d=2 \\
                    \Big((\alpha t)^{-1}+(1-\alpha)^{t^{\frac{1}{400}}}\Big)^{-1}\Big(\absBK{j_+(t)} + \absBK{j_-(t)}\Big)   &   d=1
                \end{array}}
=
	 O(1) \VertBK{g_{in}}_{\infty,\mu}.
\end{equation}
Hence, \eqref{eq:z7} actually holds for all $t$, and this implies $\calN(t)= O(1) \VertBK{g_{in}}_{\infty,\mu}$. 
From this estimate, the definition of $\calN(t)$, \eqref{priorinorm}, and the boundedness of $j(\bfy,t)$, \eqref{uniform},
it is easy to see that
\begin{equation*}
j(\bfy,t)=O(1) \VertBK{g_{in}}_{\infty,\mu}\BK{\frac{1}{(1+\alpha t)^d}+(1-\alpha)^{t^{\frac{1}{400}}}},
\end{equation*} 
and the theorem is proved. 
\end{proof}

\end{subsection}
\end{section}

\begin{section}{Damped Free Molecular Flow}\label{sec:DF}

In this section we consider the damped free molecular flow, with a goal of proving our main theorem, Theorem \ref{thm:FullBz:Main:1} and Theorem \ref{thm:FullBz:Main:2}. As noted in Section 1, we treat nondimensionalized Boltzmann equation \eqref{eq:FullBz:Eq}, so from now on we set $\TM=1$.

We first review some basic properties of the collision operator $Q$ and the linearized collision operator $L$. Both operators act on $\bfzeta$ but not on $\bfx$. Hence we will frequently neglect the spacial dependence in the following discussion.

\begin{subsection}{Preliminaries}\hfil

Recall the collision operator $Q(\cdot,\cdot)$:
\begin{multline*}
     Q(g,h)(\bfzeta)
	=
    \frac12 \intLim_{S^2\times\bbR^3}
    \Big(
                g(\bfzeta') h(\bfzeta'_*)   + h(\bfzeta') g(\bfzeta'_*)
            -   g(\bfzeta) h(\bfzeta_*)     - h(\bfzeta) g(\bfzeta_*)
    \Big)
\\
    \times
     B(\theta,|\bfzeta_*-\bfzeta|) d\Omega d\bfzeta_*.
\end{multline*}
As intermolecular collision conserves total molecular number, total momentum, and total energy, we have the following identity, cf. \cite{Sone},
\begin{equation}\label{eq:FullBz:Collision:Invar:Q}
    \int_{\bbR^3}
    \curBK{ \begin{array}{c}
                1   \\ \bfzeta  \\  |\bfzeta|^2
            \end{array} }
    Q(g,h)(\bfzeta) ~ d\bfzeta
=
    \curBK{ \begin{array}{c}
                0   \\ {\bf 0}  \\  0
            \end{array} }.
\end{equation}
$1,\ \bfzeta, \ |\bfzeta|^2$ are called the collision invariances.

In this paper we assume an inverse power hard potential with Grad's angular cut-off or hard spheres. Under this model, $ B(\theta,|\bfzeta_*-\bfzeta|)\sim |\bfzeta-\bfzeta_*|^{\frac{u-4}{u}}|\cos\theta|$, for some $u\geq 4$.
For the linearized collision operator $Lf = \frac{2}{\sqrt M} Q(f\sqrt M,M)$. As a direct consequence of \eqref{eq:FullBz:Collision:Invar:Q},
\begin{equation}\label{eq:LB:Collision:Invar:L}
    \int_{\bbR^3}
    \curBK{ \begin{array}{c}
                1   \\ \bfzeta  \\  |\bfzeta|^2
            \end{array} }
    \sqrt M(\bfzeta)
    (Lf)(\bfzeta) ~ d\bfzeta
=
    \curBK{ \begin{array}{c}
                0   \\ {\bf 0}  \\  0
            \end{array} }.
\end{equation}

Recall that $L$ can be decomposed as the difference of an integral operator $K$ and a multiplicative operator $\nu$: $L=K-\nu$.
Moreover, we have the following estimates, cf. \cite{Grad-Boltzmann-II}:
\begin{align}
        \label{eq:LB:nu:Esti}
        \nu(\bfzeta)
\sim&
         (1+|\bfzeta|)^{\frac{u-4}{u}},
        \\
        \label{eq:LB:K:Ptws:Esti}
         K(\bfzeta,\bfzeta_*)
=&
         O(1) |\bfzeta-\bfzeta_*|^{\frac{u-4}{u}} e^{-\frac{|\bfzeta-\bfzeta_*|^2}{8}}
         +
         O(1) \frac{ 1 }{ |\bfzeta-\bfzeta_*| }
         e^{-\frac{|\bfzeta-\bfzeta_*|^2}{8}}.
\end{align}

%In view of \eqref{eq:LB:nu:Esti}, \eqref{eq:FullBz:Linfty:Ineq:Q:pre} can be rewritten as
And it is well-known that 
\begin{subequations}
\begin{align}
    \label{eq:FullBz:Linfty:Ineq:Q}
&
    \VertBK{
                \frac{ Q\BK{ \phi\sqrt M, \psi\sqrt M } }{ \nu \sqrt M }
            }_{L^\infty_\bfzeta}
=
     O(1) \VertBK{ \phi }_{L^\infty_\bfzeta} \times \VertBK{ \psi }_{L^\infty_\bfzeta},
    \\
    \label{eq:FullBz:Linfty:Ineq:Q:prime}
&
        \VertBK{
                \frac{ Q\BK{ \phi\sqrt M, \psi\sqrt M } }{ \sqrt M }
            }_{L^\infty_\bfzeta}
        \\
        \notag
=&
     O(1)\BK{ \VertBK{ \nu \phi }_{L^\infty_\bfzeta} \times \VertBK{ \psi }_{L^\infty_\bfzeta} + \VertBK{ \phi }_{L^\infty_\bfzeta} \times \VertBK{ \nu \psi }_{L^\infty_\bfzeta} }.
\end{align}
\end{subequations}

Let
\begin{equation*}
    \VertBK{ f }_{L^{\infty,b}_\bfzeta} \equiv \mathop{\esssup}_{ \bfzeta\in\bbR^3 } 
    (1+|\bfzeta|)^b |f(\bfzeta)|.
\end{equation*}
Also, for any $\gamma\geq 0$,
\begin{equation}\label{eq:LB:K:Norm:Esti}
    \VertBK{ Kf }_{L^{\infty,-\gamma}_\bfzeta}
=\VertBK{
                \frac{ \int f(\bfzeta_*) K(\bfzeta,\bfzeta_*) d\bfzeta_* }{ (1+|\bfzeta|)^\gamma }
            }_\infty=
    O(1) 2^\gamma \VertBK{ f }_{L^{\infty,-\gamma}_\bfzeta}.
\end{equation}
\eqref{eq:LB:K:Norm:Esti} with $\gamma=0$ states that $K: L^\infty_\bfzeta\rightarrow L^\infty_\bfzeta$ is a bounded operator. Actually, more is true: $K: L^{\infty,\beta}_\bfzeta\rightarrow L^{\infty,\beta +1}_\bfzeta$ is a bounded operator for each $\beta\geq 0$,  \cite{Grad-Boltzmann-II}. However, \eqref{eq:LB:K:Norm:Esti} suffices for our purpose. From \eqref{eq:LB:K:Norm:Esti},
\begin{equation}\label{eq:LB:Linfty:Ineq:L}
    \VertBK{ Lf }_{L^\infty_\bfzeta}
=
     O(1) \VertBK{ \nu f }_{L^\infty_\bfzeta}.
\end{equation}
Let $f_\varepsilon (\bfzeta) =
\bbbone_{\curBK{|\bfxi|<\varepsilon}} (1+|\bfzeta|)^\gamma$ for
$\varepsilon\ll 1$. By \eqref{eq:LB:K:Ptws:Esti} we have
\begin{equation}\label{eq:LB:K:low:speed:Esti}
    \VertBK{ K \Big (f_\varepsilon \Big) }_{ L^{\infty,-\gamma}_\bfzeta }
=
    \curBK{ \begin{array}{c@{,~}l}
                    O(\varepsilon)                          &   d=1    \\
                    O(\varepsilon^2 |\log \varepsilon|)     &   d=2    
            \end{array} },
    \text{ for } \varepsilon \ll 1.
\end{equation}
The details can be found in \cite{Kuo-Liu-Tsai-2}.

As mentioned in Section 1, we take \eqref{eq:LB:Eq} as our linearized problem. To show an exponential decay property of \eqref{eq:LB:Eq}, we conduct two reductions: first we reduce $(\partial_t+\sum\zeta_i\partial_{x_i}-\frac{1}{\kappa}L)$ to $(\partial_t+\sum\zeta_i\partial_{x_i}+\frac{\nu}{\kappa})$, and then reduce $(\partial_t+\sum\zeta_i\partial_{x_i}+\frac{\nu}{\kappa})$ to $(\partial_t+\sum\zeta_i\partial_{x_i})$. To facilitate the following discussion, we invoke the notion of solution operator.
\begin{definition}
$\SLB{t}$, $\SDF{t}$, $\SFr{t}$ are linear operators defined as the following:
\begin{multline}\label{eq:LB:Soln:Op:Eq}
    \Big( \SLB{t}\BK{f_{in}} \Big) (\bfx,\bfzeta) = f(\bfx,\bfzeta,t),
	\\
    \begin{dcases}
        \frac{\partial f}{\partial t} + \sum_{i=1}^d \zeta_i \frac{\partial f}{\partial x_i} 
        - \frac{1}{\kappa}Lf = 0, \
        f(\bfx,\bfzeta,0) = f_{in}(\bfx,\bfzeta)
        \\
        \text{ Maxwell-type boundary condition } \eqref{eq::Diff:Ref:BC:expand}
    \end{dcases}
\end{multline}
\begin{multline}\label{eq:DF:Soln:Op:Eq}
    \Big( \SDF{t}\BK{f_{in}} \Big) (\bfx,\bfzeta) = f(\bfx,\bfzeta,t),
	\\
    \begin{dcases}
        \frac{\partial f}{\partial t} + \sum_{i=1}^d \zeta_i \frac{\partial f}{\partial x_i} 
        + \frac{1}{\kappa}\nu f = 0, \
        f(\bfx,\bfzeta,0) = f_{in}(\bfx,\bfzeta)
        \\
        \text{ Maxwell-type boundary condition } \eqref{eq::Diff:Ref:BC:expand}
    \end{dcases}
\end{multline}
\begin{equation*}
	\Big( \SFr{t}\BK{g_{in}} \Big) (\bfx,\bfzeta) = g(\bfx,\bfzeta,t),
	\quad
    \begin{dcases}
        \frac{\partial g}{\partial t} + \sum_{i=1}^d \zeta_i \frac{\partial g}{\partial x_i} = 0, \
         g(\bfx,\bfzeta,0) = g_{in}(\bfx,\bfzeta)
        \\
        \text{ Maxwell-type boundary condition } \eqref{eq::Diff:Rel:BC}
    \end{dcases}
\end{equation*}
We call $\SLB{t}$, $\SDF{t}$, and $\SFr{t}$ the solution operators for the Linearized Boltzmann equation, free molecular flow with damping, and free molecular flow, respectively.
\end{definition}

From \eqref{eq::Choice:of:nu}, any $\mu>4$ is an admissible choice for free molecular flow. Since there is no need to vary $\mu$, for definiteness, from now on we fix $\mu=5$. The pointwise results for the free molecular flow, Theorem \ref{thm:Fr:Soln:Op:Ptws:Esti:new}, is written in the following form for $\epsilon=\frac{1}{400}$:

\begin{theorem}[Main Theorem of Free Molecular Flow: Solution Operator Form]
\label{thm:Fr:Soln:Op:Ptws:Esti}
For $f_{in}\in L^{\infty,-\gamma}_{\bfx,\bfzeta}$, $0 \leq \gamma \leq 1$,
\begin{align*}
    \frac{ \SFr{t} \BK{ \nu f_{in} \sqrt M }(\bfx,\bfzeta) }{ \sqrt M (1+|\bfzeta|)^\gamma }
=&
     O(1)\VertBK{ f_{in} }_{\infty,-\gamma}  \nu(\bfzeta),
     \\
    \frac{ \SFr{t} \BK{ f_{in} \sqrt M }(\bfx,\bfzeta) }{ \sqrt M (1+|\bfzeta|)^\gamma }
=&
     O(1)\VertBK{ f_{in} }_{\infty.-\gamma}.
\end{align*}
If, in addition, $\int f_{in} \sqrt M d\bfx d\bfzeta = 0$,
\begin{multline*}
    \frac{ \SFr{t} \BK{ f_{in} \sqrt M }(\bfx,\bfzeta) }{ \sqrt M (1+|\bfzeta|)^\gamma }\\
=
     O(1)
    \VertBK{ f_{in} }_{\infty,-\gamma}
    \Bigg\{
            \BK{\frac{ 1 }{ (1+\alpha t)^d }+(1-\alpha)^{\frac{t^{\frac{1}{400}}}{2}}}
            \bbbone_{ \curBK{ |\bfxi|>\frac{2}{t^{\frac{399}{400}}} } }
            +
            \bbbone_{ \curBK{ |\bfxi|<\frac{2}{t^{\frac{399}{400}}} } }
    \Bigg\}.
\end{multline*}
\end{theorem}

So far we have only obtained the existence, uniqueness, and pointwise esitmate of free molecular flow, Theorem \ref{thm:Fr:Soln:Op:Ptws:Esti}. Now we will settle down this issue for the damped free molecular flow, \eqref{eq:DF:Soln:Op:Eq}, and obtain a pointwise estimate of the solution.
\end{subsection}

\begin{subsection}{Global Existence and Boundedness.}\label{sec:sub:DF:Bddness}\hfil

In this subsection we establish the global existence and boundedness of boundary flux of the damped free molecular flow:
\begin{equation}\label{eq:DF:Eq}
\begin{dcases}
    \frac{\partial g^\nu}{\partial t}
    +
    \sum_{i=1}^d \zeta_i \frac{\partial g^\nu}{\partial x_i}
    +
    \frac{1}{\kappa}\nu g^\nu = 0, \\
     g^\nu (\bfx,\bfzeta,0) = g^\nu_{in} (\bfx,\bfzeta)
    \\
    \text{ Maxwell-type boundary condition } \eqref{eq::Diff:Rel:BC}.
\end{dcases}
\end{equation}
This can be done easily by the comparison with the free molecular flow.

By the characteristic method, solutions of \eqref{eq:DF:Eq} can be represented as
\begin{equation}\label{eq:FD:Chara:Rep}
     g^\nu (\bfx,\bfzeta,t)
=
\begin{dcases}
         &\alpha \sum\limits_{i=0}^{m-1}(1-\alpha)^{i} e^{ -\frac{\nu(\bfzeta)}{\kappa}(t_1+it_2)
          }j^\nu\BKK{ \bfx_{(i+1)}, t-t_1-it_2 }
          \BK{ \frac{2\pi}{RT(\bfx_{(i+1)})} }^\frac12 M_{T(\bfx_{(i+1)})}\\
         & +(1-\alpha)^{m}e^{ -\frac{\nu(\bfzeta)}{\kappa}t }g^\nu_{in}(\bfx_{(m)}-\bfxi^m(t-t_1-(m-1)t_2),\bfzeta^m)
           \quad  
        \text{ for } t_1 < t,
        \\
         &e^{ -\frac{\nu(\bfzeta)}{\kappa}t } g^\nu_{in}(\bfx-\bfxi t,\bfzeta)
        \quad
        \text{ for } t < t_1,
     \end{dcases}
\end{equation}
where
\begin{equation*}
m=\lfloor\frac{|\bfxi|t-|\bfx-\bfx_{(1)}|}{|\bfx_{(1)}-\bfx_{(2)}|}\rfloor +1.
\end{equation*}
Consequently,
\begin{multline*}
j^\nu(\bfy,t)=\intLim_{t<\frac{|\bfy-\bfy_{(1)}|}{|\bfxi_1|}}
\BK{-\bfxi_1\cdot\bfn(\bfy)}e^{-\frac{\nu(\bfzeta_1)}{\kappa}t}g^\nu_{in}(\bfy-\bfxi_1t,\bfzeta_1)d\bfzeta_1\\
+\intLim_{t>\frac{|\bfy-\bfy_{(1)}|}{|\bfxi_1|}}
\BK{-\bfxi_1\cdot\bfn(\bfy)}\alpha e^{-\frac{\nu(\bfzeta_1)}{\kappa}\frac{|\bfy-\bfy_{(1)}|}{|\bfxi_1|}}\BK{ \frac{2\pi}{RT(\bfy_{(1)})} }^{\frac{1}{2}}
M_{T(\bfy_{(1)})}(\bfzeta_1)j^\nu\BK{\bfy_{(1)},t-\frac{|\bfy-\bfy_{(1)}|}{|\bfxi_1|}}d\bfzeta_1\\
+\intLim_{t>\frac{|\bfy-\bfy_{(1)}|}{|\bfxi_1|}}
\BK{-\bfxi_1\cdot\bfn(\bfy)}(1-\alpha)e^{-\frac{\nu(\bfzeta_1)}{\kappa}\frac{|\bfy-\bfy_{(1)}|}{|\bfxi_1|}}g^\nu\BK{\bfy_{(1)},t-\frac{|\bfy-\bfy_{(1)}|}{|\bfxi_1|},\bfxi_1^1,\bfeta_1}d\bfzeta_1.\\
\end{multline*}
One can follow the discussion in Subsection \ref{sec:Fr:Stoch:Formulation} to derive the formula for $j^\nu(\bfy,t)$:
%One can iterate the equation of $j^\nu(\bfy,t)$ to obtain
\begin{equation*}
\begin{split}
j^\nu(\bfy,t)=&\sum\limits_{k=0}^n\BK{j^\nu_{(k)}(\bfy,t)+\sum\limits_{l=1}^k\sum\limits_{k_1+\ldots+k_l=l}^k
j^{\nu(l,k)}_{(k_1,\ldots,k_l)}(\bfy,t)}\\
+&J^\nu_{n+1}(\bfy,t)+\sum\limits_{l=1}^{n+1}\sum\limits_{k_1+\ldots+k_l=l}^{n+1}
J^{\nu(l,n+1)}_{(k_1,\ldots,k_l)}(\bfy,t),
\end{split}
\end{equation*}
where
\begin{multline*}
j^\nu_{(k)}(\bfy,t)=\intLim_{0<t-\sum\limits_{i=0}^{k-1}\frac{|\bfy_{(i)}-\bfy_{(i+1)}|}{|\bfxi_1^i|}<\frac{|\bfy_{(k)}-\bfy_{(k+1)}|}{|\bfxi_1^k|}}e^{-\sum\limits_{i=0}^{k-1}\frac{\nu(\bfzeta_1^i)}{\kappa}\frac{|\bfy_{(i)}-\bfy_{(i+1)}|}{|\bfxi_1^i|}}\\
\BK{-\bfxi_1\cdot\bfn(\bfy)}(1-\alpha)^k 
e^{-\frac{\nu(\bfzeta_1^k)}{\kappa}\BK{t-\sum\limits_{i=0}^{k-1}\frac{|\bfy_{(i)}-\bfy_{(i+1)}|}{|\bfxi_1^i|}}}
g^\nu_{in}\BK{\bfy_{(k)}-\bfxi_1^k(t-\sum\limits_{i=0}^{k-1}\frac{|\bfy_{(i)}-\bfy_{(i+1)}|}{|\bfxi_1^i|}),\bfzeta_1^k}d\bfzeta_1,
\end{multline*}
\begin{multline*}
j^{\nu(l,k)}_{(k_1,\ldots,k_l)}(\bfy,t)=
\intLim_{A^{(l,k)}_{(k_1,\ldots,k_l)}}
\prod\limits_{i=1}^l \BK{-\bfxi_i\cdot\bfn(\bfy_{(k_1,\ldots,k_{i-1})}}(1-\alpha)^{k_i-1}\alpha\\
e^{-\sum\limits_{j=1}^{k_i}\frac{\nu(\bfzeta_i^{j-1})}{\kappa}\frac{|\bfy_{(k_1,\ldots,k_{i-1},j-1)}-\bfy_{(k_1,\ldots,k_{i-1},j)}|}{|\bfxi_i^{j-1}|}}
\BK{ \frac{2\pi}{RT(\bfy_{(k_1,\ldots,k_i)})} }^{\frac{1}{2}}
M_{T(\bfy_{(k_1,\ldots,k_i)})}(\bfzeta_i^{k_i-1})\\
\BK{-\bfxi_{l+1}\cdot\bfn(\bfy_{(k_1,\ldots,k_l)})}
(1-\alpha)^{k-k_1-\ldots-k_l}e^{-\sum\limits_{i=1}^{k-k_1-\ldots-k_l}\frac{\nu(\bfzeta_{l+1}^{i-1})}{\kappa}\frac{|\bfy_{(k_1,\ldots,k_l,i-1)}-\bfy_{(k_1,\ldots,k_l,i)}|}{|\bfxi_{l+1}^{i-1}|})}\\
e^{-\frac{\nu(\bfzeta_{l+1}^{k-k_1-\ldots-k_l})}{\kappa}\BK{t-\sum\limits_{i=1}^{l}\sum\limits_{j=1}^{k_i}\frac{|\bfy_{(k_1,\ldots,k_{i-1},j-1)}-\bfy_{(k_1,\ldots,k_{i-1},j)}|}{|\bfxi_i^{j-1}|}-\sum\limits_{i=1}^{k-k_1-\ldots-k_l}\frac{|\bfy_{(k_1,\ldots,k_l,i-1)}-\bfy_{(k_1,\ldots,k_l,i)}|}{|\bfxi_{l+1}^{i-1}|}}}\\
g^\nu_{in}\left(\bfy_{(k_1,\ldots,k_l,k-k_1-\ldots-k_l)}-\bfxi_{l+1}^{k-k_1-\ldots-k_l}
(t-\sum\limits_{i=1}^{l}\sum\limits_{j=1}^{k_i}\frac{|\bfy_{(k_1,\ldots,k_{i-1},j-1)}-\bfy_{(k_1,\ldots,k_{i-1},j)}|}{|\bfxi_i^{j-1}|}\right.\\
\left.-\sum\limits_{i=1}^{k-k_1-\ldots-k_l}\frac{|\bfy_{(k_1,\ldots,k_l,i-1)}-\bfy_{(k_1,\ldots,k_l,i)}|}{|\bfxi_{l+1}^{i-1}|}),\bfzeta_{l+1}^{k-k_1-\ldots-k_l}\right)
d\bfzeta_{l+1}\ldots d\bfzeta_1,
\end{multline*}
\begin{multline*}
J^\nu_{n+1}(\bfy,t)=
\intLim_{t>\sum\limits_{i=0}^{n}\frac{|\bfy_{(i)}-\bfy_{(i+1)}|}{|\bfxi_1^i|}}
\BK{-\bfxi_1\cdot\bfn(\bfy)}e^{-\sum\limits_{i=0}^{n}\frac{\nu(\bfzeta_1^i)}{\kappa}\frac{|\bfy_{(i)}-\bfy_{(i+1)}|}{|\bfxi_1^i|}}(1-\alpha)^{n+1} \\
g^\nu\BK{\bfy_{(n+1)},t-\sum\limits_{i=0}^{n}\frac{|\bfy_{(i)}-\bfy_{(i+1)}|}{|\bfxi_1^i|},\bfzeta_1^{n+1}}d\bfzeta_1,
\end{multline*}

\begin{multline*}
J^{\nu(l,n+1)}_{(k_1,\ldots,k_l)}(\bfy,t)=
\intLim_{B^{(l,n+1)}_{(k_1,\ldots,k_l)}}
\prod\limits_{i=1}^l \BK{-\bfxi_i\cdot\bfn(\bfy_{(k_1,\ldots,k_{i-1})}}(1-\alpha)^{k_i-1}\alpha\\
e^{-\sum\limits_{j=1}^{k_i}\frac{\nu(\bfzeta_i^{j-1})}{\kappa}\frac{|\bfy_{(k_1,\ldots,k_{i-1},j-1)}-\bfy_{(k_1,\ldots,k_{i-1},j)}|}{|\bfxi_i^{j-1}|}}
\BK{ \frac{2\pi}{RT(\bfy_{(k_1,\ldots,k_i)})} }^{\frac{1}{2}}
M_{T(\bfy_{(k_1,\ldots,k_i)})}(\bfzeta_i^{k_i-1})\\
\BK{-\bfxi_{l+1}\cdot\bfn(\bfy_{(k_1,\ldots,k_l)})}
(1-\alpha)^{n+1-k_1-\ldots-k_l}e^{-\sum\limits_{i=1}^{n+1-k_1-\ldots-k_l}\frac{\nu(\bfzeta_{l+1}^{i-1})}{\kappa}\frac{|\bfy_{(k_1,\ldots,k_l,i-1)}-\bfy_{(k_1,\ldots,k_l,i)}|}{|\bfxi_{l+1}^{i-1}|}}\\
g^\nu\left(\bfy_{(k_1,\ldots,k_l)},
t-\sum\limits_{i=1}^{l}\sum\limits_{j=1}^{k_i}\frac{|\bfy_{(k_1,\ldots,k_{i-1},j-1)}-\bfy_{(k_1,\ldots,k_{i-1},j)}|}{|\bfxi_i^{j-1}|}\right.\\
\left.-\sum\limits_{i=1}^{n+1-k_1-\ldots-k_l}\frac{|\bfy_{(k_1,\ldots,k_l,i-1)}-\bfy_{(k_1,\ldots,k_l,i)}|}{|\bfxi_{l+1}^{i-1}|},\bfzeta_{l+1}^{n+1-k_1-\ldots-k_l}\right)
d\bfzeta_{l+1}\ldots d\bfzeta_1,
\end{multline*}
and
\begin{equation*}
\begin{split}
&\bfy_{(0)}\equiv \bfy, \quad \bfy_{(k_1,\ldots,k_l,0)}\equiv\bfy_{(k_1,\ldots,k_l)}, 
\quad \bfxi^{0}_l \equiv \bfxi_l, \\ 
&\bfy_{(k_1,\ldots,k_{l-1},i)} =\bfy_B \BKK{ \bfy_{(k_1,\ldots,k_{l-1},i-1)}, \frac{\bfxi^{i-1}_{l}}{|\bfxi^{i-1}_{l}|} },\\
&\bfxi^{i}_l = \bfxi^{i-1}_l -2(\bfxi^{i-1}_l\cdot\bfn(\bfy_{(k_1,\ldots,k_{l-1},i)}))\bfn(\bfy_{(k_1,\ldots,k_{l-1},i)}),\\
&\bfzeta^{i}_l=(\bfxi^{i}_l,\bfeta_l),
\end{split}
\end{equation*}
\begin{multline*}
A^{(l,k)}_{(k_1,\ldots,k_l)}=\\
\left\{0<t-\sum\limits_{i=1}^{l}\sum\limits_{j=1}^{k_i}\frac{|\bfy_{(k_1,\ldots,k_{i-1},j-1)}-\bfy_{(k_1,\ldots,k_{i-1},j)}|}{|\bfxi_i^{j-1}|}-\sum\limits_{i=1}^{k-k_1-\ldots-k_l}\frac{|\bfy_{(k_1,\ldots,k_l,i-1)}-\bfy_{(k_1,\ldots,k_l,i)}|}{|\bfxi_{l+1}^{i-1}|}\right.\\
\left.<\begin{dcases}
\frac{|\bfy_{(k_1,\ldots,k_l)}-\bfy_{(k_1,\ldots,k_l,1)}|}{|\bfxi_{l+1}|}\quad \text{ if } k-k_1-\ldots-k_l=0,\\
\frac{|\bfy_{(k_1,\ldots,k_l,k-k_1-\ldots-k_l)}-\bfy_{(k_1,\ldots,k_l,k_l,k-k_1-\ldots-k_l+1)}|}{|\bfxi_{l+1}^{k-k_1-\ldots-k_l}|}\quad \text{ if } k-k_1-\ldots-k_l>0
\end{dcases}\right\},
\end{multline*}
\begin{multline*}
B^{(l,n+1)}_{(k_1,\ldots,k_l)}=\\
\left\{t>\sum\limits_{i=1}^{l}\sum\limits_{j=1}^{k_i}\frac{|\bfy_{(k_1,\ldots,k_{i-1},j-1)}-\bfy_{(k_1,\ldots,k_{i-1},j)}|}{|\bfxi_i^{j-1}|}+\sum\limits_{i=1}^{n+1-k_1-\ldots-k_l}\frac{|\bfy_{(k_1,\ldots,k_l,i-1)}-\bfy_{(k_1,\ldots,k_l,i)}|}{|\bfxi_{l+1}^{i-1}|}\right\}.
\end{multline*}

Since the kernel of free molecular flow always dominates that of damped free molecular flow:
\begin{multline*}
j^\nu_{(k)}(\bfy,t)\leq %e^{-\frac{\nu_0}{\kappa}t}
\intLim_{0<t-\sum\limits_{i=0}^{k-1}\frac{|\bfy_{(i)}-\bfy_{(i+1)}|}{|\bfxi_1^i|}<\frac{|\bfy_{(k)}-\bfy_{(k+1)}|}{|\bfxi_1^k|}}
\BK{-\bfxi_1\cdot\bfn(\bfy)}(1-\alpha)^k \\
g^\nu_{in}\BK{\bfy_{(k)}-\bfxi_1^k(t-\sum\limits_{i=0}^{k-1}\frac{|\bfy_{(i)}-\bfy_{(i+1)}|}{|\bfxi_1^i|}),\bfzeta_1^k}d\bfzeta_1,
\end{multline*}
\begin{multline*}
j^{\nu(l,k)}_{(k_1,\ldots,k_l)}(\bfy,t)\leq %e^{-\frac{\nu_0}{\kappa}t}
\intLim_{A^{(l,k)}_{(k_1,\ldots,k_l)}}
\prod\limits_{i=1}^l \BK{-\bfxi_i\cdot\bfn(\bfy_{(k_1,\ldots,k_{i-1})}}(1-\alpha)^{k_i-1}\alpha\\
\BK{ \frac{2\pi}{RT(\bfy_{(k_1,\ldots,k_i)})} }^{\frac{1}{2}}
M_{T(\bfy_{(k_1,\ldots,k_i)})}(\bfzeta_i^{k_i-1})
\BK{-\bfxi_{l+1}\cdot\bfn(\bfy_{(k_1,\ldots,k_l)})}
(1-\alpha)^{k-k_1-\ldots-k_l}\\
g^\nu_{in}\left(\bfy_{(k_1,\ldots,k_l,k-k_1-\ldots-k_l)}-\bfxi_{l+1}^{k-k_1-\ldots-k_l}
(t-\sum\limits_{i=1}^{l}\sum\limits_{j=1}^{k_i}\frac{|\bfy_{(k_1,\ldots,k_{i-1},j-1)}-\bfy_{(k_1,\ldots,k_{i-1},j)}|}{|\bfxi_i^{j-1}|}\right.\\
\left.-\sum\limits_{i=1}^{k-k_1-\ldots-k_l}\frac{|\bfy_{(k_1,\ldots,k_l,i-1)}-\bfy_{(k_1,\ldots,k_l,i)}|}{|\bfxi_{l+1}^{i-1}|}),\bfzeta_{l+1}^{k-k_1-\ldots-k_l}\right)
d\bfzeta_{l+1}\ldots d\bfzeta_1.
\end{multline*}
\begin{equation*}
\begin{split}
j^\nu(\bfy,t)=&\sum\limits_{k=0}^\infty\BK{j^\nu_{(k)}(\bfy,t)+\sum\limits_{l=1}^k\sum\limits_{k_1+\ldots+k_l=l}^k
j^{\nu(l,k)}_{(k_1,\ldots,k_l)}(\bfy,t)}\\
=&O(1) \VertBK{g_{in}^{\nu}}_{\infty,5}.
\end{split}
\end{equation*}
This proves global existence of the solution and uniform boundedness of $j^\nu$ by comparison method, for the case  $g^\nu_{in} \geq 0$. For general initial configuration, let $g^{\nu,\pm}$ be the solution of damped free molecular flow, \eqref{eq:DF:Eq}, with initial configuration $(g^\nu_{in})^{\pm}$, the positive/negative part of $g_{in}^\nu$. Let $j^{\nu,\pm}$ be the flux of $g^{\nu,\pm}$. Note that \eqref{eq:DF:Eq} is linear, $j^\nu = j^{\nu,+} - j^{\nu,-}$. (However, $(j^\nu)^\pm \neq j^{\nu,\pm}$ in general.)
\begin{align*}
&
    |j^\nu(\bfy,t)| 
    \\
\leq& 
	 j^{\nu,+}(\bfy,t) + j^{\nu,-}(\bfy,t)
	 \\
=&
    O(1)  \BK{ \VertBK{(g_{in}^{\nu})^+}_{\infty,5} + \VertBK{(g_{in}^{\nu})^-}_{\infty,5} }
=
    O(1) \VertBK{g_{in}^{\nu}}_{\infty,5}.
\end{align*}

To sum up, we have
\begin{theorem}\label{thm:DF:Bddness:Soln}
For all $g^\nu_{in} \in L^{\infty,5}_{\bfx,\bfzeta}$, the solution of \eqref{eq:DF:Eq} exists globally, with
\begin{equation*}
    \BK{ \VertBK{j^\nu}_{L^\infty_\bfy} }(t) = O(1)  \VertBK{g_{in}^{\nu}}_{\infty,5}.
\end{equation*}
Therefore, for $f_{in}\in L^{\infty,-\gamma}_{\bfx,\bfzeta}$,
\begin{equation*}
    \VertBK{ \SDF{t}(f_{in}) }_{ \infty, -\gamma }
=
     O(1) \VertBK{f_{in}}_{\infty, -\gamma }.
\end{equation*}
\end{theorem}
\end{subsection}

\begin{subsection}{A Pointwise Estimate.}\hfil

For convenient, from now on we will frequently abbreviate functions $f(\bfx,\bfzeta,t)$, $\psi(\bfx,\bfzeta,t)$, etc., as $f(t)$, $\psi(t)$, etc..

\begin{lemma}\label{lem:DF:Local:Esti}
Suppose that $f_{in}\in L^{\infty,-\gamma}_{\bfx,\bfzeta}$, for some constant $\gamma$, $0\leq \gamma \leq 1$. Then, under the zero total initial molecular number assumption $\int f_{in}\sqrt M d\bfx d\bfxi = 0$,
\begin{multline*}
    \frac{ \SDF{t}(f_{in}) }{ (1+|\bfzeta|)^\gamma }
=
     O(1) \VertBK{ f_{in} }_{\infty,-\gamma}
    \Bigg\{
            \BK{\frac{ 1 }{ (1+\alpha t)^d }+(1-\alpha)^{\frac{t^{\frac{1}{400}}}{2}}}
            \bbbone_{ \curBK{ |\bfxi|>\frac{2}{t^{\frac{399}{400}}} } }
            \\+
            \bbbone_{ \curBK{ |\bfxi|<\frac{2}{t^{\frac{399}{400}}} } }
            +
            \frac{t}{\kappa} \nu(\bfzeta)
        \Bigg\}.
\end{multline*}
\end{lemma}
\begin{proof}\hfil

By Duhamel principle, \eqref{eq:DF:Soln:Op:Eq} is equivalent to

\begin{equation*}
	f(t) = \SDF{t}(f_{in})
=
	\frac{ \SFr{t} (f_{in} \sqrt M) }{ \sqrt M }
    -
    \frac{1}{\kappa} \int_0^t 
    \frac{ \SFr{t-s} (\nu f(s)\sqrt M) }{ \sqrt M } ds.
\end{equation*}
From Theorem \ref{thm:Fr:Soln:Op:Ptws:Esti}
\begin{multline*}
	\frac{ \SFr{t}\BK{ f_{in} \sqrt M } }{ \sqrt M (1+|\bfzeta|)^\gamma }
=
     O(1) \VertBK{ f_{in} }_{\infty,-\gamma}
    \Bigg\{
            \BK{\frac{ 1 }{ (1+\alpha t)^d }+(1-\alpha)^{\frac{t^{\frac{1}{400}}}{2}}}
            \bbbone_{ \curBK{ |\bfxi|>\frac{2}{t^{\frac{399}{400}}} } }
            \\+
            \bbbone_{ \curBK{ |\bfxi|<\frac{2}{t^{\frac{399}{400}}} } }
    \Bigg\}.
\end{multline*}
From Theorem \ref{thm:Fr:Soln:Op:Ptws:Esti} and Theorem
\ref{thm:DF:Bddness:Soln},
\begin{equation*}
	\frac{\SFr{t-s} (\nu f(s)\sqrt M)}{\sqrt M (1+|\bfzeta|)^\gamma} 
=	
	 O(1) \Vert f(s) \Vert_{\infty,-\gamma} \nu(\bfzeta) 
=
	 O(1)\Vert f_{in}\Vert_{\infty,-\gamma} \nu(\bfzeta).
\end{equation*}
Hence
\begin{equation*}
	\frac{1}{\kappa}\int_0^t 
	\frac{\SFr{t-s} (\nu f(s)\sqrt M)}{\sqrt M (1+|\bfzeta|)^\gamma} 
=	
	 O(1) \frac{t}{\kappa} \Vert f_{in}\Vert_{\infty,-\gamma} \nu(\bfzeta).
\end{equation*}
\end{proof}

Our next step is to remove the undesirable factor of $\nu(\bfzeta)$ in Lemma \ref{lem:DF:Local:Esti}. To do this, we conduct a posteriori estimate through the characteristic method.

\begin{theorem}\label{thm:DF:Local:Esti}
Suppose that $f_{in}\in L^{\infty,-\gamma}_{\bfx,\bfzeta}$, $0\leq\gamma\leq 1$. Then
under the zero total initial molecular number assumption $\int f_{in}\sqrt M d\bfx d\bfxi = 0$,
\begin{multline*}
    \frac{ \SDF{t}(f_{in}) }{ (1+|\bfzeta|)^\gamma }
=
     O(1)  \VertBK{f_{in}}_{\infty,-\gamma}
    \Bigg\{
            \Bigg(\frac{ 1 }{ (1+\alpha t)^d }+(1-\alpha)^{\frac{t^{\frac{1}{400}}}{2}}
            +\BK{\frac{1}{(1+t)^{\frac{399}{400}}}}^{d+1}\Bigg)\\
            \times \bbbone_{ \curBK{ |\bfxi|>\frac{2}{t^{\frac{399}{400}}} } }
            +
            \bbbone_{ \curBK{ |\bfxi|<\frac{2}{t^{\frac{399}{400}}} } }+\frac{t}{\kappa}
        \Bigg\}.
\end{multline*}
\end{theorem}

\begin{proof}\hfil

By the characteristic method,
\begin{align*}
& 
   \SDF{t} \Big( f_{in} \Big) (\bfx,\bfzeta) =
   f(\bfx,\bfzeta,t)
    \\
=&
    \begin{dcases}
         &\alpha \sum\limits_{i=0}^{m-1}(1-\alpha)^{i} e^{ -\frac{\nu(\bfzeta)}{\kappa}(t_1+it_2)
          }j\BKK{ \bfx_{(i+1)}, t-t_1-it_2 }
          \BK{ \frac{2\pi}{RT(\bfx_{(i+1)})} }^\frac12 M_{T(\bfx_{(i+1)})}\\
         & +(1-\alpha)^{m}e^{ -\frac{\nu(\bfzeta)}{\kappa}t }f_{in}(\bfx_{(m)}-\bfxi^m(t-t_1-(m-1)t_2),\bfzeta^m)
           \quad  
        \text{ for } t_1 < t,
        \\
         &e^{ -\frac{\nu(\bfzeta)}{\kappa}t } f_{in}(\bfx-\bfxi t,\bfzeta)
        \quad
        \text{ for } t < t_1.
     \end{dcases}
\end{align*}
From Lemma \ref{lem:DF:Local:Esti},
\begin{equation*}
     j(\bfy,s)
=
     O(1) \VertBK{f_{in}}_{\infty,-\gamma}
    \BK{ \frac{ 1 }{ (1+\alpha s)^d }+(1-\alpha)^{\frac{s^{\frac{1}{400}}}{2}}
            +\BK{\frac{1}{(1+s)^{\frac{399}{400}}}}^{d+1} +\frac{s}{\kappa} }.
\end{equation*}
Hence Theorem \ref{thm:DF:Local:Esti} follows.
\end{proof}
\end{subsection}
\end{section}

\begin{section}{Steady State Solution of The Boltzmann Equation and Its Time Asymptotic Stability}\label{sec:Full:Boltz}

\begin{subsection}{Linearized Boltzmann Equation.}\label{sec:sub:LB}\hfil

In this subsection, we study the linearized Boltzmann equation \eqref{eq:LB:Eq}. We first establish the local (in time) existence and a local estimate. By local in time we mean $\frac{t}{\kappa} \ll 1$, {\it not} $t \ll 1$. From now on we always consider perturbations of the form $f \sqrt M$ ($f_{in} \in L^{\infty,-\gamma}_{\bfx,\bfzeta}$).

\begin{theorem}[Local Existence and Estimate]\label{thm:LB:Local:Esti}
Let $f_{in}\in L^{\infty,-\gamma}_{\bfx,\bfzeta}$, $0 \leq \gamma \leq 1$, then there exists a constant $c_*>0$, such that whenever $\frac{t}{\kappa}<c_*$, the solution of \eqref{eq:LB:Eq} exists and satisfies
\begin{multline*}
    \frac{ \SLB{t}(f_{in}) }{ (1+|\bfzeta|)^\gamma }
=
   O(1)  \VertBK{f_{in}}_{\infty,-\gamma}
    \Bigg\{
            \Bigg(\frac{ 1 }{ (1+\alpha t)^d }+(1-\alpha)^{\frac{t^{\frac{1}{400}}}{2}}
            +\BK{\frac{1}{(1+t)^{\frac{399}{400}}}}^{d+1}\Bigg)\\
            \times \bbbone_{ \curBK{ |\bfxi|>\frac{2}{t^{\frac{399}{400}}} } }
            +
            \bbbone_{ \curBK{ |\bfxi|<\frac{2}{t^{\frac{399}{400}}} } }+\frac{t}{\kappa}
        \Bigg\}.
\end{multline*}
\end{theorem}

\begin{proof}\hfil

Write \eqref{eq:LB:Eq} as
\begin{equation}\label{eq:LB:Iterate:Eq}
    f(t) = \SLB{t}(f_{in})
=
    \SDF{t} (f_{in})
    +
    \frac{1}{\kappa} \int_0^t \SDF{t-s} \BK{ Kf(s) } ds.
\end{equation}
We solve \eqref{eq:LB:Iterate:Eq} by iteration:
\begin{equation*}
        f^{(0)}(t)
\equiv
        \SDF{t}(f_{in}),
        \\
\quad
        f^{(i)}(t)
\equiv \SDF{t}(f_{in})+
    \frac{1}{\kappa} \int_0^t \SDF{t-s} \BK{ Kf^{(i-1)}(s) } ds,
    \
     i \geq 1.
\end{equation*}
From Theorem \ref{thm:DF:Local:Esti} we have,
\begin{multline}\label{eq:z20}
    \frac{ f^{(0)}(t) }{ (1+|\bfzeta|)^\gamma }
=
    O(1)  \VertBK{f_{in}}_{\infty,-\gamma}
     \Bigg\{
            \Bigg(\frac{ 1 }{ (1+\alpha t)^d }+(1-\alpha)^{\frac{t^{\frac{1}{400}}}{2}}
            +\BK{\frac{1}{(1+t)^{\frac{399}{400}}}}^{d+1}\Bigg)\\
            \times \bbbone_{ \curBK{ |\bfxi|>\frac{2}{t^{\frac{399}{400}}} } }
            +
            \bbbone_{ \curBK{ |\bfxi|<\frac{2}{t^{\frac{399}{400}}} } }+\frac{t}{\kappa}
        \Bigg\}.
\end{multline}

From \eqref{eq:z20},
\begin{equation*}
    \VertBK{ f^{(0)} (t) }_{\infty,-\gamma} = O(1) \VertBK{ f_{in} }_{\infty,-\gamma} \BK{1+\frac{t}{\kappa}}.
\end{equation*}
Consequently, from \eqref{eq:LB:K:Norm:Esti} and Theorem \ref{thm:DF:Bddness:Soln},
\begin{align*}
&
        \VertBK{ f^{(1)} (t) }_{\infty,-\gamma}
\leq
        \frac{1}{\kappa} \int_0^t \VertBK{ \SDF{t-s} \BK{ Kf^{(0)}(s) } }_{\infty,-\gamma} ds
        \\
=&
        \frac{O(1)}{\kappa} \int_0^t \VertBK{ f^{(0)}(s) }_{\infty,-\gamma} ds
=
        \VertBK{ f_{in} }_{\infty,-\gamma}  \BK{1+\frac{t}{ \kappa}}
        \BK{ O(1) \frac{t}{\kappa} }.
\end{align*}
Similarly,  we  obtain, by  induction,
\begin{align*}
        \VertBK{ f^{(i)} (t) }_{\infty,-\gamma}
=
        \VertBK{ f_{in} }_{\infty,-\gamma}  \BK{1+\frac{t}{ \kappa}}
        \BK{ O(1) \frac{t}{\kappa} }^i,  i\geq 1,.
\end{align*}
Hence, for  $\frac{t}{\kappa}$ sufficiently  small, 
\begin{equation*}
	\sum_{i=1}^\infty \VertBK{ f^{(i)}(t) }_{\infty,-\gamma} = O(1) \frac{t}{\kappa} < \infty,
\end{equation*}
and Theorem \ref{thm:LB:Local:Esti} follows.
\end{proof}

With the local estimate of Theorem \ref{thm:LB:Local:Esti}, we are ready to prove the global (in time) exponential decay of \eqref{eq:LB:Eq}, Theorem \ref{thm:LB:Main}. Recall that $\nu_0=\inf \nu(\bfzeta)$.

\begin{proof}{\textit{of Theorem \ref{thm:LB:Main}}}\hfil

For simplicity we write $\SLB{t}(f_{in})(\bfx,\bfzeta)$ as $f(t)$ in this proof. Define
\begin{equation*}
\calF(t) \equiv \sup_{0\leq s \leq t} e^{\frac{\nu'_1}{\kappa}s}\VertBK{ f(s) }_{\infty,-\gamma}.
\end{equation*}
Let $c<c_*/2$ be a small constant to be specified later. Recall that $c_*$ is a constant given in Theorem \ref{thm:LB:Local:Esti}. Since $2c<c_*$, we may apply Theorem \ref{thm:LB:Local:Esti} to obtain
\begin{multline*}
        \calF(2 c\kappa)
=
        \sup \curBK{
                        \frac{ |f(\bfx,\bfzeta,t)| e^{\nu'_1\frac{t}{k}} }{ (1+|\bfzeta|)^a }
                     :  \bfx \in D, \ \bfzeta\in\bbR^3, \ 0\leq t\leq 2 c\kappa
                    }
        \\
\leq
         e^{2\nu'_1 c} \sup_{0\leq t\leq 2 c\kappa}\VertBK{f(t)}_{\infty,-\gamma}
\leq
         e^{\nu_0c_*} \sup_{0\leq t\leq  c_*\kappa} \VertBK{f(t)}_{\infty,-\gamma}
=
         O(1) \VertBK{f_{in}}_{\infty,-\gamma}.
\end{multline*}
By the definition of $\calF$, it is clear that $\calF$ increases with $t$. We now claim that under some appropriate choice of $c$ and  $\kappa$
\begin{equation}\label{eq:LB:Main:Claim}
	\calF(t) \leq \calF(t- c\kappa), \text{ whenever } t > 2 c\kappa.
\end{equation}
This claim implies $\calF(t)=O(1)\VertBK{f_{in}}_{\infty,-\gamma}$, which proves this theorem.

Tracing back from time $t$ to the earlier time $t- c\kappa$ by the characteristic method, we can represent $f(t)$ as:
\begin{equation}\label{eq:LB:Chara:Rep}
\begin{split}
    	f (\bfx,\bfzeta,t)
=&
         e^{-\nu(\zeta) c} f(\bfx- c\kappa\bfxi,\bfzeta,t- c\kappa)
    	\\
&+
        \frac{1}{\kappa} \int_0^{ c\kappa} e^{ -\frac{\nu(\bfzeta)}{\kappa}s }
         K(f(s))(\bfx-s\bfxi,\bfzeta,t-s) ds,
        \text{ when }    t_1 >  c\kappa,
        \\
    	f (\bfx,\bfzeta,t)
=&
         \sum\limits_{i=0}^{m-1} \left\{(1-\alpha)^{i} e^{ -\frac{\nu(\bfzeta)}{\kappa}(t_1+it_2) }
         \alpha j(\bfx_{(i+1)},t-t_1-it_2) 
        \BK{\frac{2\pi}{RT(\bfx_{(i+1)})}}^\frac12 M_{T(\bfx_{(i+1)})}(\bfzeta)\right.
 		\\       
&\left.+
         (1-\alpha)^{i}
         \frac{1}{\kappa}\int_{t_1+\ldots+t_i}^{t_1+\ldots+t_{i+1}} e^{ -\frac{\nu(\bfzeta)}{\kappa}s }
         K(f(s))(\bfx_{(i)}-s\bfxi^i,\bfzeta^i,t-s) ds\right\}\\
         +&(1-\alpha)^m e^{-\nu(\bfzeta) c}
         f(\bfx_{(m)}-\bfxi^m(t-t_1-(m-1)t_2),\bfzeta^m,t- c\kappa),
        \text{ when }    t_1 <  c\kappa,
\end{split}
\end{equation}
where
\begin{equation*}
\begin{split}
t_1=&\frac{|\bfx-\bfx_{(1)}|}{|\bfxi|},\\
t_2=&\frac{|\bfx_{(1)}-\bfx_{(2)}|}{|\bfxi|},\\
m=&\lfloor\frac{|\bfxi| c\kappa-|\bfx-\bfx_{(1)}|}{|\bfx_{(1)}-\bfx_{(2)}|}\rfloor +1.
\end{split}
\end{equation*}

Consider first $j(\bfx_{(i+1)},t-t_1-it_2)$ with $t_1+it_2< c\kappa$, for $i=0,\ldots,m-1$. We trace back an extra $2 c\kappa-t_1-it_2$ amount of time to arrive at $t-2 c\kappa$. Since $t_1+it_2< c\kappa$ for $i=0,\ldots,m-1$, by Theorem \ref{thm:LB:Local:Esti} we have
\begin{equation}\label{eq:z15}
\begin{split}
 &    j(\bfx_{(i+1)},t-t_1-it_2)
=
     O(1)
    \int_{\bfxi_*\cdot\bfn<0} -\bfxi_*\cdot\bfn
    \VertBK{f(t-2 c\kappa)}_{\infty,-\gamma}
    (1+|\bfzeta_*|)^\gamma
    \\
&
    \times
    \Bigg\{
           \BK{ \frac{1}{\alpha^{d}(2 c\kappa-t_1-it_2)^d} + (1-\alpha)^{\frac{(2 c\kappa-t_1-it_2)^{\frac{1}{400}}}{2}} +\BK{\frac{1}{(2 c\kappa-t_1-it_2)^{\frac{399}{400}}}}^{d+1}}
           \\
           &\times \bbbone_{ \curBK{|\bfxi_*|>\frac{2}{(2 c\kappa-t_1-it_2)^{\frac{399}{400}}}} }
            +
            \bbbone_{ \curBK{|\bfxi_*|<\frac{2}{(2 c\kappa-t_1-it_2)^{\frac{399}{400}}}} }
            +
             c
        \Bigg\}
    \sqrt M d\bfzeta_*
    \\
=&
     O(1) \VertBK{f(t-2 c\kappa)}_{\infty,-\gamma}
     \BK{
            \frac{1}{(\alpha ck)^d} + (1-\alpha)^{\frac{(c\kappa)^{\frac{1}{400}}}{2}} +\BK{\frac{1}{ (c\kappa)^{\frac{399}{400}}}}^{d+1}+ c
        }.
\end{split}
\end{equation}
Therefore,
\begin{multline*}
\sum\limits_{i=0}^{m-1} (1-\alpha)^{i} e^{ -\frac{\nu(\bfzeta)}{\kappa}(t_1+it_2) }
         \alpha j(\bfx_{(i+1)},t-t_1-it_2) 
        \BK{\frac{2\pi}{RT(\bfx_{(i+1)})}}^\frac12 M_{T(\bfx_{(i+1)})}(\bfzeta)\\
        = O(1) \VertBK{f(t-2 c\kappa)}_{\infty,-\gamma}
     \BK{
            \frac{1}{(\alpha ck)^d} + (1-\alpha)^{\frac{(c\kappa)^{\frac{1}{400}}}{2}} +\BK{\frac{1}{ (c\kappa)^{\frac{399}{400}}}}^{d+1}+ c
        }\BK{\frac{2\pi}{R\Tm}}^\frac12 M(\bfzeta).
\end{multline*}
For Theorem \ref{thm:LB:Local:Esti} to apply, we need $\int f(t-2 c\kappa)\sqrt M d\bfx d\bfzeta=0$. This is true because of \eqref{eq:LB:Collision:Invar:L}.

Consider next $\int_{t_1+\ldots+t_i}^{t_1+\ldots+t_{i+1}} e^{-\frac{\nu}{k}s} K(f(t-s))ds$. We also trace back to the time $t-2 c\kappa$. From Theorem \ref{thm:LB:Local:Esti}, \eqref{eq:LB:K:Norm:Esti}, and \eqref{eq:LB:K:low:speed:Esti},
\begin{equation}\label{eq:z16}
\begin{split}
&
    \sum\limits_{i=0}^{m-1} (1-\alpha)^{i}\frac{1}{\kappa}\int_{t_1+\ldots+t_i}^{t_1+\ldots+t_{i+1}} e^{ -\frac{\nu(\bfzeta)}{\kappa}s }
         K(f(s))(\bfx_{(i)}-s\bfxi^i,\bfzeta^i,t-s) ds
    \\
=&
      O(1) \VertBK{f(t-2 c\kappa)}_{\infty,-\gamma} (1+|\bfzeta|)^\gamma
     \\
&
     \times
     \sum\limits_{i=0}^{m-1} \frac{1}{\kappa}\int_{t_1+\ldots+t_i}^{t_1+\ldots+t_{i+1}}
              \Bigg\{\frac{1}{\alpha^d (2 c\kappa - s)^d}+(1-\alpha)^{\frac{(2c\kappa-s)^{\frac{1}{400}}}{2}} +\BK{\frac{1}{ (2c\kappa-s)^{\frac{399}{400}}}}^{d+1}\\
              &+
            \curBK{ \begin{array}{c@{,}l}
                        \displaystyle
                       \frac{1}{ (2c\kappa-s)^{\frac{399}{400}}}                       &   \text{ when } d=1        \\
                        \displaystyle
                        \BK{\frac{1}{ (2c\kappa-s)^{\frac{399}{400}}}}^2 \log(2 c\kappa-s)        &   \text{ when } d=2        
                    \end{array} }
            +
            c
        \Bigg\} ds
     \\
=&
     O(1) \VertBK{f(t-2 c\kappa)}_{\infty,-\gamma} (1+|\bfzeta|)^\gamma
     \BK{
            \frac{1}{\kappa}
             P( c\kappa)
            +
              c^2
        },
\end{split}
\end{equation}
where
\begin{equation*}
	P(z)
=
   \begin{dcases}
   \displaystyle
    \frac{1}{\alpha} + z(1-\alpha)^{\frac{z^{\frac{1}{400}}}{2}} + z^{\frac{1}{400}}                         &   \text{ for } d=1,    \\
   \displaystyle
   \frac{1}{\alpha^2 z} + z(1-\alpha)^{\frac{z^{\frac{1}{400}}}{2}} + \frac{\log z}{z^{\frac{199}{200}}}  &   \text{ for } d=2.
   \end{dcases}
\end{equation*}

Plugging \eqref{eq:z15} and \eqref{eq:z16} back to \eqref{eq:LB:Chara:Rep}, we have
\begin{equation*}
\begin{split}
&
    \frac{ e^{\frac{\nu'_1}{\kappa}t} |f(t)| }{ (1+|\bfzeta|)^\gamma }
    \\
\leq&
    \begin{dcases}
         e^{ -(\nu_0-\nu'_1) c } \calF(t- c\kappa)
        +
         C' e^{ 2 c\nu'_1 } \calF(t-2 c\kappa)
        \BK{ \frac{1}{\kappa} P( c\kappa) +  c^2 }
         &
        \text{ for }    t_1 >  c\kappa,
        \\
         C' e^{ 2 c\nu'_1 } \calF(t-2 c\kappa)
        \Bigg\{
                \frac{1}{(\alpha ck)^d} + (1-\alpha)^{\frac{(c\kappa)^{\frac{1}{400}}}{2}} +\BK{\frac{1}{ (c\kappa)^{\frac{399}{400}}}}^{d+1}+ c\\
               + \BK{ \frac{1}{\kappa} P( c\kappa) +  c^2 }\Bigg\}
            + (1-\alpha)^m e^{ -(\nu_0-\nu'_1) c } \calF(t-c\kappa)
         &
        \text{ for }    t_1 <  c\kappa,
    \end{dcases}
    \\
\leq&
    \calF(t- c\kappa)
\\
&
	\times
    \begin{dcases}
         e^{ -(\nu_0-\nu'_1) c } 
        +
         C' e^{ 2 c\nu'_1 }
  		\BK{ \frac{1}{\kappa} P( c\kappa) + c^2 }
         &
        \text{ for }    t_1 >  c\kappa,
        \\
   1-\alpha
        +
        C' e^{ 2 c\nu'_1 }
        \Bigg\{
                \frac{1}{(\alpha ck)^d} + (1-\alpha)^{\frac{(c\kappa)^{\frac{1}{400}}}{2}} +\BK{\frac{1}{ (c\kappa)^{\frac{399}{400}}}}^{d+1}+ c
                \\+
                \BK{ \frac{1}{\kappa} P( c\kappa) +  c^2 }                
            \Bigg\}
         &
        \text{ for }    t_1 <  c\kappa,
    \end{dcases}
\end{split}
\end{equation*}
for some positive constant $C'$.

For \eqref{eq:LB:Main:Claim} to hold, we need
\begin{equation}\label{eq:z17}
\left.	\begin{array}{c}
         \displaystyle
         e^{ -(\nu_0-\nu'_1) c } 
        +
         C' e^{ 2 c\nu'_1 }
  		\BK{ \frac{1}{\kappa} P( c\kappa) + c^2 }
        \\
        \displaystyle
       1-\alpha
        +
        C' e^{ 2 c\nu'_1 }
        \sqBK{
                \frac{1}{(\alpha ck)^d} + (1-\alpha)^{\frac{(c\kappa)^{\frac{1}{400}}}{2}} +\BK{\frac{1}{ (c\kappa)^{\frac{399}{400}}}}^{d+1}+ c
                +
                \BK{ \frac{1}{\kappa} P( c\kappa) +  c^2 }                
            }
        \end{array} \right\}
\leq
         1.
\end{equation}
For this purpose, large $\kappa$ and small $c$ are desirable. We will {\it fix} $c$ according to $\nu'_1$ and $\alpha$, and find the admissible choice of $\kappa$ after $c$ has been specified. We fix some small $c$, so small that
\begin{equation}\label{eq:LB:requirement:of:c}
\begin{split}
	 c &< \alpha,
\\
     e^{-(\nu_0-\nu'_1) c} &< 1- \frac12(\nu_0-\nu'_1) c,
\\
     C' e^{ 2c\nu'_1 } c &< \frac{1}{4} \min \curBK{ \nu_0-\nu'_1, \alpha}.
\end{split}
\end{equation}
It is not difficult to see that we can choose some $c\sim \alpha(\nu_0-\nu'_1)$ to meet all these requirements. From \eqref{eq:LB:requirement:of:c},
\begin{multline*}
	 e^{ -(\nu_0-\nu'_1) c } 
        +
         C' e^{ 2 c\nu'_1 }
  		\BK{ \frac{1}{\kappa} P( c\kappa) + c^2 }
	\leq
	1 - \frac14 (\nu_0-\nu_1') c +  C' e^{ 2 c\nu'_1 }
	\frac{1}{\kappa} P( c\kappa),
\end{multline*}
\begin{multline*}
1-\alpha
        +
        C' e^{ 2 c\nu'_1 }
        \sqBK{
                \frac{1}{(\alpha ck)^d} + (1-\alpha)^{\frac{(c\kappa)^{\frac{1}{400}}}{2}} +\BK{\frac{1}{ (c\kappa)^{\frac{399}{400}}}}^{d+1}+ c
                +
                \BK{ \frac{1}{\kappa} P( c\kappa) +  c^2 }                
            }
\\
	\leq
	1 - \frac{\alpha}{2} +  C' e^{ 2 c\nu'_1 }
	\sqBK{
             \frac{1}{(\alpha ck)^d} + (1-\alpha)^{\frac{(c\kappa)^{\frac{1}{400}}}{2}} +\BK{\frac{1}{ (c\kappa)^{\frac{399}{400}}}}^{d+1}
                +
                 \frac{1}{\kappa} P( c\kappa)               
            }.
\end{multline*}

For this specific $c\sim \alpha(\nu_0-\nu_1')$, we need $\kappa$ to satisfy
\begin{equation}\label{eq:LB:requirement:of:k}
\begin{split}
   &C' e^{ 2 c\nu'_1 }
	\frac{1}{\kappa} P( c\kappa)
\leq
      c \frac{\nu_0-\nu'_1}{4} \sim \alpha(\nu_0-\nu'_1)^2,
    \\
   &C' e^{ 2 c\nu'_1 }
	\sqBK{
             \frac{1}{(\alpha ck)^d} + (1-\alpha)^{\frac{(c\kappa)^{\frac{1}{400}}}{2}} +\BK{\frac{1}{ (c\kappa)^{\frac{399}{400}}}}^{d+1}
                +
                 \frac{1}{\kappa} P( c\kappa)               
            }
\leq
     \frac{\alpha}{2} .
\end{split}
\end{equation}
It is not difficult to see that there exists a positive constant $C_1$, {\it independent of} $c$, such that \eqref{eq:LB:requirement:of:k:derived} implies \eqref{eq:LB:requirement:of:k}. Consequently, under the assumption \eqref{eq:LB:requirement:of:k:derived}, \eqref{eq:LB:Main:Claim} holds.
\end{proof}
\end{subsection}

\begin{subsection}{Exponential Convergence for Full Boltzmann Equation.}\label{sec:sub:FullBz}\hfil

Using the exponential decay of Linearized Boltzmann equation, Theorem \ref{thm:LB:Main}, we are able to  establish the existence of steady state solution and the exponential decay for full Boltzmann equation, \eqref{eq:FullBz:Eq} or equivalently \eqref{eq:FullBz:Eq:expand}. In terms of $\SLB{t}$, \eqref{eq:FullBz:Eq:expand} is equivalent to
\begin{multline}\label{eq:FullBz:Eq:expand:Soln:Op}
    f(t)
=
    \SLB{t}(f_{in})
    +
    \frac{1}{\kappa} \int_0^t \SLB{t-s}\BK{ L\BKK{ \frac{S-M}{\sqrt M} } } ds
    \\
    +
    \frac{1}{\kappa} \int_0^t \SLB{t-s}\BK{ \frac{ Q(S-M+\sqrt M f,S-M+\sqrt Mf) }{ \sqrt M } } ds.
\end{multline}
In view of \eqref{eq:FullBz:Eq:expand:Soln:Op}, the equation for the steady state solution $\Phi= \Phi(\bfx,\bfzeta)$ is
\begin{multline}\label{eq:FullBz:Steady:Eq:expand:Soln:Op}
    \Phi
=
    \frac{1}{\kappa} \int_0^\infty \SLB{s}\BK{ L\BKK{ \frac{S-M}{\sqrt M} } } ds
    \\
    +
    \frac{1}{\kappa} \int_0^\infty \SLB{s}\BK{ \frac{ Q(S-M+\sqrt M\Phi,S-M+\sqrt M\Phi) }{ \sqrt M } } ds.
\end{multline}
We first prove the existence of the steady state solution $\Phi$, Theorem \ref{thm:FullBz:Main:1}.

\begin{proof}{\textit{of the Theorem \ref{thm:FullBz:Main:1}}}\hfil

We solve \eqref{eq:FullBz:Steady:Eq:expand:Soln:Op} by Picard iteration:
\begin{align}
    \notag
    \Phi^{(0)}(\bfx,\bfzeta)
=&
    \frac{1}{\kappa} \int_0^\infty \SLB{s}\BK{ L\BKK{ \frac{S-M}{\sqrt M} } } ds
    +
    \frac{1}{\kappa} \int_0^\infty \SLB{s}\BK{ \frac{ Q(S-M,S-M) }{ \sqrt M } } ds,
    \\
    \label{eq:z18}
    \Phi^{(i)}(\bfx,\bfzeta)
=&
    \frac{1}{\kappa} \int_0^\infty \SLB{s}\BK{ L\BKK{ \frac{S-M}{\sqrt M} } } ds
    \\
&
	\notag
    +
    \frac{1}{\kappa} \int_0^\infty 
    \SLB{s}\BK{ \frac{ Q(S-M+\sqrt M\Phi^{(i-1)},S-M+\sqrt M\Phi^{(i-1)}) }{ \sqrt M } } ds.
\end{align}
From \eqref{eq:FullBz:Collision:Invar:Q} and \eqref{eq:LB:Collision:Invar:L}, $\int\Phi^{(i)}d\bfx d\bfzeta=0$ for all $i$. Hence \eqref{eq:FullBz:Phi:Zero:Total:Mass} follows, provided $\sum\VertBK{\Phi^{(i)}}_\infty$ converges. \\

{\sc Claim: } Under the assumptions $1-\Tm\ll 1$ and $\kappa\gg 1$, $\Vert(\Phi^{(i)}-\Phi^{(i-1)})\Vert_\infty=[O(1)(1-\Tm)]^{i+1}$, for $i\geq 0$. Here we set $\Phi^{(-1)}=0$.\\

Note that the claim directly concludes this theorem.

We prove the claim by induction.  First, we follow the proof of Theorem 13 in \cite{Kuo-Liu-Tsai-2} to conclude 
\begin{equation}\label{eq:FullBz:Steady:Thm:Eq:1}
    \VertBK{ \frac{ F^{(n)}(t) }{ \nu } }_\infty
=
    O(1)(1-\Tm)[ O(1)(1-\Tm) ]^{n+1}.
\end{equation}

The next step is to remove the undesirable factor $\nu$ sitting under $F^{(n)}(t)$ on the LHS of \eqref{eq:FullBz:Steady:Thm:Eq:1}. To do this, we conduct a posterior estimate on $F^{(n)}$. We note that $F^{(n)}$ can be defined equivalently by the the differential equation
\begin{equation*}
\begin{dcases}
    \BK{ 
    		\frac{\partial~}{\partial t} 
    		+ \sum_{i=1}^d \zeta_i \frac{\partial~}{\partial x_i} + \frac{\nu}{\kappa} 
    	} 
    F^{(n)}
=
    \frac{1}{\kappa} KF^{(n)}
    \\
\quad
    +
    \frac{1}{\kappa}   
    \frac{ 
    		Q\BK{ \sqrt M \Phi^{(n)} + \sqrt M \Phi^{(n-1)} + 2(S-M), 
    		\sqrt M \Phi^{(n)}-\sqrt M \Phi^{(n-1)} } }{ \sqrt M 
    	},
    \\
    \text{ Maxwell-type boundary condition } \eqref{eq::Diff:Ref:BC:expand},
    \\
     F^{(n)}(\bfx,\bfzeta,0) = 0.
\end{dcases}
\end{equation*}
Hence $F^{(n)}$ can be represented by the characteristic method as:
\begin{equation}\label{eq:FullBz:Steady:Thm:Eq:2}
\begin{split}
&
     F^{(n)}(t)
    \\
=&
	\bbbone_{\curBK{t_1 > t}} \times    
    \frac{1}{\kappa}
    \int_{0}^{t} e^{ -\frac{\nu}{\kappa} s }
     K F^{(n)}\BK{\bfx_{(i)}-\bfxi^i s,\bfzeta^i,t-s } ds
    \\
&
    +
    \frac{1}{\kappa}
    \int_{0}^{t}  e^{ -\frac{\nu}{\kappa} s }
    \\
&
    \times
    \frac{
             Q\BK{ \sqrt M \Phi^{(n)} + \sqrt M \Phi^{(n-1)} + 2(S-M), \sqrt M \Phi^{(n)}-\sqrt M \Phi^{(n-1)} }
        }{
            \sqrt M
        }\BK{\bfx_{(i)}-\bfxi^i s,\bfzeta^i,t-s}
     ds\\
+&
	\bbbone_{\curBK{t_1 < t}} \times
	\sum\limits_{i=0}^{m-1} (1-\alpha)^i \Bigg(
     e^{ -\frac{\nu}{\kappa}(t_1+it_2) }\alpha j(\bfx_{(i+1)},t-t_1-it_2)
    \BK{ \frac{2\pi}{RT(\bfx_{(i+1)})} }^\frac12
     M_{T(\bfx_{(i+1)})}
    \\
&    
    \frac{1}{\kappa}
    \int_{t_1+\ldots+t_i}^{t_1+\ldots+t_{i+1}} e^{ -\frac{\nu}{\kappa} s }
     K F^{(n)}\BK{\bfx_{(i)}-\bfxi^i s,\bfzeta^i,t-s } ds
    \\
&
    +
    \frac{1}{\kappa}
    \int_{t_1+\ldots+t_i}^{t_1+\ldots+t_{i+1}}  e^{ -\frac{\nu}{\kappa} s }
    \\
&
    \times
    \frac{
             Q\BK{ \sqrt M \Phi^{(n)} + \sqrt M \Phi^{(n-1)} + 2(S-M), \sqrt M \Phi^{(n)}-\sqrt M \Phi^{(n-1)} }
        }{
            \sqrt M
        }\BK{\bfx_{(i)}-\bfxi^i s,\bfzeta^i,t-s}
     ds
    \Bigg).
\end{split}
\end{equation}
From \eqref{eq:FullBz:Steady:Thm:Eq:1},
\begin{equation}\label{eq:FullBz:Steady:Thm:Eq:3}
\begin{split}
    &\sum\limits_{i=0}^{m-1} (1-\alpha)^i 
     e^{ -\frac{\nu}{\kappa}(t_1+it_2) }\alpha j(\bfx_{(i+1)},t-t_1-it_2)
    \BK{ \frac{2\pi}{RT(\bfx_{(i+1)})} }^\frac12
     M_{T(\bfx_{(i+1)})}\\
=&
	\alpha\sum\limits_{i=0}^{m-1} (1-\alpha)^i[O(1)(1-\Tm)]^{n+2}
	=[O(1)(1-\Tm)]^{n+2}
	;
	\end{split}
\end{equation}
and, from \eqref{eq:FullBz:Steady:Thm:Eq:1} and \eqref{eq:LB:K:Norm:Esti},
\begin{equation}\label{eq:FullBz:Steady:Thm:Eq:4}
\begin{split}
&
  \sum\limits_{i=0}^{m-1} (1-\alpha)^i \frac{1}{\kappa}
    \int_{t_1+\ldots+t_i}^{t_1+\ldots+t_{i+1}} e^{ -\frac{\nu}{\kappa} s }
     K F^{(n)}\BK{\bfx_{(i)}-\bfxi^i s,\bfzeta^i,t-s } ds
    \\
=&
	 [O(1)(1-\Tm)]^{n+2}
	\sum\limits_{i=0}^{m-1}  \frac{1}{\kappa}
    \int_{t_1+\ldots+t_i}^{t_1+\ldots+t_{i+1}} e^{ -\frac{\nu}{\kappa} s } ds
	\\
=&
	 [O(1)(1-\Tm)]^{n+2}
	 \frac{1}{\kappa}\int_0^t e^{ -\frac{\nu}{\kappa} s } ds
	 \\
=&
	 [O(1)(1-\Tm)]^{n+2}
\end{split}
\end{equation}
From \eqref{eq:FullBz:Linfty:Ineq:Q}, \eqref{eq:FullBz:S:minus:M:Estimate}, and the induction hypothesis,
\begin{equation}\label{eq:FullBz:Steady:Thm:Eq:5}
\begin{split}
&
	\Bigg | \sum\limits_{i=0}^{m-1} (1-\alpha)^i \frac{1}{\kappa}
    \int_{t_1+\ldots+t_i}^{t_1+\ldots+t_{i+1}} e^{ -\frac{\nu}{\kappa} s }
    \\
&
    \times
    \frac{
             Q\BK{ \sqrt M \Phi^{(n)} + \sqrt M \Phi^{(n-1)} + 2(S-M), \sqrt M \Phi^{(n)}-\sqrt M \Phi^{(n-1)} }
        }{
            \sqrt M
        }
     ds \Bigg |
	\\
\leq&
    \absBK{ \int_0^t  \frac{\nu}{\kappa} e^{ -\frac{\nu}{\kappa} s } ds }
    \\
&
    \times
    \VertBK{
    		\frac{
             		Q\BK{ 
             				\sqrt M \Phi^{(n)} + \sqrt M \Phi^{(n-1)} + 2(S-M), 
             				\sqrt M \Phi^{(n)}-\sqrt M \Phi^{(n-1)} 
             			}
        		}{
            		\nu \sqrt M
        		}
     		}_\infty
     \\
=&
	  O(1) 
	 \BK{ \VertBK{ \Phi^{(n)} + \Phi^{(n-1)} } + O(1)(1-\Tm) }
	 \VertBK{ \Phi^{(n)} - \Phi^{(n-1)} }_\infty
     \\
=&
	  [O(1)(1-\Tm)][O(1)(1-\Tm)]^{n+1}.
\end{split}
\end{equation}
Plugging \eqref{eq:FullBz:Steady:Thm:Eq:3}, \eqref{eq:FullBz:Steady:Thm:Eq:4}, and \eqref{eq:FullBz:Steady:Thm:Eq:5} back to \eqref{eq:FullBz:Steady:Thm:Eq:2}, we obtain
\begin{equation*}
	F^{(n)}(t) = [O(1)(1-\Tm)]^{n+2}.
\end{equation*}
Consequently, $\Vert\Phi^{(n+1)}-\Phi^{(n)}\Vert_\infty = [O(1)(1-\Tm)]^{n+2}$. This concludes the claim and therefore this theorem.
\end{proof}

Now we have already obtained the steady state solution $F_\infty \equiv S+\sqrt M\Phi$ for full Boltzmann equation \eqref{eq:FullBz:Eq}.
To establish the nonlinear stability for the initial-boundary value problem \eqref{eq:FullBz:Eq:psi}, Theorem \ref{thm:FullBz:Main:2}, one can follow the proof of Theorem 14 in \cite{Kuo-Liu-Tsai-2}. The idea is basically the same and we omit the details here.

\end{subsection}
\end{section}

%% For one-column wide figures use
%\begin{figure}
%% Use the relevant command to insert your figure file.
%% For example, with the graphicx package use
%  \includegraphics{example.eps}
%\end{figure}
%%
%% For two-column wide figures use
%\begin{figure*}
%% Use the relevant command to insert your figure file.
%% For example, with the graphicx package use
%  \includegraphics[width=0.75\textwidth]{example.eps}
%\end{figure*}
%\begin{table}
%\begin{tabular}{lll}
%\hline\noalign{\smallskip}
%\end{tabular}
%\end{table}

\begin{section}*{Acknowledgements}
Part of this work was written during the stay at Department of Mathematics, Stanford University. The author would like to thank Professor Tai-Ping Liu for his kind hospitality. This work was supported by NSC grant 102-2115-M-006-018-MY2.
\end{section}



\end{document}